\documentclass[11pt]{article}
\usepackage{fullpage}

\usepackage{epigraph}
\usepackage{nicefrac}
\usepackage[margin=1in]{geometry}
\usepackage{relsize}

\usepackage{dsfont}

\usepackage{latexsym}

\usepackage[dvipsnames,usenames]{xcolor}
\usepackage{parskip}
\usepackage{color}
\usepackage{float}
\usepackage{bbm}
\floatstyle{boxed}
\newfloat{algorithm}{t}{lop}
\usepackage{tikz}
\usepackage{varwidth}

\usepackage{float}
\usepackage[titletoc,title]{appendix}
\usepackage[classfont=sanserif,langfont=roman,funcfont=italic]{complexity}
\usepackage{caption}
\usepackage{comment}
\usepackage{enumerate}
\usepackage{amsmath, amsthm, amssymb, amstext,  graphicx,amsopn} 
\usepackage{subfigure}

\usepackage{booktabs} 
\DeclareMathOperator{\Ex}{\mathbb{E}}

\usepackage{epigraph}
\usepackage{nicefrac}
\usepackage{hyperref}
\hypersetup{colorlinks=true,linkcolor=black,citecolor=blue}

\usepackage{mathtools,extarrows}
\usepackage{url}
\newtheorem{theorem}{Theorem}
\newtheorem{definition}{Definition}

\newtheorem{observation}{Observation}

\newtheorem{lemma}{Lemma}

\newtheorem{example}{Example}

\allowdisplaybreaks

\allowdisplaybreaks

\renewcommand{\vec}[1]{\mathbf{#1}}

\usepackage{thmtools} 
\usepackage{thm-restate}
\usepackage[capitalise,nameinlink]{cleveref}

\newtheorem{mytheorem}[theorem]{Theorem}
\newtheorem{mylemma}[lemma]{Lemma}

\newtheorem{mydefinition}[definition]{Definition}

\crefname{figure}{Figure}{Figure}

\usepackage{mathrsfs}
\hypersetup{colorlinks=true,linkcolor=blue,citecolor=blue}
\definecolor{mygreen}{rgb}{0.0, 0.55, 0.0}
\definecolor{blue-violet}{rgb}{0.54, 0.17, 0.89}

\usepackage{tikz}
\usetikzlibrary{calc, graphs, graphs.standard, shapes, arrows, arrows.meta, positioning, decorations.pathreplacing, decorations.markings, decorations.pathmorphing, fit, matrix, patterns, shapes.misc, tikzmark}

\usepackage{amsfonts}

\newcommand{\eps}{\varepsilon}

\newcommand{\N}{\mathbb{N}}
\renewcommand{\R}{\mathbb{R}}
\newcommand{\cA}{\mathcal{A}}
\newcommand{\cB}{\mathcal{B}}
\newcommand{\cC}{\mathcal{C}}

\newcommand{\cO}{{O}}
\newcommand{\cH}{\mathcal{H}}
\renewcommand{\cP}{\mathcal{P}}
\newcommand{\cX}{\mathcal{X}}
\newcommand{\cW}{\mathcal{W}}

\newcommand{\cZ}{\mathcal{Z}}

\newcommand{\argmin}{\mathrm{argmin}}
\newcommand{\X}{\vec{x}}
\newcommand{\Y}{\vec{y}}
\newcommand{\Z}{\vec{z}}
\newcommand{\numPaths}[1]{\mathfrak{q}_{#1}}
\newcommand{\multiplicity}{\mu}

\newcommand{\explain}[1]{\tag{\textcolor{gray}{#1}}}

\title{The Sharp Power Law of Local Search on Expanders\footnote{Alphabetical author ordering. This work was done in part while the authors were visiting the Simons Institute for the Theory of Computing. This research/project is supported by the National Research Foundation, Singapore under its AI Singapore Programme (AISG Award No: AISG-PhD/2021-08-013). This research is  supported  by the US National Science Foundation under grant CCF-2238372.}}
\author{Simina Br\^anzei\footnote{Purdue University. E-mail: simina.branzei@gmail.com}
\and 
Davin Choo\footnote{National University of Singapore. E-mail: davin@u.nus.edu.}
\and 
Nicholas Recker\footnote{Purdue University. E-mail: nrecker@purdue.edu}
}
\date{}

\begin{document}

	\maketitle

\begin{abstract}


Local search  is a powerful heuristic in optimization and computer science, the  complexity of which has been studied in  the white box and black box models. In the black box model, we are given  a graph $G = (V,E)$ and oracle access to a function $f : V \to \mathbb{R}$. The local search problem is to find a vertex $v$ that is a local minimum, i.e. with $f(v) \leq f(u)$ for all $(u,v) \in E$, using as few queries to the oracle as possible. The query complexity is well   understood on the grid and the   hypercube, but much less is known beyond. 

\medskip 

We show  the query complexity  of local search on $d$-regular expanders with constant degree is $\Omega\left(\frac{\sqrt{n}}{\log{n}}\right)$, where $n$ is the number of vertices of the graph. This  matches within a logarithmic factor the upper bound of $\cO(\sqrt{n})$  for constant degree graphs from \cite{aldous1983minimization}, implying that steepest descent with a warm start is essentially an optimal algorithm for expanders. The best  lower bound known from prior literature was $\Omega\left(\frac{\sqrt[8]{n}}{\log{n}}\right)$,  shown by \cite{santha2004quantum} for  quantum and randomized algorithms.

\medskip

We obtain this result by considering a broader framework of graph features such as vertex congestion and separation number. 
We show that for each graph, the randomized query complexity of local search is  $\Omega\left(\frac{n^{1.5}}{g}\right)$, where $g$ is the vertex congestion of the graph; and 
 $\Omega\left(\sqrt[4]{\frac{s}{\Delta}}\right)$, where $s$ is the separation number and $\Delta$ is the maximum degree. For separation number the previous  bound was 
  $\Omega\left(\sqrt[8]{\frac{s}{\Delta}} /\log{n}\right)$, given by  \cite{santha2004quantum} for  {quantum} and randomized  algorithms.

We also show a variant of the relational adversary method from \cite{Aaronson06}. Our variant is asymptotically at least as strong as the version in \cite{Aaronson06} for all randomized algorithms,  as well as  strictly stronger on some problems and easier to apply in our setting.

\medskip

%
\end{abstract}


\newpage 

\tableofcontents

 \newpage 

\section{Introduction}
\label{sec:intro}

%
Local search is a powerful heuristic for solving hard optimization problems, which works by  starting with an initial solution to a problem and then iteratively improving  it. 
 Its simplicity and ability to handle large and complex search spaces make it a useful tool in a wide range of fields, including computer science,  engineering, optimization,  economics, and finance. 

The complexity of local search has been extensively studied in both  the white box  model (see, e.g., \cite{DBLP:journals/jcss/JohnsonPY88}) and the black box model (see, e.g., \cite{aldous1983minimization}). The latter type of complexity, also known as  query complexity, is well   understood when the neighbourhood structure of the underlying graph is the  Boolean  hypercube or the $d$-dimensional grid, but much less is known for general graphs.

Many optimization techniques rely on gradient-based methods. The speed at which gradient methods find a stationary point of a function can be estimated by analyzing the complexity of local search on the corresponding discretized space.  Constructions for analyzing the hardness of computing stationary points are often similar to those for local search, modulo handling the smoothness of the function (see, e.g., \cite{vavasis1993black}). 
Meanwhile, the difficulty of local search itself is strongly related to the neighbourhood structure of the  underlying graph.
At one extreme, local search on a line graph on $n$ nodes is easy and can be solved  via binary search in $\cO(\log{n})$ queries.
At the other extreme, local search on a clique on $n$ nodes takes $\Omega(n)$ queries, thus requiring brute force.

In this paper, we consider the following high level question: 
\emph{How does the geometry of the graph influence the complexity of local search?}
In general, the neighbourhood graph search structure in optimization settings may correspond to more general graphs beyond the well-studied Boolean hypercubes and $d$-dimensional grids.
For example, when the data in low rank matrix estimation is subjected to adversarial corruptions, it is helpful to consider the function on a Riemannian manifold rather than Euclidean space.
That is, the discretization of an optimization search space may not necessarily always correspond to some $d$-dimensional grid.
Multiple works consider optimization in non-Euclidean spaces, such as that of \cite{bonnabel2013stochastic}, which adapts stochastic gradient descent to work on Riemannian manifolds.
See \cite{amari_book} and \cite{Boumal_manifolds} for more discussion.

Our paper tackles the challenge of understanding local search on general graphs and obtains several new results by considering a broader framework of graph features such as vertex congestion and separation number. A corollary is a  lower bound of the right order for expanders with constant degree.

Our methodology is strongly inspired by, and can be seen as a variant of, the relational adversary method of \cite{Aaronson06}.
However, where Aaronson's method focuses on the contribution of a query towards distinguishing two inputs from all others, our method considers the aggregate impact of a query across many inputs at once.
This allows our method to be asymptotically at least as strong as the version in \cite{Aaronson06} for all randomized algorithms,  as well as strictly stronger on some problems and easier to apply in our setting.
This strength also comes at a cost: we get results for randomized algorithms, whilst Aaronson's method works in  quantum settings.

\paragraph{Roadmap to the paper.}

The model is defined in \cref{sec:model}. Our contributions are given in \cref{sec:contributions}.
Related work is discussed in \cref{sec:related-work}.
Our variant of the relational adversary method is stated together with an example and discussion in \cref{sec:relational-adv-method}  (with full material in \cref{sec-appendix-relational-adversary}).

Lower bounds via vertex congestion, as well as the corollary for expanders, are given in \cref{sec:application-congestion} (with full material in \cref{sec:appendix-congestion}). 
Lower bounds via the  separation number  
 are  in \cref{sec:application-separation-number} (with full material in \cref{sec:appendix-separation-number}).

Finally, 
\cref{sec:appendix-expansion-and-congestion} reviews some known results from prior work that we  use.
 \cref{sec:appendix-valid-functions-have-unique-local-minimum} provides the proof for a   lemma used throughout the paper.

\section{Model} \label{sec:model}
Let $G = (V,E)$ be a connected undirected  graph and $f : V \to \mathbb{R}$ a function defined on the vertices.  
A vertex $v \in V$  is a local minimum if $f(v) \leq f(u)$ for all $\{u,v\} \in E$.
We will write $V = [n] = \{1, \ldots, n\}$.

Given as input a graph $G$ and oracle access to function $f$, the local search problem is to find a local minimum of $f$ on $G$ using as few queries as possible. {Each query is of the form: ``Given a vertex $v$, what is $f(v)$?''}.

\paragraph{Query complexity.} The \emph{deterministic query complexity} of a task is the total number of queries necessary and sufficient for a correct deterministic algorithm to find a solution.
The \emph{randomized query complexity} is the expected number of queries required to find a solution with probability at least $9/10$ for each input, where the expectation is taken over the coin tosses of the protocol.

\paragraph{Congestion.}
Let $\cP = \{P^{u,v}\}_{u,v \in V}$  be an all-pairs set of paths in $G$, where $P^{u,v}$ is a path from $u$ to $v$.  {For convenience, we assume $P^{u,u} = (u)$  for all $u \in V$; our results will  hold even if  $P^{u,u} = \emptyset$.}

For a path $Q = (v_1, \ldots, v_s)$ in $G$, let $c^Q_v$ be the number of times a vertex $v \in V$ appears in $Q$ and $c^Q_e$ the number of times an edge $e \in E$ appears in $Q$.
The \emph{vertex congestion} of the set of paths $\cP$ is $\max_{v \in V} \sum_{Q \in \cP} c_v^Q$, while the  \emph{edge congestion} of $\cP$ is $\max_{e \in E} \sum_{Q \in \cP} c_e^Q$.

The \emph{vertex congestion of $G$} is the smallest integer $g$ for which the graph has an all-pairs set of paths $\mathcal{P}$ with vertex  congestion $g$. Clearly, $g \geq n$ since each vertex belongs to at least $n$ paths in $\mathcal{P}$ and $g \leq n^2$ since each vertex appears at most once on each path and there are $n^2$ paths in $\mathcal{P}$. 
 The \emph{edge congestion} $g_e$ is similarly defined, but with respect to the edge congestion of a  set of paths $\mathcal{P}$.

\paragraph{Separation number.}
For each subset of vertices $A \subseteq V$, let $\delta(A) \subseteq V \setminus A$ be the set of vertices outside $A$ and adjacent to vertices in $A$.
The \emph{separation number} $s$ of $G$  (see, e.g., \cite{graph_notions_survey}) is  \footnote{For example, the separation number of a barbell graph (i.e., two cliques of size $n/2$ connected by a single edge) is $n/8$, since the maximization will choose $H$ to be one clique of size $n/2 = 4n/8$ and the minimization will choose $A$ to be an arbitrary subset of $H$ of size $3n/8$; then $\delta(A)$ is the rest of the clique of size $4n/8 - 3n/8 = n/8$.}:
\[
s = \max_{H \subseteq V} \min_{\substack{A \subseteq H :\\ |H|/4 \;\leq\; |A| \;\leq\; 3|H|/4}} |\delta(A)| \;.
\]

\paragraph{$d$-regular expanders.}
For each set of vertices $S \subseteq V$, the edges with one endpoint in $S$ and another in $V \setminus S$ are called \emph{cut edges} and denoted $E(S, V \setminus S) = \{ (u,v) \in E \mid  u \in S, v \not\in S \}$. 
The graph is a $\beta$-expander if $|E(S, V \setminus S)| \geq \beta \cdot |S|$, for all $S \subseteq V$ with $0 < |S| \leq n/2$ (see, e.g. \cite{alon04}).
The graph is \emph{$d$-regular} if each vertex has degree $d$.

\paragraph{Distance.} For each $u,v \in V$, let $dist(u,v)$ be the length of the shortest path from $u$ to $v$.

\section{Our contributions}
\label{sec:contributions}


Guided by the high level question of understanding how graph geometry influences hardness of local search, we obtain the following results.
%


\subsection{Our variant of the relational adversary method} Our first contribution is to design a new variant of the relational adversary method of \cite{Aaronson06}.
While \cite{Aaronson06} relates the query complexity to the progress made on pairs of inputs, we relate the query complexity to  progress made on {subsets of inputs} via a different expression.

The precise statement of our variant is \cref{thm:our-variant} in \cref{sec:relational-adv-method}. 
In \cref{sec:relational-adv-method}, we also show an example on which our variant is strictly stronger,  giving a tight lower bound for the query complexity of a simple ``matrix game''. Then we prove our variant is asymptotically at least as strong in general, for randomized algorithms.


\subsection{Lower bounds for local search via congestion} 
Next we  give the first known lower bound for local search as a function of  (vertex) congestion, which is enabled by our \cref{thm:our-variant}.


\begin{restatable}{mytheorem}{lowcongestionimplieslocalsearchhard}
\label{thm:low_congestion_implies_local_search_hard}
Let $G = (V, E)$ be a connected undirected graph with $n$ vertices.
Then the randomized query complexity of local search on $G$ is $\Omega\left(\frac{n^{1.5}}{g}\right)$, where  $g$ is the vertex congestion of the graph.
\end{restatable}

Since $g \in [n, n^2]$, \cref{thm:low_congestion_implies_local_search_hard} {cannot} be used to show a lower bound stronger than $\Omega(\sqrt{n})$ queries, matching a general upper bound  of $\cO(\sqrt{n})$ for graphs with bounded degree (\cite{aldous1983minimization}).
\cref{thm:low_congestion_implies_local_search_hard} gives meaningful results precisely when one can construct an all-pairs set of paths with vertex congestion $g = o(n^{1.5})$; e.g.\ our bound is vacuous on trees since $g \in \Theta(n^2)$ on trees.

\cref{thm:low_congestion_implies_local_search_hard} also implies a lower bound of $\Omega\left(\frac{n^{1.5}}{g_e \cdot \Delta} \right)$ on any graph $G$, where $g_e$ is the edge congestion and $\Delta$ the maximum degree of $G$.

\paragraph{{High level approach.}} To obtain the result in \cref{thm:low_congestion_implies_local_search_hard}, we apply Yao's lemma and design a hard input distribution where the input is  a random function $f : V \to \mathbb{R}$ induced by a ``staircase'': the value of the function  $f$ outside the staircase is equal to the distance to the entrance of the staircase, while the value at vertices  on the staircase  decreases as one walks away from the entrance. The local minimum is unique and can be found at the end of the staircase. This type of construction is classical (see, e.g., \cite{aldous1983minimization,vavasis1993black,Aaronson06}). We designate the space of such functions as $\mathcal{X}$.

A staircase is characterized by a sequence of vertices (milestones). Then we connect each pair $w,z$ of consecutive milestones with the path $P^{w,z}$ from the all-pairs set of paths $\mathcal{P} = \{ P^{u,v}\}_{u,v \in [n]}$ with ``low'' congestion (that is, with congestion as low as possible given the graph).


Then, we design a function $r$ --- also called the relation --- that ``relates'' any two such functions.
Any such non-zero relation induces a distribution over functions where each function $F_1$ is sampled with likelihood equal to $\sum_{F_2 \in \mathcal{X}} r(F_1,F_2)$.

Our choice of $r$ increases according to the number of milestones the underlying staircases have in common (specifically, by the longest initial prefix of milestones shared by the two staircases).
With our choice of $r$, the distribution over staircases is as though the milestones were chosen uniformly at random without replacement.

Any two functions $F_1$ and $F_2$ with long initial prefix in their corresponding staircases are very similar and so will be hard to distinguish by an algorithm.
Roughly speaking, without querying sufficiently many vertices and chancing upon a vertex which $F_1$ and $F_2$ disagrees on, the algorithm will not be able to distinguish them.
In order to capture this difficulty of distinguishing  such $F_1$ and  $F_2$, the relation $r(F_1,F_2)$ will be set to a high value.
Our congestion lower bound then follows by invoking  our variant of the relational adversary (\cref{thm:our-variant}) with this choice of $r$. 

Our  approach was  inspired by previous works such as   \cite{Aaronson06}, who gave lower bounds for the $d$-dimensional grid and the Boolean hypercube, and of  \cite{dinh2010quantum}, who gave lower bounds for Cayley and vertex transitive graphs by using a  system of \emph{carefully chosen} shortest paths rather than arbitrary shortest paths; this inspired our choice of the set of paths $\mathcal{P}$. We explain in more detail the comparison with \cite{Aaronson06} and \cite{dinh2010quantum} in Section  \ref{sec:comparison} of Related work.

\subsection{Lower bounds for local search via separation number} We also give an \emph{improved} lower bound for local search with respect to the graph separation number $s$.
Our construction is heavily inspired by the one in \cite{santha2004quantum} , who gave a lower bound of  $\Omega\left(\sqrt[8]{\frac{s}{\Delta}}/\log n\right)$, for both the quantum and classical randomized algorithms.
Adapting this construction within the framework of \cref{thm:our-variant} is non-trivial however. 




\begin{restatable}{mytheorem}{theoremseparationnumberlowerbound}
\label{thm:separation-number-result}
Let $G = (V, E)$ be a connected undirected graph with $n$ vertices, maximum degree $\Delta$, and separation number $s$.  
Then the randomized query complexity of local search on $G$ is
$\Omega\left(\sqrt[4]{\frac{s}{\Delta}}\right)$.
\end{restatable}

The best known upper bound with respect to graph separation number is $\cO((s + \Delta) \cdot \log n)$ due to \cite{santha2004quantum}, which was obtained via a refinement of the divide-and-conquer procedure of \cite{llewellyn1989local}.
It is an interesting open question whether the current upper and lower bounds can be improved.

 
 
 
\subsection{Corollaries for expanders, Cayley graphs, and the hypercube} Since $d$-regular $\beta$-expanders with constant $d$ and $\beta$ admit an all-pairs set of paths with congestion $\cO(n \cdot \log{n})$ (e.g., see \cite{BFU99}), we  get the next lower bound for constant degree expanders.

%

\begin{restatable}{mycorollary}{corollaryExpandersLowerBound}
\label{cor:expanders_lower_bound}
Let $G = (V, E)$ be an undirected $d$-regular $\beta$-expander with $n$ vertices, where $d$ and $\beta$ are constant.
Then the randomized query complexity of local search on $G$ is $\Omega\left(\frac{\sqrt{n}}{\log n}\right)$.
\end{restatable}

The lower bound of \cref{cor:expanders_lower_bound} is tight within a logarithmic factor. A simple algorithm known as steepest descent with warm start (\cite{aldous1983minimization}) can be used to see this: 
\begin{quote}
    First query $t$ vertices $x_1, \ldots, x_t$ selected uniformly at random and pick the vertex $x^*$ that minimizes the function among these\footnote{That is, the vertex $x^*$ is defined as: $x^* = x_j$, where $j = \argmin_{i=1}^t f(x_i)$.}. Then run steepest descent from $x^*$ and stop when no further improvement can be made, returning the final vertex reached. When $t = \sqrt{n\Delta}$, where $\Delta$ is the maximum degree of the graph, the algorithm issues $\cO(\sqrt{n\Delta})$ queries in expectation.
\end{quote}
Thus steepest descent with a warm start has expected query complexity $\cO(\sqrt{n})$ on constant degree expanders. Our lower bound implies this algorithm is essentially optimal  on such graphs.

We also get a  lower bound  as a function of the expansion and maximum degree of the graph $G$.

\begin{restatable}{mycorollary}{corollaryViaExpansionLowerBound}
\label{cor:expansion_parameter_lower_bound}
Let $G = (V, E)$ be an undirected $\beta$-expander with $n$ vertices and maximum degree $\Delta$.
Then the randomized query complexity of local search on $G$ is $\Omega\left(\frac{\beta \sqrt{n}}{\Delta   \log^2{n}}\right)$.
\end{restatable}



The congestion framework also allows us to recover a lower bound for undirected Cayley graphs, which were studied before in \cite{dinh2010quantum}. 
An \emph{undirected Cayley} graph is formed from a group $\mathcal{G}$ and a generating set $S$.
The vertices are the elements of $\mathcal{G}$ and there is an edge between vertices $u,v$ if $u = w \cdot v$ or $v = w 
\cdot u$ for some $w \in S$.

\begin{restatable}{mycorollary}{corollaryCayleyLowerBound}
\label{cor:cayley_lower_bound}
Let $G = (V,E)$ be an undirected Cayley graph with $n$ vertices and diameter $d$.
Then the randomized query complexity of local search on $G$ is $\Omega\left(\frac{\sqrt{n}}{d}\right)$.
\end{restatable}

We also get the next corollary for the query complexity of local search on the Boolean hypercube.

\begin{restatable}{mycorollary}{corollaryHypercubeLowerBound} \label{cor:corollaryHypercubeLowerBound}
The randomized query complexity of local search on the Boolean hypercube $\{0,1\}^n$  is $\Omega\left(\frac{2^{n/2}}{n}\right)$.
\end{restatable}

The lower bound in \cref{cor:corollaryHypercubeLowerBound}   is sandwiched between the lower bound of   $\Omega\left(\frac{2^{n/2}}{\sqrt{n}}\right)$ by \cite{zhang2009tight} and the lower bound of $\Omega\left(\frac{2^{n/2}}{{n}^2}\right)$ by \cite{Aaronson06}.

\section{Related work}
\label{sec:related-work}


The query complexity of local search was first studied experimentally  by \cite{tovey81}.
The first breakthrough in the theoretical analysis of local search was obtained by \cite{aldous1983minimization}.
Aldous stated the algorithm based on  steepest descent with a warm start and showed the first nontrivial lower bound of $\Omega(2^{n/2-o(n)})$ on the query complexity for the Boolean hypercube $\{0,1\}^n$.

 The lower bound construction from \cite{aldous1983minimization} uses Yao's lemma and describes a hard distribution, such that if a deterministic algorithm receives a random function according to this distribution, the expected number of queries 
 issued until finding a local minimum will be large. The random function is obtained as follows:
 
 \begin{quote} Consider an initial vertex $v_0$ uniformly at random. Set the function value at $v_0$ to $f(v_0) = 0$. From this vertex, start an unbiased random walk $v_0, v_1, \ldots$ For each vertex $v$ in the graph, set $f(v)$ equal to the first hitting time of the walk at $v$; that is, let $f(v) = \min\{t \mid v_t=v\}$.
 \end{quote}
 The function $f$ defined this way has a unique local minimum at $v_0$.
 By a very delicate  analysis of this distribution,  \cite{aldous1983minimization} showed a lower bound of $\Omega(2^{n/2-o(n)})$ on the hypercube $\{0,1\}^n$.
 
 This almost matches the query complexity of steepest descent with a warm start, which was also analyzed in \cite{aldous1983minimization} and shown to take $\cO(\sqrt{n} \cdot 2^{n/2})$ queries in expectation on the hypercube.
 The steepest descent with a warm start algorithm applies to generic graphs too,
 resulting in $\cO(\sqrt{n \cdot \Delta})$ queries overall for any graph with maximum degree $\Delta$.



Aldous' lower bound for the hypercube was later improved via more refined types of random walks and/or more careful analysis.
\cite{Aaronson06} improved the bound to $\Omega(2^{n/2}/n^2)$ via his relational adversary method, which is a combinatorial framework that avoids analyzing the posterior distribution and also yielded a quantum bound of $\Omega(2^{n/4}/n)$.
\cite{zhang2009tight}  improved the randomized lower bound to a tight bound of $\Theta(2^{n/2} \cdot \sqrt{n})$ via a ``clock''-based random walk construction, which avoids self-intersections.

Meanwhile, \cite{llewellyn1989local} developed a deterministic divide-and-conquer approach to solving local search that is theoretically optimal over all graphs in the deterministic context, albeit hard to apply in practice.
On the hypercube, their method yields a lower bound of $\Omega(2^n/\sqrt{n})$ and an upper bound of $\cO(2^n \log(n)/\sqrt{n})$: a mere $\log(n)$ factor gap.

Another commonly studied graph for local search is the $d$-dimensional grid $[n]^d$.
\cite{Aaronson06} used his relational adversary method there to show a randomized lower bound of $\Omega(n^{d/2-1} / \log n)$ for every constant $d \ge 3$. 
\cite{zhang2009tight} proved a randomized lower bound of $\Omega(n^{d/2})$ for every constant $d \ge 4$; this is tight as shown by Aldous' generic upper bound.
Zhang also showed improved bounds of $\Omega(n^{2/3})$ and $\Omega(n^{3/2} / \sqrt{\log n})$ for $d=2$ and $d=3$ respectively, as well as some quantum results.
The work of \cite{sun2009quantum} closed further gaps in the quantum setting as well as the randomized $d=2$ case.
The problem of local search on the grid was also studied under the context of multiple limited rounds of adaptive interactions by \cite{branzei2022query}.

More general results are few and far between.
On many graphs, the simple bound from \cite{aldous1983minimization} of $\Omega(\Delta)$ queries is the best known lower bound: hiding the local minimum in one of the $\Delta$ leaves of a star subgraph requires checking about half the leaves in expectation to find it.

\cite{santha2004quantum} gave a quantum lower bound of $\Omega\left( \sqrt[8]{\frac{s}{\Delta}} / \log(n) \right)$, where $s$ is the separation number of the graph. This  implies the same lower bound in a randomized context, using the spectral method.
Meanwhile, the best known upper bound is $\cO((s + \Delta) \cdot \log n)$ due to \cite{santha2004quantum}, which was obtained via a refinement of the divide-and-conquer procedure of \cite{llewellyn1989local}. 

\cite{dinh2010quantum} studied Cayley and vertex transitive graphs and gave lower bounds for local search as a function of the number of vertices and the diameter of the graph. We explain the comparison with their work more precisely in Section \ref{sec:comparison}.
\cite{Verhoeven06} obtained upper bounds as a function of the genus of the graph.

\cite{babichenko2019communication} studied the communication complexity of local search. This captures  distributed settings, where data is stored in the cloud, on different computers.

There is a rich literature analyzing the congestion of graphs. E.g., the notion of edge congestion is important in routing problems, where systems of paths with low edge congestion can enable traffic with minimum delays (see, e.g., \cite{BFU99,leighton1998multicommodity,chuzhoy2016routing}).
This problem is sometimes called \emph{multicommodity flow} or \emph{edge disjoint paths with congestion}.
Others study routing with the goal of maximizing the number of demand pairs routed using \emph{node} disjoint paths; this is the same as requiring vertex congestion equal to $1$ (see, e.g., \cite{chuzhoy2017new,chuzhoy2018almost}).







Local search is strongly related to the problem of local optimization where one is interested in finding an approximate local minimum of a function on $\mathbb{R}^d$.
A common way to solve local optimization problems is to employ gradient-based methods, which find approximate stationary points. 
To show lower bounds for finding stationary points, one can similarly  define a function that selects a walk in the underlying space and hide a stationary point at the end of the walk.  Handling the requirement that the function is smooth and ensuring there is a unique stationary point are additional challenges. 

For examples of works on algorithms and complexity of computing approximate stationary points, see, e.g., \cite{vavasis1993black,pmlr-v119-zhang20p,  stationary_I,stationary_II,bubeck2020trap,pmlr-v119-drori20a}). Constructions where the function is induced by a  hidden walk have first been designed for showing lower bounds on the query complexity of finding Brouwer fixed points in classical work by \cite{hirsch1989exponential}.


Works like \cite{bonnabel2013stochastic} study stochastic gradient descent, which is one method of finding approximate local minima.
Moreover, they do this on Riemann manifolds, which are a very broad class of spaces.
This motivates the need to study local search not only on hypercubes and grids, but also on other, broader classes of graphs. For a more extensive survey, see, e.g., \cite{amari_book}.




The computational complexity of local search is captured by the class PLS, defined by \cite{DBLP:journals/jcss/JohnsonPY88} to model the difficulty of finding locally optimal solutions to optimization problems.
A related class is PPAD, introduced by \cite{Papadimitriou_1994} to study the computational complexity of finding a Brouwer fixed-point.
Both PLS and PPAD are subsets of of the class TFNP.

The class PPAD contains many natural problems that are computationally equivalent to the problem of finding a Brouwer fixed point (\cite{CD09}), such as finding an approximate Nash equilibrium in a multiplayer or two-player game (\cite{daskalakis2009complexity,chen2009settling}), an Arrow-Debreu equilibrium in a market (\cite{vazirani2011market,chen2017complexity}), and a local min-max point (\cite{costis_minmax}).

The query complexity of computing an $\eps$-approximate Brouwer fixed point was studied in a series of papers starting with \cite{hirsch1989exponential}, later improved by \cite{chen2005algorithms} and \cite{chen2007paths}.
Recently, \cite{fearnley2022complexity} showed that the class CLS, introduced by \cite{daskalakis2011continuous} to capture continuous local search, is equal to PPAD $\cap$ PLS.
The query complexity of continuous local search has also been studied (see, e.g., \cite{hubavcek2017hardness}).

\subsection{Comparison with prior works and corrections} \label{sec:comparison}
Our approach for the lower bound as a function of congestion was  directly inspired by the relational adversary method of \cite{Aaronson06} and an ingenious   application to vertex-transitive graphs by \cite{dinh2010quantum}.  In both of these papers, a   hard distribution  is  obtained by getting a random input function induced by a ``staircase'' (walk): the value of the function outside the staircase is equal to the distance to the entrance of the staircase, while the value of the function on the staircase is decreasing as one moves away from the entrance. 

\cite{dinh2010quantum} choose a staircase by first selecting several random points (milestones) and then connecting them with a path from a system of paths.
While one typically chooses \emph{arbitrary} shortest paths between two endpoints when constructing lower bounds, the system of paths in their work  consists of  \emph{carefully chosen} shortest paths as follows:
$(i)$ For paths that start from a fixed vertex $v_0$, fix arbitrary shortest paths;
$(ii)$ For paths starting from other vertices $v_i$ with $i > 0$, use one of the same paths as from $v_0$, but transformed by an automorphism mapping $v_0$ to $v_i$ (which is  defined when the graph is vertex transitive). 

The high-level approach and the careful selection of paths  in \cite{dinh2010quantum} inspired our choice for the set of paths $\mathcal{P}$ with low congestion when the graph is not necessarily vertex transitive.

Given the family of functions and staircases described, what remains is to define a relation between functions and compute the lower bound obtained by invoking the relational adversary method.

\paragraph{Corrections.} 
There appears to be a potential issue in the proof of Proposition 2 in \cite{dinh2010quantum}, where the probability of an event is higher than would be required for the argument to go through \footnote{
In \cite{dinh2010quantum}, for any vertex $v$ that appears in a walk $\X = (x_0, x_1, x_2, \ldots, x_L)$, the function is defined as $f_{\X}(v) = L - \max\{i : v = x_i\}$, i.e.\ $L$ minus the last index $v$ last appears on $\X$.
Then, a vertex $v$ should be a disagreement between two walks $\X$ and $\Y$ if $v$ lies on both walks but $f_{\X}(v) \neq f_{\Y}(v)$.
Afterwards, in Proposition 2,  if vertex $v$ is a disagreement between two walks $\X$ and $\Y$, then $v$ must lie on at least one of the walks, and moreover, $v$ was stated to be contained in a ``segment'' of the tail that is not the one immediately after the divergence place of  $\X$ and $\Y$  (quote: ``\emph{We can't have both $t < j + s$ and $t' < j + s$... either $t \geq j + s$ or $t' \geq j + s$}'', where $t = \max\{i : v = x_i\}$ is the last index of $v$ on $\X$, $t' = \max\{i : v = y_i\}$ is the last index of $v$ on $\Y$, $j$ is the point of splitting, and $s$ is the ``segment length'').
However, one can check  the following setup which violates this assertion also makes $v$ a disagreement: $\max\{i : v = x_i\} < j < \max\{i : v = y_i\} < j + s$, where $v$ appears in the shared prefix and then only on $\X$ within the first segment after they diverge.}. We bypassed this  using our  \cref{thm:our-variant}, which enabled us to recover  the randomized lower bound for Cayley graphs from \cite{dinh2010quantum}; see \cref{cor:cayley_lower_bound}.

This  also occurs in the proof \footnote{In \cite{Aaronson06}, Lemma 6.2: in the first case, where $t > j - n$ and $t^* > j - n$,
the probability $P[x_t = y_{t^*}]$ is not zero, but rather can be nearly $1/2$ (e.g. at $t = t^* = j$),  since the coordinate loop allows staying in place. With probability $1/2$, exactly one of $x_{j-1}$ and $y_{j-1}$ is equal to $x_j$, which usually makes $f_X(x_j) \neq f_Y(x_j)$.
We have reached out to the author and he has acknowledged the error, as well as suggesting another possible fix.
} of Lemma 6.2 of \cite{Aaronson06}, where it could be corrected  by setting the function $f_X(v)$, for a vertex $v$, as $
f_X(v) = \min\{t : x_t = v \text{ and } x_{t+1} \neq v\}\,.$

\section{A variant of the relational adversary method}
\label{sec:relational-adv-method}

In this section we state our variant  (\cref{thm:our-variant}).
After stating the variant, we also design and analyze a ``matrix game'', to illustrate a simple problem for which our variant yields a better (in fact tight)  lower bound for the randomized query complexity.  The complete details are included  in \cref{sec-appendix-relational-adversary}, together with the original theorem from \cite{Aaronson06} for comparison. 

\begin{restatable}{mytheorem}{ourvariant}
\label{thm:our-variant}
Consider finite sets $A$ and $B$, a set $\cX \subseteq B^{A}$ of functions \footnote{Each function $F \in \cX$ has the form $F: A \to B$.}, and a map $\cH : \cX \to \{0,1\}$ which assigns a label to each function in $\cX$.
Additionally, we get oracle access to an unknown function $F^* \in \cX$.
The problem is to compute $\cH(F^*)$ using as few queries to $F^*$ as possible.\footnote{In other words, we have free access to $\cH$ and the only queries  counted are the ones to $F^*$, which will be of the form: ``What is $F^*(a)$?'', for some $a \in A$. The oracle will return $F^*(a)$ in one computational step.}

Let $r : \cX \times \cX \to \R_{\geq 0}$ be a non-zero symmetric function of our choice with $r(F_1, F_2) = 0$ whenever $\cH(F_1) = \cH(F_2)$.
For each $\cZ \subseteq \cX$, define  
\begin{align} 
M(\cZ) = \sum_{F_1 \in \cZ} \sum_{F_2 \in \cX} r(F_1, F_2) \; ; \; \; \; \text{ and } \; \; q(\cZ) = \max_{a \in A}{\sum_{F_1 \in \cZ} \sum_{F_2 \in \cZ} r(F_1,F_2) \cdot \mathbbm{1}_{\{F_1(a) \neq F_2(a)\}}}\,.
\end{align}

If there exists a subset $\cZ \subseteq \cX$ with $q(\cZ) > 0$, 
then the randomized query complexity of the problem is
at least
\begin{equation}
\label{eq:relational_variant}
\min\limits_{\substack{\cZ \subseteq \cX: q(\cZ) > 0}}
 \frac{M(\cZ)}{100 \cdot q(\cZ)}  \,.
\end{equation}
\end{restatable}

\cref{thm:our-variant} uses Yao's lemma (see \cref{sec:appendix-expansion-and-congestion}), thus the algorithm can be assumed to be deterministic and receive as input a random function sampled from some distribution $P$. 

The theorem considers the probability distribution $P$ where each function $F \in \cX$ is given as input with probability $P(F) = \frac{M(\{F\})}{M(\cX)}$, where $\mathcal{X}$ is the space of possible functions.

The term $M(\cZ)$ represents the likelihood that the function $F$ given as input comes from the set $\cZ \subseteq \cX$, while $q(\cZ)$ is proportional to a lower bound on the number of queries needed in the worst case to narrow down the function $F$ within $\cZ$, conditioned on $F$ being in $\cZ$.

Clearly some choices of $r$ are more useful than others, and so the  challenge when giving lower bounds is to design the function $r$ and estimate the expression in  equation \eqref{eq:relational_variant}.

\subsection{Matrix game}
\label{sec:matrix-game}

In this section, we describe a toy problem upon for which our new variant (\cref{thm:our-variant}) can show  a stronger lower bound than the original relational adversary theorem.

\paragraph{Setup}
Let $n \in \N$ be a perfect square  and $\cX$ be a subset of square $\sqrt{n} \times \sqrt{n}$ matrices with entries from $\{0,1,2\}$.
There are two types of matrices within $\cX$: ``row'' matrices and ``column'' matrices.
Row matrices have one row of $1$s with all other entries $0$ while column matrices have one column of $2$s with all other entries $0$.
For example:
\[
    \begin{bmatrix}
        0&0&0&0\\
        1&1&1&1\\
        0&0&0&0\\
        0&0&0&0\\
    \end{bmatrix}
    \quad \text{is a row matrix, and} \qquad
    \begin{bmatrix}
        0&0&2&0\\
        0&0&2&0\\
        0&0&2&0\\
        0&0&2&0\\
    \end{bmatrix}
    \quad \text{is a column matrix.}
\]
So, $|\cX| = 2 \sqrt{n}$ since there are $\sqrt{n}$ distinct row matrices and $\sqrt{n}$ distinct column matrices.

\paragraph{The game.}
Given $n$ and oracle access to a matrix $F \in \cX$, the goal is to correctly declare whether $F$ is a row or column matrix.

\begin{restatable}{mylemma}{gridgametheoremour}
\label{thm:grid-game-theorem-our}
The randomized query complexity of the matrix game is $\Theta(\sqrt{n})$.
\end{restatable}
\begin{proof}[Proof sketch]
We  give a high level explanation of the proof, while the complete details can be found in \cref{sec-appendix-relational-adversary}.
For the upper bound, we can check that $\sqrt{n}$ queries suffice. Query the entries of the main diagonal and then proceed as follows: if any ``$1$'' is detected, declare ``row''; if any ``$2$'' is found, declare ``column''.

For the lower bound,  one intuitively expects that $\Omega(\sqrt{n})$ queries are necessary even allowing  randomization.
In fact, this is what we can show using \cref{thm:our-variant}.
Choose the function $r$ so that $r(F_1, F_2)$ represents the indicator function for whether the two matrices are ``opposite types'', for any two matrices $F_1, F_2 \in \cX$.
This means that the probability to sample any subset $\cZ \subseteq \cX$ will be proportional to its size $|\cZ|$.

Meanwhile, for any subset $\cZ \subseteq \cX$ of matrices, \emph{any single query} on a matrix coordinate will distinguish at most two matrices from all the others \emph{within} $\cZ$.
Thus, we can show that the ratio $M(\cZ)/q(\cZ) \in \Omega(\sqrt{n})$ for \emph{any} subset $\cZ \subseteq \cX$ of matrices, which implies a $\Omega(\sqrt{n})$ lower bound of the matrix game via \cref{thm:our-variant}.
\end{proof}

On the other hand, one can only show a lower bound of $\Omega(1)$ for the matrix game using the version of the relational adversary method from \cite{Aaronson06}; see  details in \cref{sec-appendix-relational-adversary}.
\subsection{Advantage of our variant for randomized algorithms} 

In addition to being strictly stronger on some problems (like this matrix game), our variant is at least as strong in general, for randomized algorithms.

\begin{restatable}{myproposition}{stronger}

Consider any problem and let $T$ be the expected number of queries required in the worst case by the best randomized algorithm to succeed with probability $9/10$.

If the relational adversary method from \cite{Aaronson06} provides a lower bound of $T \geq \Lambda$, then \cref{thm:our-variant} can prove a lower bound of $T \geq \Lambda/40$.
\end{restatable}

\section{Lower bound for local search via congestion}
\label{sec:application-congestion}

In this section, we explain at a high level the proof of \cref{thm:low_congestion_implies_local_search_hard}, which gives a lower bound as a function of congestion, as well as the corollary for expanders. 
See \cref{sec:appendix-congestion} for the details.

\subsection{Proof sketch for the congestion lower bound}
The proof of \cref{thm:low_congestion_implies_local_search_hard} is sketched in  the next sequence of  steps.

\paragraph{Fixing a set of paths $\mathcal{P}$.} Since the graph $G = ([n], E)$ has vertex congestion $g$, we can fix an all-pairs set of paths $\cP = \{P^{u,v}\}_{u,v \in [n]}$, such that $P^{u,v} \in \cP$ is a simple path from $u$ to $v$ and $P^{u,u} = (u)$, for each $u,v \in [n]$. Moreover, each vertex is used at most $g$ times across paths in $\cP$. 

\paragraph{Staircases.} We fix a parameter $L \in [n]$ to be set later. We consider sequences of vertices of the form $\vec{x} = (x_1, \ldots, x_{L+1}) \in \{1\} \times [n]^L$, i.e. with $x_1 = 1$.
The \emph{staircase} induced by  $\vec{x}$  is a walk $S_{\X} = S_{\X,1} \circ \ldots \circ S_{\X, L}$, where each $S_{\X, i}$ is a path in $G$ starting at vertex $x_{i}$ and ending at $x_{i+1}$.
Each vertex $x_i$ is called a \emph{milestone} and each path $S_{\X,i}$ a \emph{quasi-segment}.

The staircase $S_{\vec{x}}$ is \emph{induced by   $\vec{x}$ and  $\mathcal{P}$} if  we  additionally have $S_{\X,i} = P^{x_i, x_{i+1}}$ for all $i \in [L]$.
In other words, to build such a staircase we first decide on the sequence of ``milestones'' $\vec{x}$; then to get from each milestone $x_i$ to the next milestone $x_{i+1}$, we travel using the path $P^{x_i,x_{i+1}}$ from the set of paths $\mathcal{P}$.
Note that $P^{x_i,x_{i+1}}$ may not be the shortest path between $x_i$ and $x_{i+1}$ since $\mathcal{P}$ does not necessarily only consist of shortest paths.

\paragraph{The value function $f_{\vec{x}}$.} For each staircase $S_{\vec{x}}$ induced by $\vec{x}$ and $\mathcal{P}$, we define a corresponding   function $f_\X : [n] \to \mathbb{R}$ as follows. For each vertex $v$ in $G$:

 \begin{enumerate}[(a)]
        \item If $v \notin S_{\X}$, then  set $f_{\X}(v) = dist(v,1)$.
        \item If $v \in S_{\X}$, then set $f_{\X}(v) = -i\cdot n - j$, where $i$ is the maximum index with $v \in P^{x_i,x_{i+1}}$, and $v$ is the $j$-th vertex in $P^{x_i,x_{i+1}}$.
    \end{enumerate}
    
\cref{ex:staircase-example-congestion} gives a visual illustration of an induced staircase $S_{\X}$ and its associated function $f_{\X}$.
Such a function $f_\X$ has a unique local minimum at the end of the staircase $S_{\vec{x}}$. 
Our lower bound construction uses such functions.

\begin{example}[Staircase with the associated value function.]
\label{ex:staircase-example-congestion}
Consider the graph $G$ in \cref{fig:staircase_with_cycles_example}, with $n=12$ vertices labelled $\{v_1, \ldots, v_{12}\}$. \\
\begin{figure}[h!]
    \centering
    \subfigure[Graph with $n=12$ vertices]{
    \includegraphics[scale=1.1]{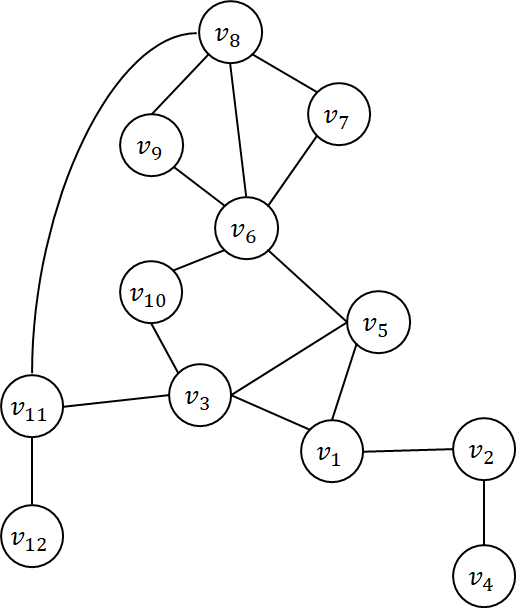}
    }\qquad 
    \subfigure[Staircase with  induced value function.]{
    \includegraphics[scale=1.1]{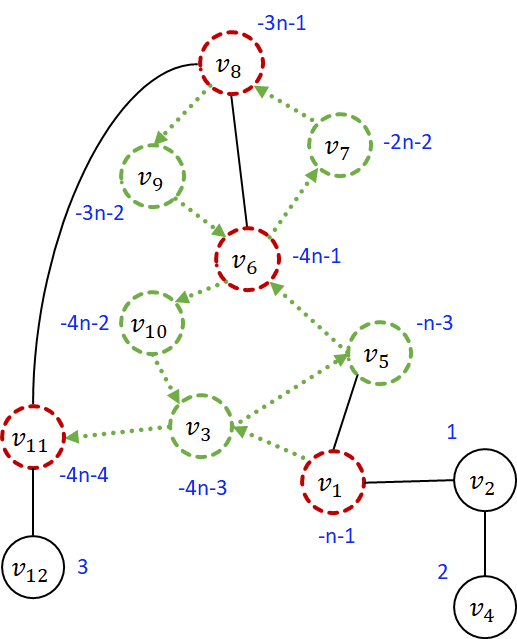}
    }
    \caption{Graph $G$ with $n=12$ vertices  labelled $\{v_1, \ldots, v_{12}\}$.  The staircase shown consists of the walk given by the  dotted vertices  $(v_1, v_3, v_5, v_6, v_7, v_8, v_9, v_6, v_{10}, v_3, v_{11})$.
    The value of the function at each vertex is shown in  blue  near that vertex. The local minimum is at the end of the staircase, at vertex $v_{11}$.}
    \label{fig:staircase_with_cycles_example}
\end{figure}

Let the set of paths $\mathcal{P} = \{ P^{u,v}\}_{u,v\in[n]}$ be 
\begin{itemize}
    \item $P^{v_1, v_6} = (v_1, v_3, v_5, v_6)$; $P^{v_6, v_8} = (v_6, v_7, v_8)$; $P^{v_8, v_6} = (v_8, v_9, v_6)$;\\
    $P^{v_6, v_{11}} = (v_6, v_{10}, v_3, v_{11})$.
    \item For all other pairs of vertices $(u,w)$, we set $P^{u,w}$ as the shortest path from $u$ to $w$.
\end{itemize}
Let $L=4$.
Consider a sequence $\X= (x_1, x_2, x_3, x_4, x_{4}) = (v_1, v_6, v_8, v_6, v_{11})$, where each milestone is highlighted by a red dotted circle.
Observe we allow repeated vertices. 
The staircase induced by $\vec{x}$ and $\mathcal{P}$, given by the green dotted walk, is \[S_{\vec{x}}= (v_1, v_3, v_5, v_6, v_7, v_8, v_9, v_6, v_{10}, v_3, v_{11})\,.
\]
The value function $f_{\vec{x}}$, computed using the definition (a-b), of each vertex is given in blue.
\end{example}

For technical reasons, in the lower bound proof we will actually work with a decision problem. There is a simple way to turn a search problem into a decision problem (see \cite{dinh2010quantum,Aaronson06}): associate with each function $f_{\vec{x}}$ a function $g_{\vec{x},b}$ that hides a bit at the local minimum vertex (while hiding $-1$ at every other vertex). Formally, $g_{\vec{x},b}$ is defined next.

Let $A=[n]$ and $B=\{- n^2 - n, \ldots, 0, \ldots, n\} \times \{-1,0,1\}$ be finite sets.
For each vertex $\X \in \{1\} \times [n]^L$ and  bit $b \in \{0,1\}$, let $g_{\X,b} : A \to B$ be such that, for all $v \in [n]$:
\begin{align} 
    g_{\X,b}(v) = \begin{cases}
            \bigl(f_{\X}(v), b\bigr) & \text{if} \; v = x_{L+1}\\
            \bigl(f_{\X}(v), -1\bigr) & \text{if}  \; v \ne x_{L+1}
        \end{cases} \,. \label{eq:def:g_X_b_main}
\end{align} 
The set of functions we consider is 
$\mathcal{X} = \Bigl\{g_{\X,b} \mid  \X \in \{1\} \times [n]^{L} \mbox{ and }  b \in \{0,1\}\Bigr\}\,.$
\medskip 

The decision problem is: given a graph $G$  and oracle access to a function $g_{\vec{x},b} \in \mathcal{X}$, return the value $\mathcal{H}(g_{\vec{x},b}) =b$. This means: find the hidden bit, which only exists at the vertex corresponding to the local minimum of $f_{\vec{x}}$.
Measuring the query complexity of this decision problem will give  the answer for local search, as the next two problems have query complexity within additive $1$:

\smallskip 

$\bullet \; \;$ \emph{search problem:} given oracle access to a function $f_{\vec{x}}$,  find a vertex $v$ that is a local minimum; 

\smallskip

$\bullet \; \;$ \emph{decision problem:}  given oracle access to the  function $g_{\vec{x},b}$, find $\mathcal{H}(g_{\vec{x},b})$.

\paragraph{Good/bad sequences of vertices; Good/bad functions.} We divide the set $\mathcal{X}$ of functions into ``good'' and ``bad'' functions.
Our analysis later will focus on ``good'' sequences and functions.

A sequence of vertices $\X = (x_1, \ldots, x_{L+1})$ is  \emph{good} if  $x_i \ne x_j$ for all $i,j \in [n]$ with $ i < j$; otherwise, $\X$ is \emph{bad}.
That is, $\X$ only involves distinct milestones.

For each bit $b \in \{0,1\}$, the function $g_{\X,b} \in \mathcal{X}$ is \emph{good} if $\X$ is good, and \emph{bad} otherwise.

\paragraph{The relation function $r$.} To be able to invoke \cref{thm:our-variant}, we need to also define the relation function $r$ whose role is to ``relate'' pairs of input functions $F_1$ and $F_2$ in order to roughly capture the difficulty of differentiating $F_1$ from $F_2$.

\smallskip 

Intuitively, an algorithm $\Gamma$ will query vertices of the graph to eliminate options to figure out which is the underlying input function, from which it can query the local minimum to retrieve the hidden bit $b$.
However, if two functions $F_1 = g_{\vec{x},b_1}$ and $F_2 = g_{\vec{y},b_2}$ are very ``similar'' (i.e. their underlying staircases $S_{\vec{x}}$ and $S_{\vec{y}}$  are almost identical), then it may take  $\Gamma$  many queries to learn whether the input is $F_1$ or $F_2$, even if it knows the input can only be one of these two functions. Consequently, $\Gamma$ will have great difficulty finding the local minimum on certain inputs.

\medskip 

 Formally, define $r : \mathcal{X} \times \mathcal{X} \to \mathbb{R}_{\ge 0}$  as a symmetric function such that for  each \mbox{$\X, \Y \in \{1\} \times [n]^L$} and $b_1,b_2 \in \{0,1\}$: 
    \[
        r(g_{\X,b_1}, g_{\Y,b_2}) = \begin{cases}
            0 & \text{ if at least one of the following holds: } b_1 = b_2 \text{ or } \X \text{ is bad or } \Y \text{ is bad.}\\
            n^j & \text{ otherwise, where } j \text{ is the maximum index for which } \X_{1 \to j} = \Y_{1 \to j}\,. 
        \end{cases}
    \]

The choice of $r$ is deliberate, as we see next.

\paragraph{Invoking \cref{thm:our-variant}.}
We are now ready to invoke \cref{thm:our-variant} using the definitions of $f_{\vec{x}}$, $g_{\vec{x},b}$, $\cX, \cH,$ and $r$ from above.
We will show there exists a subset $\cZ \subseteq \cX$ with $q(\cZ) > 0$, and so
we get that the randomized query complexity of local search is:
\begin{equation}
    \label{eq:main_congestion_raw_result}
    \Omega\left( \min\limits_{\substack{\cZ \subseteq \cX:\\ q(\cZ) > 0}}
\frac{M(\cZ)}{q(\cZ)} \right) \;, \text{ where } \; \; q(\cZ) = \max_{v \in [n]}{\sum_{F_1 \in \cZ} \sum_{F_2 \in \cZ} r(F_1,F_2) \cdot \mathbbm{1}_{\{F_1(v) \neq F_2(v)\}}} \,.
\end{equation}
Estimating the lower bound in \eqref{eq:main_congestion_raw_result} precisely is quite challenging.
Instead, for any arbitrary $\cZ \subseteq \cX$ with $q(\cZ) > 0$,
we show a lower bound $M(\cZ)$ and an upper bound for $q(\cZ)$ using our choice of the function $r$.
The two bounds we show only depend on $\cZ$ via its size $|\cZ|$.
For any fixed $\cZ$, this dependency will be cancelled out through the division, and thus our result follows.

\paragraph{Lower bounding $M(\cZ)$.}
By our choice of $r$, only \emph{good} functions affect the value of $M(\cZ)$.
Furthermore, we know that a function is good only if it was constructed using a good sequence of milestones $\cX$.
For any good function $F_1 \in \cX$, a counting argument tells us that there are roughly $n^{L+1-j}$ good functions $F_2 \in \cX$  with $r(F_1,F_2) = n^j$, i.e.\ $F_1$ and $F_2$ have the same milestones from the first to the $j^{th}$ one.
This approximation holds for $j \in O(\sqrt{n})$; we will eventually have $j \le L \in O(\sqrt{n})$, so this is fine.
Therefore, {for any good  function $F_1 \in \cX$,} we have 
\[
\sum_{F_2 \in \cX:\, r(F_1,F_2) = n^j} r(F_1,F_2)
\approx \sum_{j=1}^L n^j \cdot n^{L+1-j}
= L \cdot n^{L+1}\,.
\]
Therefore $M(\cZ) \in \Omega(|\cZ| \cdot L \cdot n^{L+1})$.

\paragraph{Upper bounding $q(\cZ)$.}
%
Let $v \in [n]$ be a vertex and  $F_1 = g_{\X,b} \in \cZ$ a good function 
with hidden bit $b \in \{0,1\}$. 
By our choice of $r$, it suffices to relate $F_1$ to  functions $F_2 \in \cZ$ that are  also good but have hidden bit  $1-b$, i.e. of the form  $F_2 = g_{\vec{y},1-b}$, for some good sequence $\vec{y}$.

To upper bound the inner summation of $q(\cZ)$, we partition the functions $F_2$ based on the length of the shared prefix 
of $\vec{y}$ and  $\vec{x}$. We get 



\begin{equation}
\label{eq:partition-to-j}
\sum_{F_2 \in \cZ} r(F_1,F_2) \cdot \mathbbm{1}_{\{F_1(v) \ne F_2(v)\}}
= \sum_{j=1}^{L+1} \sum_{\substack{F_2 \in \cZ:\\ j = \max\{i : \X_{i \to j} = \Y_{1 \to j}\}}} r(F_1,F_2) \cdot \mathbbm{1}_{\{F_1(v) \ne F_2(v)\}}
\,.
\end{equation}

When $j = L+1$, there is exactly one function $F_2 \in \cZ$ with such a shared prefix. Thus we get a contribution of $n^{L+1}$ by our definition of $r$.

Thus from now on, we focus on $1 \leq j \leq L$. 
Let $Tail(j, S_{\Y})$ denote the suffix of the staircase $S_{\Y}$ after $y_j$. We can then upper bound the $j$-th term from \eqref{eq:partition-to-j} as follows:
\begin{equation}
\label{eq:upper-bound-by-tail}
\sum_{\substack{F_2 \in \cZ:\\ j = \max\{i : \X_{i \to j} = \Y_{1 \to j}\}}} r(F_1,F_2) \cdot \mathbbm{1}_{\{F_1(v) \ne F_2(v)\}}
\leq 2 \cdot \sum_{\substack{F_2 \in \cZ:\\ j = \max\{i : \X_{i \to j} = \Y_{1 \to j}\}\\ v \in Tail(j, S_{\Y})}} r(F_1,F_2)    \,.
\end{equation}

Recall the staircase has the form $S_{\Y} = P^{y_1, y_2} \circ \ldots \circ P^{y_j, y_{j+1}} \circ \ldots \circ P^{y_L, y_{L+1}}$, where each $P^{y_i, y_{i+1}}$ is a path from  $\cP$.
When $v$ is in $Tail(j, S_{\Y})$, it means that $v \in P^{y_i, y_{i+1}}$ for some $j \leq i \leq L$. We will upper  bound the number of  sequences $\vec{y}$ depending on the location of $v$.
\begin{description}
    \item[\qquad Case $i = j$.] Let $\numPaths{v}(u)$ denote the number of paths in the set $\mathcal{P}$ that start at vertex $u$ and contain $v$. There are $\numPaths{v}(x_j)$ choices of vertices for $y_{j+1}$ and $n^{L-j}$ choices for sequences $(y_{j+2}, \ldots, y_{L+1})$, yielding a total count of $\numPaths{v}(x_j) \cdot n^{L-j}$.
    \item[\qquad Case $j < i \leq L$.] There are at most $L$ choices for $i$ such that $j < i \leq L$. For fixed $i$, there are at most $g$ choices for a pair $(y_i, y_{i+1})$ such that $v \in P^{y_i, y_{i+1}}$ since each vertex appears in at most $g$ paths within $\cP$. For fixed $i$ and $(y_i, y_{i+1})$, there are $n^{L-j-1}$ tuples of the form $(y_{j+1}, \ldots, y_{i-1}, y_{i+2}, \ldots, y_{L+1})$. That is, the total count is at most $L \cdot g \cdot n^{L-j-1}$.
\end{description}


The two cases are not mutually exclusive, so we  will combine the  counts in \eqref{eq:partition-to-j} and \eqref{eq:upper-bound-by-tail} by summing them. Then we obtain the following upper bound for the inner summation of $q(\mathcal{Z})$:
\[
\sum_{F_2 \in \cZ} r(F_1,F_2) \cdot \mathbbm{1}_{\{F_1(v) \ne F_2(v)\}}
\leq n^{L+1} + 2 \cdot \sum_{j=1}^L n^j \cdot \left( \numPaths{v}(x_j) \cdot n^{L-j} + L \cdot g \cdot n^{L-j-1} \right)\,.
\]

Since sequence $\X$ is good, it has no repeated vertices. Thus the elements of  $\vec{x}$ represent a subset of $[n]$. 
This yields the inequality $\sum_{j=1}^L \numPaths{v}(x_j) \leq \sum_{u \in [n]} \numPaths{v}(u) \leq g$ for any vertex $v$, where $g$ is the vertex congestion of $\mathcal{P}$ and $\numPaths{v}(u)$ is the number of paths in $\mathcal{P}$ that start at $u$ and contain $v$. We thus obtain the next bound on the inner summation of $q(\mathcal{Z})$: 
\[
\sum_{F_2 \in \cZ} r(F_1,F_2) \cdot \mathbbm{1}_{\{F_1(v) \ne F_2(v)\}}
\leq n^{L+1} + 2 \cdot L \cdot n^L \cdot g \cdot \left( 1 + \frac{L \cdot g}{n} \right) \,.
\]

Finally, since $n \leq g$, we get $q(\cZ) \in \cO(|\cZ| \cdot g \cdot n^L \cdot (1 + L^2/n))$. 

\paragraph{Wrapping up.}
Since we showed $M(\cZ) \in \Omega(|\cZ| \cdot L \cdot n^{L+1})$ and $q(\cZ) \in \cO(|\cZ| \cdot g \cdot n^L \cdot (1 + L^2/n))$ for arbitrary $\cZ$ with $q(\cZ) > 0$,
our randomized query complexity bound \cref{eq:main_congestion_raw_result} yields
\[
    \Omega\left(\min_{\substack{\cZ \subseteq \cX:\\ q(\cZ) > 0}}
    \frac{M(\cZ)}{q(\cZ)}\right)
    \subseteq \Omega\left(\frac{|\cZ| \cdot L \cdot n^{L+1}}{|\cZ| \cdot g \cdot n^L \cdot (1 + L^2/n)}\right)
    \subseteq \Omega\left(\frac{n \cdot L}{g\cdot (1 + L^2/n)}\right)\,.
\]
Setting $L \approx \sqrt{n}$ gives a query complexity of local search of $\Omega(n^{1.5}/g)$.

\subsection{Corollary for expanders}
To obtain the lower bound for expanders from \cref{cor:expanders_lower_bound}, we use a result from ~\cite{BFU99}.
Their work shows that constant-degree constant-expansion graphs have an all-pairs set of paths with vertex congestion $g \in \cO(n \ln n)$; see \cref{sec:appendix-expansion-and-congestion} for details.
 \cref{thm:low_congestion_implies_local_search_hard} then implies  the randomized query complexity of local search on such graphs is $\Omega\left(\frac{\sqrt{n}}{\log{n}}\right)$.
\section{Lower bound for local search  via  separation number}
\label{sec:application-separation-number}

We briefly discuss how to obtain the lower bound of $\Omega\left(\sqrt[4]{\frac{s}{\Delta}}\right)$ of \cref{thm:separation-number-result}, where $s$ is the separation number and $\Delta$ the maximum degree of the graph. For details, see  \cref{sec:appendix-separation-number}.

We apply \cref{thm:our-variant} with a similar strategy as the one discussed in \cref{sec:application-congestion} of lower bounding $M(\cZ)$ and upper bounding $q(\cZ)$ for arbitrary subset $\cZ \subseteq \cX$ with $q(\cZ) > 0$.

However, we use a slightly different $r$ function and now construct staircases with respect to another graph-theoretic notion known as ``path arrangement parameter'' and ``cluster walks'' (instead of using a pre-defined set of pairs $\cP$ with congestion $g$). Our construction is heavily inspired by that in \cite{santha2004quantum}.
While the construction is highly similar, using it within \cref{thm:our-variant} is non-trivial.

\bibliographystyle{alpha}

\bibliography{local_search_bib}

\appendix
\newpage
\section{Theorems from prior work}
\label{sec:appendix-expansion-and-congestion}

In this section we include several  theorems from prior work.

The first is a result from \cite{BFU99} about systems of paths with low congestion for expanders.


\begin{mytheorem}[\cite{BFU99}, Theorem 1]
\label{thm:bfu99}
Let $G = (V,E)$ be a $d$-regular $\beta$-expander where $d \in \mathbb{N}$ and $\beta \in \mathbb{R}_+$ are constant.
Let $\alpha: \mathbb{N} \to \mathbb{R}_+$ be a function. Consider a collection of $K = \alpha(n) n / \ln(n)$ pairs of vertices denoted $\{(a_1, b_1), \ldots, (a_K, b_K)\}$ such that no vertex in $V$ participates in more than $s$ pairs.

Then there is a set of $K$ paths $\{P_1, \ldots, P_K\}$ such that $P_i$ connects $a_i$ to $b_i$ and the congestion on each edge is at most 
\begin{align}
    g = \begin{cases}
        \cO\left(s + \left\lceil \frac{\ln\ln n}{\ln(1/\min\{\alpha,\, 1/\ln \ln n\})}\right\rceil\right) & \text{ for } \alpha < \frac{1}{2}\,.\\
        \cO(s + \alpha + \ln \ln n) & \text{ for } \alpha \ge \frac{1}{2}\,.
    \end{cases}
\end{align}
\end{mytheorem}

Next is a result from \cite{leighton1999multicommodity} on multi-commodity flow, which gives a corollary for finding  systems of paths with low congestion for expanders. We state the corollary as described in \cite{chuzhoy2016routing}.



\begin{mytheorem}[\cite{chuzhoy2016routing}, Corollary C.2]
\label{thm:chuzhoy2016routing}

Let $G = ([n],E)$ be a $\beta$-expander with maximum vertex degree $\Delta$ and let $M$ be any partial matching over the vertices of $G$.
Then there is an efficient randomized algorthm that finds, for every pair $(u,v) \in M$, a set $\cP_{u,v}$ of $\lceil \ln n \rceil$ paths of length $O(\Delta \cdot \ln (n) / \beta)$ each, such that the set $\cP = \bigcup_{(u,v) \in M} \cP_{u,v}$ has edge congestion $O(\ln^2 (n) / \beta)$.
The algorithm succeeds with high probability.

\end{mytheorem}

Next we present a lemma reducing local search to a decision problem.
This reduction is not new; see for example \cite{dinh2010quantum}.
\begin{mylemma} \label{lem:reduction_to_decision}
Suppose the randomized query complexity of local search on $G$ is $\chi$, recalling that in local search we have  a graph $G$ and a function  $f:V \to \mathbb{R}$, and the problem is to find a local minimum.
Then the query complexity of the following decision problem  is at most $\chi+1$:
\begin{itemize}
    \item Input: graph $G$ and oracle access to a function $h_b:[n] \to \mathbb{R} \times \{-1,0,1\}$ for some $b \in \{0,1\}$ with the property that $h_b(v) = (f(v),-1)$ when $v$ is not a local minimum of $f$, and $h_b(v) = (f(v),b)$ when $v$ is a local minimum of $f$.
    \item Output: the bit $b$.
\end{itemize}
To clarify, the algorithm is given oracle access to $h_b$, but it is not given $b$ itself.
\end{mylemma}
\begin{proof}
Let $\Gamma$ be a randomized algorithm that can solve local search on $G$ such that  
\begin{itemize} 
\item $\Gamma$ has success probability at least $p$ on every input $\langle G, h_b\rangle$.
\item $\Gamma$ issues at most $\chi$ queries in expectation.
\end{itemize}

We will use the local search algorithm $\Gamma$ to construct an algorithm $\Gamma_d$ that solves the decision problem. 
To simulate a query on $f(v)$ for $\Gamma$, algorithm $\Gamma_d$ will query the function $h_b$ at $v$, obtain $(f(v), c)$, with $c \in \{-1,0,1\}$, and will pass $f(v)$ to $\Gamma$.
Whenever $\Gamma$ locates a local minimum $v_{\min}$, we know that $h_b(v_{\min})$ contains the hidden bit output required by $\Gamma_d$.
Obtaining that hidden bit then requires only one additional query, at $v_{\min}$.
Since $\Gamma$ locates a local minimum $v_{\min}$ with $\chi$ queries in expectation and succeeds with probability at least $p$, we see that $\Gamma_d$ uses $\chi+1$ queries in expectation and succeeds with probability at least $p$.
\end{proof}
\newpage 
\section{The relational adversary method and our variant}
\label{sec-appendix-relational-adversary}

In this section we show our variant of the original relational adversary method from~\cite{Aaronson06}.
First, we introduce some preliminaries.
Consider any two functions $f,g : A \to B$, for some sets $A,B$.
An element $a \in A$ is said to \emph{distinguish} the function $f$ from $g$ if $f(a) \neq g(a)$.


Next, we include the original statement from \cite{Aaronson06}, written with our notation.

\begin{restatable}[\cite{Aaronson06}, Theorem 5]{mytheorem}{aaronsonGeneralLbOriginal}
\label{thm:aaronson_general_lb_original}
Consider finite sets $A$ or $B$.
Let $\cH : B^A \to \{0,1\}$ be a map that labels each function $F : A \to B$  with $0$ or $1$.
Let $\cA \subseteq \cH^{-1}(0)$ and $\cB \subseteq \cH^{-1}(1)$. The problem is: given $A,B,\cH,\cA,\cB$, and oracle access to a function $F^*$ \footnote{In other words, we have free access to $\cH$ and the only queries counted are the ones to $F^*$, which will be of the form: ``What is $F^*(a)$?'', for some $a \in A$. The oracle will return $F^*(a)$ in one computational step.} from $\cA$ or $\cB$, 
return the label $\cH(F^*)$.

Let $r : \mathcal{A} \times \mathcal{B} \to \mathbb{R}_{\geq 0}$ be a non-zero real-valued function of our choice. For $F_1 \in \cA, F_2 \in \cB$, and  $a \in A$, define 
\begin{equation}
\label{eq:theta}
\theta(F_1,a) = \frac{\sum_{F_3 \in \cB \;:\; F_1(a) \ne F_3(a)} r(F_1,F_3)}{\sum_{F_3 \in \cB} r(F_1,F_3)} \qquad \text{and} \qquad 
\theta(F_2,a) = \frac{\sum_{F_3 \in \cA \;:\; F_2(a) \ne F_3(a)} r(F_3,F_2)}{\sum_{F_3 \in \cA} r(F_3,F_2)} 
\end{equation}
whenever the denominators in \eqref{eq:theta} are all nonzero.
Then the randomized query complexity \footnote{Recall we defined the randomized query complexity as the expected number of queries required to achieve success probability at least $9/10$.} of the problem  is $1/(5 \cdot v_{min})$, where
\begin{align*}
    v_{min} = \max_{F_1 \in \cA, F_2 \in \cB, a \in A \;:\; r(F_1,F_2) > 0, F_1(a) \ne F_2(a)} \min \Bigl\{\theta(F_1,a), \theta(F_2,a) \Bigr\}\,.
\end{align*}
\end{restatable}

The proof centers on a difficult input distribution under which the denominator of $\theta(F_1,a)$ (respectively $\theta(F_2,a)$) is the likelihood that $F_1$ (respectively $F_2$) is sampled, conditioned on the input being sampled from $\cA$ (respectively $\cB$).
The numerator of $\theta(F_1,a)$ (respectively $\theta(F_2,a)$) is the progress that is made via querying $a$ towards distinguishing $F_1$ (respectively $F_2$) from the other functions.

Our variant uses the same framework of relating pairs of inputs through some ``relation'' $r$, but the lower bound expression is based on another average type of argument.

\subsection{Our variant of the relational adversary method}
\label{sec:relative-adv-ours}

Now we restate our variant of the relational adversary method.

\ourvariant*

Let us briefly  interpret \cref{thm:our-variant}.
We use Yao's lemma, and so it will suffice to design a ``hard'' distribution of input functions and analyze the performance of a deterministic algorithm when given inputs from this distribution.

The theorem considers the probability distribution $P$ where each function $F \in \cX$ is given as input with probability $P(F) = \frac{M(\{F\})}{M(\cX)}$, where $\mathcal{X}$ is the space of possible functions.

The quantity $M(\cZ)$ is the likelihood that the function $F^*$ sampled from this distribution lies in $\cZ$.
The quantity $q(\cZ)$ is the largest amount of progress possible in a single query once the algorithm already knows that the given function $F^*$ lies in $\cZ$.

\medskip 
Now we are ready to prove \cref{thm:our-variant}.

\begin{proof}[Proof of Theorem~\ref{thm:our-variant}]
Given a relation $r$ with the properties required by the theorem, define \begin{align}
    M(\{F_1\}) & = \sum_{F_2 \in \cX} r(F_1,F_2) \; \; \forall F_1 \in \cX \; \; \; \mbox{and} \; \; \; 
    M(\cX) = \sum_{F_1 \in \cX} M(\{F_1\}) \,.
\end{align}  
We  consider the   distribution $P$ over functions in $\cX$ that selects each function $F \in \cX$ with probability 
$P(F) = {M(\{F\})}/{M(\cX)}\,.$
The theorem claims a lower bound when there exists a subset $\cZ \subseteq \cX$ with $q(\cZ) > 0$, so we may assume such a subset $\cZ$ exists.
By \cref{lem:q_positive_implies_M_positive}, this implies $M(\cZ) > 0$, and so
\begin{align} \label{eq:M_X_positive}
M(\cX) > 0 \,.
\end{align}
Thus $P(F) = M(\{F\})/M(\cX)$ is well defined.

We say an algorithm succeeds on an input function $F$ if it outputs the correct label $\cH(F)$. Let $R$ be the randomized query complexity (on the worst case input $F \in \cX$) for success probability $19/20$.
Let $\Gamma_d$ be the best  {deterministic} algorithm \footnote{That is, with smallest expected number of queries possible.}  that succeeds with probability at least $9/10$ when the input is a random function drawn from distribution $P$.
Let $D$ be the expected number of queries issued by $\Gamma_d$ on input distribution $P$.
 Yao's lemma (\cite{yao1977minimaxprinciple}, Theorem 3) yields $2R \geq D$. Thus to lower bound $T$, it will suffice to lower bound $D$. Let $T = 10D$.


Let $\Gamma^*$ be the truncation of $\Gamma_d$ after $T$ queries. We will analyze the expected number of queries made by $\Gamma^*$ when facing distribution $P$. 
Let $X$ be the random variable representing the number of queries issued by $\Gamma_d$. Then $\Ex[X] = D$. 
By Markov's inequality, 
\begin{align}
    \Pr[X \ge T] \le \frac{\Ex[X]}{T} = \frac{D}{T} = \frac{1}{10}\,. \label{eq:application_markov_X_and_T}
\end{align}

We have  
    \begin{align} 
    \Pr[X>T] + \Pr[\Gamma_d \; \mbox{succeeds and } X \leq T] & \geq  \Pr[\Gamma_d \; \mbox{succeeds and } X>T] + \Pr[\Gamma_d \; \mbox{succeeds and } X \leq T] \notag \\
    & = \Pr[\Gamma_d \; \mbox{succeeds}]\,. \label{eq:gamma_d_succeeds_truncation_T}
    \end{align}
Then 
\begin{align}
    \Pr[\Gamma^* \text{ succeeds}] &= \Pr[\Gamma_d \text{ succeeds and } X \le T] \notag \\
        &\ge \Pr[\Gamma_d \text{ succeeds}] - \Pr[X > T] \explain{By \eqref{eq:gamma_d_succeeds_truncation_T}}\\
        &\ge 9/10 - 1/10 = 4/5 \explain{By choice of $\Gamma_d$ and \eqref{eq:application_markov_X_and_T}}
\end{align}




Algorithm $\Gamma^*$ is said to \emph{distinguish} $F_1 \in \cX$ from $F_2 \in \cX$ within the first $t$ queries if $\Gamma^*$ queries an element $a \in A$ with the property that $F_1(a) \neq F_2(a)$ within the first $t$ queries.
For all $t \in \mathbb{N}$, $F_1 \in \cX$, and $F_2 \in \cX$, define 
\[
I^{(t)}(F_1, F_2) = 
\begin{cases}
1 & \text{if algorithm $\Gamma^*$ distinguishes $F_1$ from $F_2$ within the first $t$ queries}\\
0 & \text{otherwise}\,.
\end{cases}
\]

For each function $F_1 \in \cX$ and index $t \in \{0, \ldots, T\}$, we define a ``local progress measure'' $S^{(t)}(F_1)$ that counts the number of elements $F_2 \in \cX$ distinguished from {$F_1$} by the $t$-th query, weighted by the relation $r(F_1,F_2)$.
Formally, for each $F_1 \in \cX$ and $0 \leq t \leq T$, let 
\[
S^{(t)}(F_1) = \sum_{\substack{F_2 \in \cX :\\ I^{(t)}(F_1, F_2) = 1}} r(F_1, F_2) \,.
\]
{Summing over all the functions in $\cX$}, we obtain a ``global progress measure'' $S^{(t)}$:
\[
S^{(t)} = \sum_{F_1 \in \cX} S^{(t)}(F_1) \,.
\]
The difference in progress between consecutive queries can then be defined as: 
\begin{equation}
\label{eq:delta_S_t}
\Delta S^{(t)} = S^{(t)} - S^{(t-1)} \,.
\end{equation}



To lower bound the progress, we show in three steps that (a) the initial value of the progress measure,  $S^{(0)}$, is zero;  (b) the final value of the progress measure, $S^{(T)}$, is ``large''; and (c) the difference in the progress measure between consecutive rounds, $\Delta S^{(t)}$, is ``small'' for each round $t$.

\paragraph{Step (a)}
No queries have been issued at time zero, so nothing has been distinguished.
Therefore
\begin{align}\label{eq:variant_claim_1}
    S^{(0)} = 0\,.
\end{align}

\paragraph{Step (b)}

Define $\cW = \{F_1 \in \cX \mid \Gamma^*(F_1) = \cH(F_1)\}$. This is the set of functions on which $\Gamma^*$ succeeds, i.e. finds the correct label. 

We claim   algorithm $\Gamma^*$ distinguishes each pair of functions $F_1, F_2 \in \cW$ for which $\cH(F_1) \ne \cH(F_2)$. That is,  $\Gamma^*$ must have queried an index $a \in A$ such that $F_1(a) \neq F_2(a)$ within the first $T$ queries.
To see this, suppose towards a contradiction   there exists a pair $F_1,F_2 \in \cW$ with $\cH(F_1) \neq \cH(F_2)$ such that $\Gamma^*$ only queries  indices $a \in A$ with $F_1(a) = F_2(a)$.
Then, $\Gamma^*$ cannot differentiate whether the input function  is $ F_1$ or $ F_2$, and so must have the same output $\Gamma^*(F_1) = \Gamma^*(F_2)$ since $\Gamma^*$ is deterministic. Thus $\Gamma^*$ makes a mistake on one of $F_1$ or $F_2$ because $\cH(F_1) \neq \cH(F_2)$.
This  contradicts the choice of  $F_1, F_2 \in \cW$ as  inputs on which $\Gamma^*$ is successful.
Thus, the algorithm $\Gamma^*$ distinguishes each pair of inputs $F_1, F_2 \in \cW$ within the first $T$ queries.

Formally, we have  
\begin{equation}
\label{eq:WF1_WF2_implie_PT}
I^{(T)}(F_1, F_2) = 1, \qquad \text{for all $F_1, F_2 \in \cW$ where $\cH(F_1) \neq \cH(F_2)$.}
\end{equation}

Moreover, since $\Gamma^*$ succeeds when the input is a function from $\cW$ and fails otherwise, the success probability on  input distribution $P$ is at least $4/5$,  and the distribution $P$ samples each function $F_1 \in \cX$ with probability $P(F_1) = {M(\{F_1\})}/{M(\cX)}$, we have 
\begin{align} \label{eq:gamma_star_success_prob_at_least_4_over_5}
	\sum_{F_1 \in \cW}   \frac{M(\{F_1\})}{M(\cX)} = \sum_{F_1 \in \cW} P(F_1) \geq 4/5
	\quad \text{ and } \quad
	\sum_{F_2 \in \cX \setminus \cW} \frac{M(\{F_2\})}{M(\cX)} \leq 1/5\,.
\end{align}

Thus,
\begin{align*} 
S^{(T)} & = \sum_{\substack{F_1,F_2 \in \cX :\\ I^{(T)}(F_1, F_2)=1}} r(F_1, F_2)  \explain{By definition of $S^{(T)}$}\\
&\geq \sum_{\substack{F_1,F_2 \in \cW: \\ I^{(T)}(F_1, F_2) = 1}} r(F_1, F_2) \explain{Since $\cW \subseteq \cX$ and $r$ is non-negative} \\
&= \sum_{F_1, F_2 \in \cW} r(F_1,F_2) \explain{By \cref{eq:WF1_WF2_implie_PT} and $r(F_1,F_2) = 0$ if $\cH(F_1) = \cH(F_2)$}\\
&= \sum_{F_1 \in \cW, F_2 \in \cX} r(F_1,F_2) - \sum_{F_1 \in \cW, F_2 \in \cX \setminus \cW} r(F_1, F_2)\\
&= \sum_{F_1 \in \cW} M(\{F_1\}) - \sum_{F_1 \in \cW, F_2 \in \cX \setminus \cW} r(F_1, F_2) \explain{By definition of $M(\{F_1\})$} \\
&\geq \sum_{F_1 \in \cW} M(\{F_1\}) - \sum_{F_1 \in \cX, F_2 \in \cX \setminus \cW} r(F_1, F_2) \explain{Since $\cW \subseteq \cX$ and $r$ is non-negative}\\
&= \sum_{F_1 \in \cW} M(\{F_1\}) - \sum_{F_1 \in \cX, F_2 \in \cX \setminus \cW} r(F_2, F_1) \explain{By symmetry of $r$.}\\
&= \sum_{F_1 \in \cW} M(\{F_1\}) - \sum_{F_2 \in \cX \setminus \cW} M(\{F_2\}) \explain{By definition of $M(\{F_2\})$} \\
&\geq \frac{4}{5}M(\cX) - \frac{1}{5}M(\cX)  \explain{By \cref{eq:gamma_star_success_prob_at_least_4_over_5}}\\
&= \frac{3}{5}M(\cX)\,.
\end{align*}

Therefore,
\begin{align}\label{eq:variant_claim_2}
    S^{(T)} \geq 3 \cdot M(\cX)/5\,.
\end{align}

\paragraph{Step (c)}
  By Lemma~\ref{lem:inequality_delta_S_t_over_M_X}, for each $t \in [T]$, the maximum progress made by the $t$-th query, $\Delta S^{(t)}$, can be bounded as follows:  
\begin{align}
    \Delta S^{(t)} \leq M(\cX) \cdot \max\limits_{\cZ \subseteq \cX:\; M(\cZ) > 0} \frac{q(\cZ)}{M(\cZ)}\,. \label{eq:variant_claim_3}  
\end{align} 

\paragraph{Combining steps (a,b,c).} We obtain:
\begin{align*}
\frac{3}{5} \cdot M(\cX)
&\leq S^{(T)} \explain{By \cref{eq:variant_claim_2}}\\
&= S^{(0)} + \sum_{t=1}^T \Delta S^{(t)} \explain{By definition of $S^{(T)}$}\\
&= 0 + \sum_{t=1}^T \Delta S^{(t)} \explain{Since $S^{(0)} = 0$ by \cref{eq:variant_claim_1}}\\
& \leq T \cdot M(\cX) \cdot \max\limits_{\cZ \subseteq \cX:\; M(\cZ) > 0} \frac{q(\cZ)}{M(\cZ)}  \explain{By \cref{eq:variant_claim_3}}\\
&= T \cdot M(\cX) \cdot \max\limits_{\cZ \subseteq \cX:\; M(\cZ) > 0 \text{ and } q(\cZ) > 0} \frac{q(\cZ)}{M(\cZ)}  \explain{Since there exists $\cZ \subseteq \cX$ with $q(\cZ) > 0$.}\\ 
&= T \cdot M(\cX) \cdot \max\limits_{\cZ \subseteq \cX:\; q(\cZ) > 0} \frac{q(\cZ)}{M(\cZ)} \,. \explain{Since   $q(\cZ) > 0$ implies $M(\cZ) > 0$.}
\end{align*}

Rearranging and denoting  $\zeta = \min\limits_{\cZ \subseteq \cX:\; q(\cZ) > 0}
\frac{M(\cZ)}{q(\cZ)}$, we get
\begin{align} \label{eq:lb_T_min_M_over_q}
T
\geq \frac{3}{5} \cdot \frac{1}{\max\limits_{\cZ \subseteq \cX:\; q(\cZ) > 0}
\frac{q(\cZ)}{M(\cZ)}}
= \frac{3 \zeta }{5}\,.
\end{align}

In summary, the randomized query complexity $R$ when the success  probability is $19/20$  can be bounded as follows: \begin{align} \label{eq:lb_R_3zeta_over_100}
R \geq \frac{D}{2} = \frac{T}{20} \ge \frac{3 \zeta }{100} \,.
\end{align}

Let $\mathcal{A}$ be  an arbitrary randomized algorithm that  succeeds with probability at least $9/10$. If $\mathcal{A}$ issues fewer than $  \zeta/100$ queries in expectation, 
then $\mathcal{A}$   could be repeated three times and the majority answer returned; this would achieve a greater than $19/20$ probability of success and fewer than $   3\zeta/100$ queries in expectation, which would contradict \cref{eq:lb_R_3zeta_over_100}.
Thus $\mathcal{A}$ issues at least $\zeta/100$ queries in expectation on its worst case input.
Thus the randomized query complexity is lower bounded  as required by the theorem statement.
\end{proof}

\begin{mylemma}
\label{lem:q_positive_implies_M_positive}
    Let $\cZ \subseteq \cX$ be a subset with $q(\cZ) > 0$.
    Then $M(\cZ) > 0$.
\end{mylemma}
\begin{proof} By definition of $M(\cZ)$, we have 
    \begin{align*}
        M(\cZ) &= \sum_{F_1 \in \cZ} \sum_{F_2 \in \cX} r(F_1,F_2) \\
        &\ge \max_{v \in [n]} \sum_{F_1 \in \cZ} \sum_{F_2 \in \cZ} r(F_1,F_2) \mathbbm{1}_{\{F_1(v) \ne F_2(v)\}} \explain{Since $\mathbbm{1}_{\{F_1(v) \ne F_2(v)\}} \le 1$ and $r(F_1,F_2) \ge 0$.}\\
        &= q(\cZ)\,. \explain{By definition of $q(\cZ)$.}
    \end{align*}
Since $\cZ$ was chosen such that $q(\cZ) > 0$ and  $M(\cZ) \geq q(\cZ)$, we conclude that $M(\cZ) > 0$ as required.
\end{proof}

\begin{lemma}
    \label{lem:inequality_delta_S_t_over_M_X}
    In the setting of step (c) of Theorem~\ref{thm:our-variant},  for  each $t \in [T]$, we have 
\begin{align}
    \Delta S^{(t)} \leq M(\cX) \cdot \max\limits_{\cZ \subseteq \cX:\; M(\cZ) > 0} \frac{q(\cZ)}{M(\cZ)}\,. \notag 
\end{align}
\end{lemma}
\begin{proof}
Let $\Psi = A^{t-1} \times B^{t-1}$ be the set of possible sequences of the first $t-1$
queries and their answers;  each $\psi \in \Psi$ a ``history''. 
Let $\psi_1, \psi_2, \ldots, \psi_k$ be  the histories in  $\Psi$, where $k = |A|^{t-1} \cdot |B|^{t-1}$.

For all $i \in [k]$, let $\cX_i \subseteq \cX$ be the set of inputs on which $\Gamma^*$ would have history $\psi_i$ for the first $t-1$ queries; this is well-defined because $\Gamma^*$ is deterministic.
In other words, we can partition $\cX$ into equivalence classes $\cX_1, \cX_2, \ldots, \cX_k$ so that two inputs have the same history over the first $t-1$ queries if and only if they are in the same equivalence class.
This induces three useful facts:

\paragraph{Fact 1.} For each $F_1,F_2 \in \cX_i$ for some $i$, we know that $F_1$ and $F_2$ must receive the same $t$-th query since $\Gamma^*$ is deterministic.
    Let $a^{(t)}_i \in A$ be the location of this query.
    For $i \ne j$, we may have $a^{(t)}_i \ne a^{(t)}_j$; i.e. each history may have a different query in the same round.

 For each $F_1 \in \cX_i$ and $F_2 \in \cX_j$ with $i \neq j$, functions $F_1$ and $F_2$ have already been distinguished since they have different histories.
    This implies that all future progress must come from pairs from the same equivalence class.
    For all $a \in A$ and $F_1,F_2 \in \cX$, let us define 
    \[ r_a(F_1,F_2) = r(F_1,F_2) \cdot \mathbbm{1}_{\{F_1(a) \neq F_2(a)\}} \,.
    \] 
Then  
    \begin{align}
    \Delta S^{(t)} &= S^{(t)} - S^{(t-1)}\notag\\
    &= \sum_{F_1 \in \cX} \left(  S^{(t)}(F_1) - S^{(t-1)}(F_1)\right) \explain{By definition of $S^{(t)}$ and $S^{(t-1)}$.}\\
    &= \sum_{F_1 \in \cX} \sum_{\substack{F_2 \in \cX\\ I^{(t)}(F_1,F_2) = 1\\ I^{(t-1)}(F_1,F_2) = 0}} r(F_1,F_2) \explain{By definition of $S^{(t)}(F_1)$ and $S^{(t-1)}(F_1)$}\\
    &= \sum_{i=1}^k \sum_{F_1 \in \cX_i} \left( \sum_{\substack{F_2 \in \cX_i\\ I^{(t)}(F_1,F_2) = 1\\ I^{(t-1)}(F_1,F_2) = 0}} r(F_1,F_2) + \sum_{\substack{F_2 \not \in \cX_i\\ I^{(t)}(F_1,F_2) = 1\\ I^{(t-1)}(F_1,F_2) = 0}} r(F_1,F_2)  \right) \notag\\
    &= \sum_{i=1}^k \sum_{F_1 \in \cX_i} \sum_{\substack{F_2 \in \cX_i\\ I^{(t)}(F_1,F_2) = 1}} r(F_1,F_2)  \explain{By definition of $\cX_i$, each function $F_2$ with  $I^{(t-1)}(F_1,F_2) = 0$ is in $\cX_i$.}\\
    &= \sum_{i=1}^k \sum_{F_1 \in \cX_i} \sum_{F_2 \in \cX_i} r(F_1,F_2) \cdot \mathbbm{1}_{\{F_1(a_i^{(t)}) \neq F_2(a_i^{(t)}\}} \explain{Since $F_1$ and $F_2$ are distinguished at time $t$, when query $a_i^{(t)}$ is issued.} \\ 
    &= \sum_{i=1}^k \sum_{F_1 \in \cX_i} \sum_{F_2 \in \cX_i} r_{a_i^{(t)}}(F_1,F_2)  \,. \label{eq:St}
    \end{align}

\paragraph{Fact 2.} Since $\cX_1, \cX_2, \ldots, \cX_k$ is a partition of $\cX$, we have
    \begin{equation}
    \label{eq:M*}
    M(\cX) = \sum_{i=1}^k M(\cX_i)\,.
    \end{equation}

\paragraph{Fact 3.} For all $i \in [k]$, we have
    \begin{align} 
        M(\cX_i) = \sum_{F_1 \in \cX_i} \sum_{F_2 \in \cX} r(F_1,F_2) \ge \sum_{F_1 \in \cX_i} \sum_{F_2 \in \cX_i} r_{a_i^{(t)}}(F_1,F_2) \notag \,.
    \end{align}
    Therefore, since $r$ is non-negative,
    \begin{equation} \label{eq:bi0_implies_ai0}
        \Big( M(\cX_i) = 0 \Big) \implies \left( \sum_{F_1 \in \cX_i} \sum_{F_2 \in \cX_i} r_{a_i^{(t)}}(F_1,F_2) = 0 \right)\,.
    \end{equation}

\paragraph{Combining facts 1, 2,  and 3.} We get
\begin{align*}
\frac{\Delta S^{(t)}}{M(\cX)}
&= \frac{\sum_{i \in [k]} \sum_{F_1 \in \cX_i} \sum_{F_2 \in \cX_i} r_{a_i^{(t)}}(F_1,F_2)}{\sum_{i \in [k]} M(\cX_i)}  \explain{From \cref{eq:St} and \cref{eq:M*}}\\
&= \frac{\sum_{\substack{i \in [k]:\\ M(\cX_i) > 0}} \left( \sum_{F_1 \in \cX_i} \sum_{F_2 \in \cX_i} r_{a_i^{(t)}}(F_1,F_2) \right) }{\sum_{\substack{i \in [k]:\\ M(\cX_i) > 0}} M(\cX_i)}  \explain{By \cref{eq:bi0_implies_ai0}}\\
&  \leq \max_{\substack{i \in [k]:\\ M(\cX_i) > 0}} \frac{\sum_{F_1 \in \cX_i} \sum_{F_2 \in \cX_i} r_{a_i^{(t)}}(F_1,F_2)}{M(\cX_i)} \explain{Maximizing over $i \in [k]$ with $ M(\cX_i) > 0$} \\
& {\leq \max_{\substack{\cZ \subseteq \cX,\; a \in A:\\ M(\cZ) > 0}} \frac{\sum_{F_1 \in \cZ} \sum_{F_2 \in \cZ} r_a(F_1,F_2)}{M(\cZ)}  \explain{Maximizing over $\cZ \subseteq \cX$ and $a \in A$}}\\
&= \max_{\cZ \subseteq \cX:\; M(\cZ) > 0} \frac{q(\cZ)}{M(\cZ)}\,.  \explain{By definition of $q(\cZ)$}
\end{align*}

Therefore
$ \Delta S^{(t)} \leq M(\cX) \cdot \max\limits_{\cZ \subseteq \cX:\; M(\cZ) > 0} \frac{q(\cZ)}{M(\cZ)}$, as required by the lemma.
\end{proof}
\subsection{Matrix game}
\label{sec:appendix-matrix-game}

In this section, we describe a toy problem which highlights the advantage of \cref{thm:our-variant} over \cref{thm:aaronson_general_lb_original}.

\paragraph{Setup}
Let $n \in \N$ be a square number and $\cX$ be a subset of square $\sqrt{n} \times \sqrt{n}$ matrices with entries from $\{0,1,2\}$.
There are two types of matrices within $\cX$: ``row'' matrices and ``column'' matrices.
Row matrices have one row of $1$s with all other entries $0$ while column matrices have one column of $2$s with all other entries $0$.
For example:
\[
    \begin{bmatrix}
        0&0&0&0\\
        1&1&1&1\\
        0&0&0&0\\
        0&0&0&0\\
    \end{bmatrix}
    \quad \text{is a row matrix;  } \qquad  
    \begin{bmatrix}
        0&0&2&0\\
        0&0&2&0\\
        0&0&2&0\\
        0&0&2&0\\
    \end{bmatrix}
    \quad \text{is a column matrix.}
\]
So, $|\cX| = 2 \sqrt{n}$ since there are $\sqrt{n}$ distinct row matrices and $\sqrt{n}$ distinct column matrices.

\paragraph{The game}
Given $n$ and oracle access to a matrix $F \in \cX$, the goal is to correctly declare whether $F$ is a row or column matrix.

One can check that $\sqrt{n}$ queries suffices by querying the main diagonal: if any ``$1$'' is detected, declare ``row''; if any ``$2$'' is found, declare ``column''.
Even with randomization, one intuitively expects that $\Omega(\sqrt{n})$ queries are necessary.
In fact, this is what we can show using \cref{thm:our-variant}.

\gridgametheoremour*
\begin{proof}
The upper bound of $\cO(\sqrt{n})$ follows by querying the main diagonal: if any ``$1$'' is detected, declare ``row''; if any ``$2$'' is found, declare ``column''.

To show the lower bound of $\Omega(\sqrt{n})$, we instantiate \cref{thm:our-variant} with the following definitions:
\begin{itemize}
    \item Finite set $A$ is the set of $n$ coordinates within $\sqrt{n} \times \sqrt{n}$ matrix.
    \item Finite set $B$ is the space of all possible $\sqrt{n} \times \sqrt{n}$ matrices with cell values from $\{0,1,2\}$.
    \item Set $\cX \subseteq B^A$ of $\sqrt{n}$ row matrices and $\sqrt{n}$ column matrices, for a total of $2 \sqrt{n}$ matrices.
    So, for any given $F \in \cX$ and $a \in A$, we have that $F(a) \in \{0,1,2\}$ is the value of the matrix at the coordinate indicated by $a$.
    \item Mapping $\cH : \cX \to \{0,1\}$ refers to deciding whether the matrix from $\cX$ is a row or column matrix: output 0 if ``row'' and output 1 if ``column''.
    \item For any two $F_1, F_2 \in \cX$, we define $r(F_1,F_2) = \mathbbm{1}_{\{\cH(F_1) \neq \cH(F_2)\}}$ to be the indicator whether $F_1$ and $F_2$ are of the same type.
\end{itemize}

There exists $\cZ \subseteq \cX$ with $q(\cZ) > 0$; in particular, $\cZ = \{F_{row}, F_{col}\}$ for any row function $F_{row}$ and column function $F_{col}$ has $q(\cZ) = 2$. Therefore we may invoke \cref{thm:our-variant}.

Under this instantiation, we have $\sum_{F_2 \in \cX} r(F_1,F_2) = \sqrt{n}$ for any $F \in \cX$.
Thus 
\begin{equation}
\label{eq:matrix-game-M}
M(\cZ) = \sum_{F_1 \in \cZ} \sum_{F_2 \in \cX} r(F_1, F_2) = |\cZ| \cdot \sqrt{n}
\quad \text{for any $\cZ \subseteq \cX$} \,. 
\end{equation}

Meanwhile, let $\cZ_{row} = \{F_1 \in \cZ : \cH(F_1) = 0\}$ and $\cZ_{col} = \{F_2 \in \cZ : \cH(F_2) = 1\}$.
For every $F_1,F_2 \in \cX$ with $\cH(F_1) \neq \cH(F_2)$, we have $F_1(a) \neq F_2(a)$ if and only if $a$ lies on $F_1$'s row/column or $F_2$'s row/column.
Furthermore, because for each row/column there is only one corresponding input in $\cX$, we have
\begin{equation}
\label{eq:one_input_per_row_col}
\Bigl|\{F \in \cZ \mid F(v) = 1\} \Bigr| \leq 1
\quad \text{and} \quad
\Bigl| \{F \in \cZ \mid F(v) = 2\} \Bigr| \leq 1 \;.
\end{equation}

For any arbitrary $\cZ \subseteq \cX$ and $a \in A$, we have
\begin{align}
&\sum_{F_1 \in \cZ} \sum_{F_2 \in \cZ} r(F_1,F_2) \cdot \mathbbm{1}_{\{F_1(a) \neq F_2(a)\}}\\
& \;\;\;  = \sum_{\substack{F_1 \in \cZ:\\ F_1(a) \neq 0}} \sum_{F_2 \in \cZ} r(F_1,F_2) \cdot  \mathbbm{1}_{\{F_1(a) \neq F_2(a)\}} + \sum_{\substack{F_1 \in \cZ:\\ F(a) = 0}} \sum_{F_2 \in \cZ} r(F_1,F_2) \cdot \mathbbm{1}_{\{F_1(v) \neq F_2(v)\}}\\
 &\;\;\; \leq \left( |\cZ_{row}| + |\cZ_{col}| \right)  + \sum_{\substack{F_1 \in \cZ:\\ F_1(a) = 0}} \sum_{F_2 \in \cZ} r(F_1,F_2) \cdot \mathbbm{1}_{\{F_1(a) \neq F_2(a)\}} \explain{By \cref{eq:one_input_per_row_col}}\\
&\;\;\; = |\cZ_{row}| + |\cZ_{col}| + \sum_{\substack{F_1 \in \cZ:\\ F_1(a) = 0}} \sum_{\substack{F_2 \in \cZ:\\ F_2(a) \neq 0}} r(F_1,F_2) \cdot \mathbbm{1}_{\{F_1(a) \neq F_2(a)\}}\label{eq:dag} \\
 &\;\;\; \leq |\cZ_{row}| + |\cZ_{col}| + \sum_{\substack{F_1 \in \cZ:\\ F_1(a) = 0}} 1\label{eq:ast}\\
 &\;\;\; \leq 2 \cdot |\cZ| \,. \label{eq:matrix_qz_small}
\end{align}
In the equations above,  \cref{eq:dag} holds since $F_1(a) = F_2(a) = 0$ implies $\mathbbm{1}_{\{F_1(a) \neq F_2(a)\}} = 0$. 

Next we explain  why \cref{eq:ast} holds. Take arbitrary $F_1, F_2 \in \mathcal{Z}$ with $F_1(a) = 0$ and $F_2(a) \ne 0$. Then $r(F_1,F_2) = 1$ when  $\cH(F_1) \ne \cH(F_2)$ and $r(F_1,F_2) = 0$ when $\cH(F_1) = \cH(F_2)$. By \cref{eq:one_input_per_row_col}, there are at most two functions $F_2 \in \cZ$ with $F_2(a) \ne 0$: one with $\cH(F_2) = 0$ and one with $\cH(F_2) = 1$.
So no matter what $\cH(F_1)$ is, there is at most one $F_2 \in \cZ$ with $F_2(a) \ne 0$ such that $r(F_1,F_2) = 1$. Thus \cref{eq:ast} holds.

Therefore by \cref{eq:matrix_qz_small},
\begin{equation}
\label{eq:matrix-game-q}
q(\cZ) \leq \max_{a \in A}{\sum_{F_1 \in \cZ} \sum_{F_2 \in \cZ} r(F_1,F_2) \cdot \mathbbm{1}_{\{F_1(a) \neq F_2(a)\}}} \leq 2 \cdot |\cZ|
\quad \text{for any $\cZ \subseteq \cX$}\,.
\end{equation}

By \cref{eq:matrix-game-M} and \cref{eq:matrix-game-q}, the statement of \cref{thm:our-variant} implies the query complexity is at least
\[
\Omega
\left(
\min\limits_{\substack{\cZ \subseteq \cX:\\ q(\cZ) > 0}}
    \frac{{M(\cZ)}}{q(\cZ)}
\right)
\subseteq
\Omega \left( \frac{|\cZ| \cdot \sqrt{n}}{2 \cdot |\cZ|} \right)
\subseteq
\Omega(\sqrt{n})\,.
\]
\end{proof}

On the other hand, one can only show that the randomized query complexity of the matrix game is $\Omega(1)$ by using the original relational adversary theorem of \cite{Aaronson06}, illustrating the advantage of our new variant on this matrix game.
Intuitively, the reason why \cref{thm:aaronson_general_lb_original} cannot handle this problem is because \emph{every} pair of row matrix $F_{row}$ and column matrix $F_{col}$ has \emph{some} location $x$ that distinguishes both $F_{row}$ and $F_{col}$ from \emph{many} matrices.

\begin{restatable}{mylemma}{gridgametheoremaaronson}
\label{thm:grid-game-theorem-aaronson}
Using \cref{thm:aaronson_general_lb_original} gives a lower bound of $\Omega(1)$ on the randomized query complexity of the matrix game.
\end{restatable}
\begin{proof}
We instantiate \cref{thm:aaronson_general_lb_original} with the next parameters: 
\begin{itemize}
    \item The finite set $A$ is the set of $n$ coordinates within $\sqrt{n} \times \sqrt{n}$ matrix.
    \item The finite set $B$ is $\{0,1,2\}$.
    \item The set $\cX \subseteq B^A$ consists of $\sqrt{n}$ row matrices and $\sqrt{n}$ column matrices, for a total of $2 \sqrt{n}$ matrices.
    So, for any given $F \in \cX$ and $a \in A$, we have that $F(a) \in \{0,1,2\}$ is the value of the matrix at the coordinate indicated by $a$.
    \item Mapping $\cH : \cX \to \{0,1\}$ refers to deciding whether the matrix from $\cX$ is a row or column matrix: output 0 if ``row'' and output 1 if ``column''.
\end{itemize}

Consider an arbitrary choice of $r : \cA \times \cB \to \mathbb{R}_{\ge 0}$.
Consider an arbitrary row matrix $F_{row} \in \cA$ and an arbitrary column matrix $F_{col} \in \cB$ such that $r(F_{row},F_{col}) > 0$. 
Let $a_{int}$ be the unique intersecting coordinate in the matrix within $F_{row}$'s non-zero row and $F_{col}$'s non-zero column.
That is, $F_{row}(a_{int}) = 1$ and $F_{col}(a_{int}) = 2$.
Meanwhile, $F(a_{int}) = 0$ for any other matrix $F \in \cA$ or $F \in \cB$.
Therefore,
\begin{align*}
\theta(F_{row},a_{int}) &= \frac{\sum_{F_3 \in \cB \;:\; F_{row}(a_{int}) \ne F_3(a_{int})} r(F_{row},F_3)}{\sum_{F_3 \in \cB} r(F_{row},F_3)}
= \frac{\sum_{F_3 \in \cB} r(F_{row},F_3)}{\sum_{F_3 \in \cB} r(F_{row},F_3)}
= 1\,.
\end{align*}
Similarly, we have $\theta(F_{col},a_{int}) = 1$ and we see that
\[
    \min\{\theta(F_{row}, a_{int}), \theta(F_{col},a_{int})\} = 1\,.
\]
Since $v_{min}$ is a maximum over choices of $F_{row},F_{col},$ and $v$, this suffices to show
\[
    v_{min} \ge 1\,.
\]
Thus $1/v_{min} \leq 1$ and since the choice of $r$ was arbitrary, one cannot hope to show a stronger lower bound than $\Omega(1)$ via \cref{thm:aaronson_general_lb_original} for this problem.
\end{proof}

\subsection{Our variant is stronger than the original relational adversary method}

In fact, we now show that our new variant \cref{thm:our-variant} is always at least asymptotically as good as \cref{thm:aaronson_general_lb_original}.

\stronger*

\begin{proof}


Recall that for every $a \in A$ and $F_1,F_2 \in \cX$, we define $r_a(F_1,F_2) = r(F_1,F_2) \cdot \mathbbm{1}_{\{F_1(a) \ne F_2(a)\}}$.
\cref{thm:aaronson_general_lb_original} provides a lower bound of $\Lambda = 1/(5 \cdot v_{min})$ for the expected number of queries issued by a randomized algorithm that succeeds with probability at least $9/10$, where
\[
    v_{min} = \max_{F_1 \in \cA, F_2 \in \cB, a \in A \;:\; r(F_1,F_2) > 0, F_1(a) \ne F_2(a)} \min \{\theta(F_1,a), \theta(F_2,a)\}\,.
\]
Then for all $F_1 \in \cA, F_2 \in \cB, a\in A$ where $r_a(F_1,F_2) > 0$, we have
\begin{equation}
    \frac{1}{\min\{\theta(F_1,a), \theta(F_2,a)\}} \ge \frac{1}{v_{min}}\,.
\end{equation}
By rearranging and invoking the definition of $\theta$ (\cref{eq:theta}) we get
\begin{equation} \label{eq:aaronson_consequence_-1}
    \max\left( \frac{\sum_{F_3 \in \cB} r(F_1,F_3)}{\sum_{F_3 \in \cB} r_a(F_1,F_3)}, \frac{\sum_{F_3 \in \cA} r(F_3, F_2)}{\sum_{F_3 \in \cA} r_a(F_3,F_2)} \right) \ge \frac{1}{v_{min}}\,. 
\end{equation}
By extending $r$ to $\cX \times \cX$, with $r(F_1,F_2) = 0$ if both $F_1,F_2 \in \cA$ or both $F_1,F_2 \in \cB$, we can write \cref{eq:aaronson_consequence_-1} as:
\begin{equation} \label{eq:aaronson_consequence}
    \max\left( \frac{M(\{F_1\})}{\sum_{F_3 \in \cX} r_a(F_1,F_3)}, \frac{M(\{F_2\})}{\sum_{F_3 \in \cX} r_a(F_2,F_3)} \right) \ge \frac{1}{v_{min}}\,. 
\end{equation}
Keeping the same choice of $r$, take an arbitrary choice of $\cZ \subseteq \cX$ and $a \in A$ such that $q(\cZ) > 0$.
Such a choice exists by the following argument: \cref{thm:aaronson_general_lb_original} provided a bound, so there must exist $F_1 \in \cA$ and $F_2 \in \cB$ and $a \in A$ with $r(F_1,F_2) > 0$ and $F_1(a) \ne F_2(a)$; the set $\cZ = \{F_1,F_2\}$ has $q(\cZ) > 0$.
Let $\cC \subseteq \cZ$ be the subset of functions $F_1 \in \cZ$ defined as follows:
\[
    \cC = \left\{F_1 \in \cZ : M(\{F_1\}) \ge \frac{1}{v_{min}} \cdot \sum_{F_3 \in \cX} r_a(F_1,F_3) \right\} \,.
\]

By \cref{eq:aaronson_consequence}, we know that for every pair of functions $F_1,F_2 \in \cZ$ with $r_a(F_1,F_2) > 0$, at least one of $F_1$ and $F_2$ is in $\cC$.
Therefore 
\begin{equation}
\label{eq:q-bound}
\sum_{F_1 \in \cC} \sum_{F_2 \in \cZ} r_a(F_1,F_2)\ge \frac{1}{2} \sum_{F_1 \in \cZ} \sum_{F_2 \in \cZ} r_a(F_1,F_2) = \frac{q(\cZ)}{2}\,.
\end{equation}
Meanwhile, $M(\cZ) \ge M(\cC)$.
Therefore
\begin{align*}
    \frac{M(\cZ)}{q(\cZ)} &\ge \frac{1}{2} \cdot \frac{M(\cC)}{\sum_{F_1 \in \cC} \sum_{F_2 \in \cZ} r_a(F_1,F_2)} \explain{By \cref{eq:q-bound}}\\
        & \ge \frac{1}{2} \cdot \frac{ \sum_{F_1 \in \cC} M(\{F_1\}) } { \sum_{F_1 \in \cC} \sum_{F_2 \in \cX} r_a(F_1,F_2) } \,. \explain{since $\cZ \subseteq \cX$}
\end{align*}
By an averaging argument, there exists a function $F_1 \in \cC$ such that 
\begin{align} \label{eq:average_argument_F_1_in_C}
    \frac{M(\cZ)}{q(\cZ)} & \ge \frac{1}{2} \cdot \frac{ M(\{F_1\}) } { \sum_{F_2 \in \cX} r_a(F_1,F_2) }\,.
\end{align}
But then since $F_1 \in \cC$ we get
\begin{align} \label{eq:M_F_1_over_sum_F_2_in_X_at_least_1_over_v_min}
    \frac{ M(\{F_1\}) } { \sum_{F_2 \in \cX} r_a(F_1,F_2) } & \ge \frac{1}{v_{min}}\,.
\end{align}
Combining \cref{eq:average_argument_F_1_in_C} and \cref{eq:M_F_1_over_sum_F_2_in_X_at_least_1_over_v_min}, 
\begin{align} \label{eq:M_over_q_at_least_half_v_min}
    \frac{M(\cZ)}{q(\cZ)} \geq  \frac{1}{2} \cdot \frac{ M(\{F_1\}) } { \sum_{F_2 \in \cX} r_a(F_1,F_2) } \geq \frac{1}{2 \cdot v_{min}}\,.
\end{align}
Since \cref{eq:M_over_q_at_least_half_v_min} holds for an arbitrary choice of $\cZ$ and $a$, it follows that  \cref{thm:our-variant} shows that a randomized algorithm that succeeds with probability at least $9/10$ issues at least  
${1}/{(200 \cdot v_{min})} 
$
queries in expectation.
Since $\Lambda = 1/(5 \cdot v_{min})$, the lower bound given by \cref{thm:our-variant} is at least $\Lambda/40$, which completes the proof.
\end{proof}

\newpage 

\section{Valid functions have a unique local minimum}
\label{sec:appendix-valid-functions-have-unique-local-minimum}

In this section, we define an abstract class of functions called ``valid functions''. Such functions  have a unique local minimum and will be used in both the congestion and separation lower bounds.

\begin{restatable}[Valid function]{mydefinition}{validfunctiondefinition}
\label{def:valid_function}
Let $W = (w_1, \ldots, w_s)$ be a walk in the graph $G$.
A function $f$ is \emph{valid} with respect to the walk $W$ if it satisfies the next conditions:
\begin{enumerate}
    \item For all $u,v \in W$, if $\max\{i \in [s] \mid v = w_i\} < \max\{i \in [s] \mid u = w_i\}$, then $f(v) > f(u)$.
    
    {In other words, as one walks along the walk $W$ starting from $w_1$ until $w_s$, if the last time the vertex $v$ appears is before the last time that vertex $u$ appears, then $f(v) > f(u)$.}
    \item For all $v \in V \setminus W$, we have $f(v) = dist(w_1, v) > 0$.
    %
    \item $f(w_i) \leq 0$ for all $i \in [s]$.
\end{enumerate}
\end{restatable}
An illustration of a valid function is shown in  Figure~\ref{fig:validfunction}.

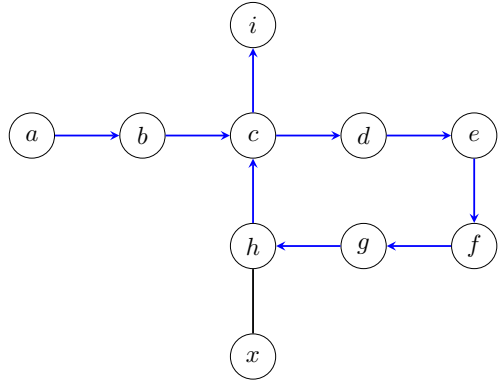
\begin{figure}[h]
\centering
\resizebox{0.4\linewidth}{!}{%
\begin{tikzpicture}
\node[draw, circle, minimum size=20pt, inner sep=2pt] at (0,0) (c) {$c$};
\node[draw, circle, minimum size=20pt, inner sep=2pt, left=of c] (b) {$b$};
\node[draw, circle, minimum size=20pt, inner sep=2pt, left=of b] (a) {$a$};
\node[draw, circle, minimum size=20pt, inner sep=2pt, right=of c] (d) {$d$};
\node[draw, circle, minimum size=20pt, inner sep=2pt, right=of d] (e) {$e$};
\node[draw, circle, minimum size=20pt, inner sep=2pt, above=of c] (i) {$i$};
\node[draw, circle, minimum size=20pt, inner sep=2pt, below=of c] (h) {$h$};
\node[draw, circle, minimum size=20pt, inner sep=2pt, right=of h] (g) {$g$};
\node[draw, circle, minimum size=20pt, inner sep=2pt, right=of g] (f) {$f$};
\node[draw, circle, minimum size=20pt, inner sep=2pt, below=of h] (x) {$x$};

\draw[blue, thick, -stealth] (a) -- (b);
\draw[blue, thick, -stealth] (b) -- (c);
\draw[blue, thick, -stealth] (c) -- (d);
\draw[blue, thick, stealth-] (c) -- (h);
\draw[blue, thick, -stealth] (c) -- (i);
\draw[blue, thick, -stealth] (d) -- (e);
\draw[blue, thick, -stealth] (e) -- (f);
\draw[blue, thick, -stealth] (f) -- (g);
\draw[blue, thick, -stealth] (g) -- (h);
\draw[thick] (h) -- (x);

\end{tikzpicture}
}
\caption{
Consider the walk $W = (a,b,c,d,e,f,g,h,c,i)$ in blue, where $f(a), \ldots, f(i) \leq 0$.
Observe that $w_5 = e$ and $w_9 = c$ so $f(e) > f(c)$.
Since the vertex $x$ is \emph{not} in the walk, $f(x) = dist(w_1,x) = dist(a,x) = 4$.
}
\label{fig:validfunction}
\end{figure}


Next we prove every valid function has a unique local minimum.

\begin{restatable}{mylemma}{validimpliesuniquelocalminimum}
\label{lem:valid_implies_unique_local_minimum}
Suppose $W = (w_1, \ldots, w_s)$ is a  walk in $G$ and $f : V \to \mathbb{R}$ is a valid function for the walk $W$. Then $f$ has a unique local minimum at $w_s$, the last vertex on the walk.
\end{restatable}
\begin{proof}
Let $v \in V$ be an arbitrary vertex. We consider two cases:

\paragraph{Case 1: $v \notin W$.} By condition $2$ in the definition of a valid function, we have 
\[f(v) = dist(w_1, v) > 0\,.\]
Let $u$ be the neighbor of $v$ on a shortest path from $v$ to $w_1$, breaking ties lexicographically if multiple such neighbors exist. We have two subcases:
\begin{enumerate}[(a)] 
\item If $u \in W$, then $f(u) \le 0 < f(v)$.
\item If $u \notin W$, then $f(v) = f(u) + 1 > f(u)$.
\end{enumerate}
In both subcases, vertex $v$ has a neighbour with a smaller value, so $v$ is not a local minimum of $f$.
\paragraph{Case 2: $v \in W$.}
Let $i = \max\{j \in [s] \mid  v = w_j\}$ be the last index where vertex $v$ appears on the walk $W$. We have two subcases:
\begin{enumerate}[(a)]
\item If $i < s$, let $u = w_{i+1}$ be the next vertex along the walk.
By maximality of the index $i$, the walk $W$ does not visit vertex $v$ anymore in $(w_{i+1}, \ldots, w_s)$. Since $v = w_i$, condition $1$ of a valid function implies  $$f(v) = f(w_i)  > f(w_{i+1}) = f(u),$$ since $\max\{j \in [s] \mid  u = w_j\} \ge i+1 > i = \max\{j \in [s] \mid v = w_j\}$. Thus vertex $w_i$ is not a local minimum.
\item If $i = s$, then $v = w_s$.
By the analysis in the previous cases (1.a, 1.b, and 2.a), each  vertex $u \in [n] \setminus \{w_s\}$ has a neighbor with a strictly smaller value than $u$, and so $u$ cannot be a local minimum. 
Since the graph $G$ must have at least one local minimum, at the global minimum, it follows that $w_s$ is the unique local and global minimum of $f$.
\end{enumerate}
\end{proof}
\newpage 

\section{Lower bound for local search via congestion}
\label{sec:appendix-congestion}

In this section we prove the lower bound of $\Omega\left(n^{1.5}/g\right)$, where $g$ is the vertex congestion.
We start with  the basic definitions. Then we state and prove the lower bound. Afterwards, we show the helper lemmas used in the proof of the theorem.


%
%

\subsection{Basic definitions for congestion}

Recall we have a graph $G = ([n], E)$ with vertex  congestion $g$.  This means there exists an all-pairs set of paths $\mathcal{P} = \left\{P^{u,v}\right\}_{u,v \in [n]}$ with vertex congestion $g$ \footnote{That is, each vertex is used at most $g$ times across all the paths.}, but no such set of paths exists for $g-1$. We fix the set $\mathcal{P}$, requiring  $P^{u,u} = (u)$ $\forall u \in [n]$. 

For each $u,v \in [n]$, let $\numPaths{v}(u)$ be the number of paths in $\mathcal{P}$ that start at vertex $u$ and contain $v$:
    \begin{equation} \label{eq:numPaths}
        \numPaths{v}(u) = \Bigl|\{P^{u,w} \in \cP : w \in [n], v \in P^{u,w}\} \Bigr|  \,. 
    \end{equation}

Let $L \in [n]$, with $L \geq 2$, be a parameter that we set later.

Given a sequence of $k$ vertices $\X = (x_1, \ldots, x_k)$, we write $\X_{1 \to j}=(x_1,\ldots,x_j)$ to refer to a prefix of the sequence, for an index $j \in [k]$.

Given a walk $Q = (v_1, \ldots, v_k)$ in $G$, let $Q_i$ refer to the $i$-th vertex in the walk  (i.e. $Q_i = v_i$). For each vertex $u \in [n]$, let $\multiplicity(Q,u)$ be the number of times that vertex $u$ appears in $Q$.





\begin{mydefinition}[Staircase]
\label{def:staircase}
Given a sequence  $\X = (x_1, \ldots, x_{k})$ of vertices in  $G$, a \emph{staircase} induced by $\X$ is a walk $S_{\X} = S_{\X,1} \circ \ldots \circ S_{\X, k-1}$, where each $S_{\X, i}$ is a path in $G$ starting at vertex $x_{i}$ and ending at $x_{i+1}$. 
Each vertex $x_i$ is called a \emph{milestone} and each path $S_{\X,i}$ a \emph{quasi-segment}.

The staircase $S_{\vec{x}}$ is said to be \emph{induced by   $\vec{x}$ and  $\mathcal{P} = \left\{P^{u,v}\right\}_{u,v \in [n]}$} if additionally we  have $S_{\X,i} = P^{x_i, x_{i+1}}$ for all $i \in [k-1]$.
\end{mydefinition}

 \begin{mydefinition}[Tail of a staircase] \label{def:tail}
    Let $S_{\vec{x}} = S_{\vec{x},1} \circ \ldots \circ S_{\vec{x}, k-1}$  be a staircase induced  by some sequence $\X  = (x_1, \ldots, x_k) \in [n]^k$.
    For each $j \in [k-1]$, let $T = S_{\vec{x},j} \circ \ldots \circ S_{\vec{x}, k-1}$.
    Then 
   $Tail(j, S_{\X})$ is obtained from $T$ by removing the first occurrence of $x_j$ in $T$ (and only the first occurrence).
   We also define $Tail(k, S_{\X})$  to be the empty sequence. 
\end{mydefinition}


Next we define the set of functions, $\mathcal{X}$, that will be used when invoking Theorem~\ref{thm:our-variant}.

\begin{mydefinition}[The functions $f_{\vec{x}}$ and $g_{\vec{x},b}$; the  set $\mathcal{X}$]
\label{def:mathcal_X}
 Suppose $\mathcal{P} =  \{P^{u,v} \}_{u,v \in [n]} $ is an all-pairs set  of  paths in $G$. For each sequence  of vertices  $\X \in \{1\} \times [n]^{L}$, define a function \\ $f_{\X} : [n] \to  \{-n^2-n, \ldots, 0,\ldots,n\}$  such that for each $v \in [n]$:
    \begin{itemize}
        \item If $v \notin S_{\X}$, then  set $f_{\X}(v) = dist(v,1)$, where $S_{\X}$ is the staircase induced by $\X$ and $\cP$.
        \item If $v \in S_{\X}$, then set $f_{\X}(v) = -i\cdot n - j$, where $i$ is the maximum index with $v \in P^{x_i,x_{i+1}}$, and $v$ is the $j$-th vertex in $P^{x_i,x_{i+1}}$.
    \end{itemize}
Also, for each $\X \in \{1\} \times [n]^L$ and $b \in \{0,1\}$, let $g_{\X,b} : [n] \to \{-n^2-n, \ldots, 0,\ldots,n\} \times \{-1,0,1\}$ be such that, for all $v \in [n]$:
\begin{align} 
    g_{\X,b}(v) = \begin{cases}
            \bigl(f_{\X}(v), b\bigr) & \text{if} \; v = x_{L+1}\\
            \bigl(f_{\X}(v), -1\bigr) & \text{if}  \; v \ne x_{L+1}
        \end{cases} \,. \label{eq:def:g_X_b}
\end{align} 
Let $\mathcal{X} = \Bigl\{g_{\X,b} \mid  \X \in \{1\} \times [n]^{L} \mbox{ and }  b \in \{0,1\}\Bigr\}$.
\end{mydefinition}

We will show later that for each sequence $\vec{x}$, the function $f_{\vec{x}}$ has a unique local minimum at the end of the staircase $S_{\vec{x}}$.

An example of a staircase $S_{\vec{x}}$ for some sequence of vertices $\vec{x}$ is shown in the next figure, together with  the accompanying value function $f_{\vec{x}}$.
\begin{figure}[h!]
\centering
\includegraphics[scale=1.1]{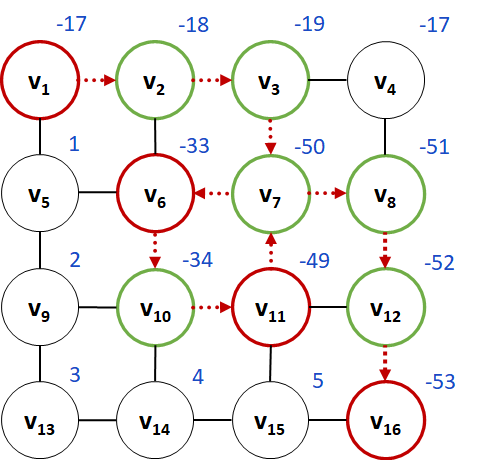}
\caption{Example of a staircase with the accompanying value function. The sequence of milestones is $(v_1, v_6, v_{11}, v_{16})$, which are shown in red. The vertices of the staircase are shown in red and green vertices, connected by red dotted edges. For each node $v$, the value of the function at $v$ is shown in blue.}
\label{fig:fx_congestion}
\end{figure}
\begin{example}
Let $G$ be the grid graph on $n = 16$ nodes from \cref{fig:fx_congestion}. Consider the sequence of vertices
\[
\X = (x_1, x_2, x_3, x_4) = (v_1, v_6, v_{11}, v_{16}),
\]
We fix an all-pairs set of paths $\mathcal{P} = \{ P^{u,v} \}_{u,v\in [n]}$, such that 
\begin{itemize} 
\item $P^{v_1,v_6} = (v_1, v_2, v_3, v_7, v_6)$; $P^{v_6, v_{11}} = (v_6, v_{10}, v_{11})$; $P^{v_{11},v_{16}} = (v_{11}, v_7, v_8, v_{12}, v_{16})$, 
where \\ $P^{v_1,v_6}, P^{v_6,v_{11}}, P^{v_{11},v_{16}} \in \cP$.
\item For each other pair of vertices $(u,w)$, we set $P^{u,w}$ as the shortest path between $u$ and $w$, breaking ties lexicographically (vertices with lower index come first).
\end{itemize}
Then the staircase induced by $\X$ and $\cP$ is 
\[ 
S_{\X} = P^{v_1,v_6} \circ P^{v_6,v_{11}} \circ P^{v_{11},v_{16}} = (v_1, v_2, v_3, v_7, v_6, v_{10}, v_{11}, v_7, v_8, v_{12}, v_{16})\,.
\]

For example, $f(v_4) = dist(v_1, v_4) = 3$ since $v_4 \not\in S_{\X}$, and $f_{\X}(v_{7}) = -3 n - 2 = -50$ since $v_{7}$ is the second node in $P^{x_3, x_4} = P^{v_{11},v_{16}}$ (even though $v_7$ is also included in the path $P^{x_1, x_2}$).
\end{example}

\begin{mydefinition}[The map $\mathcal{H}$] \label{def:map_H}
 Suppose $\mathcal{P} =  \{P^{u,v} \}_{u,v \in [n]} $ is an all-pairs set  of  paths in $G$ and $\mathcal{X}$ is the set of functions $g_{\vec{x},b}$ from Definition~\ref{def:mathcal_X}. Define $\mathcal{H} : \mathcal{X} \to \{0,1\}$ as 
 \[
\mathcal{H}(g_{\vec{x},b}) = b \qquad  \forall \vec{x} \in \{1\} \times [n]^{L} \; \text{and} \; b \in \{0,1\} \,. 
\]
 \end{mydefinition}

To define a suitable $r$ function, we only concern ourselves with pairs of staircases that do not have repeated milestones within themselves and with different hidden bits. For such a pair of staircases, we assign a relative difficulty of distinguishing them that scales with the length of their common prefix: the longer their common prefix, the more function values they agree on, and so we assign a higher $r$ value.

\begin{mydefinition}[Good/bad sequences of vertices; Good/bad functions] \label{def:good}
A sequence of $k$ vertices $\X = (x_1, \ldots, x_k)$ is  \emph{good} if  $x_i \ne x_j$ for all $i,j$ with $1 \le i < j \le k$; otherwise, $\X$ is \emph{bad}.
Moreover, for each $b \in \{0,1\}$ a function $F = g_{\X,b} \in \mathcal{X}$ is good if $\X$ is good, and \emph{bad} otherwise.
\end{mydefinition}



\begin{mydefinition} [The function $r$]\label{def:r}
    Let $r : \mathcal{X} \times \mathcal{X} \to \mathbb{R}_{\ge 0}$  be a symmetric function defined as follows. For  each \mbox{$\X, \Y \in \{1\} \times [n]^L$} and $b_1,b_2 \in \{0,1\}$, we have 
    \[
        r(g_{\X,b_1}, g_{\Y,b_2}) = \begin{cases}
            0 & \text{ if at least one of the following holds: } b_1 = b_2 \text{ or } \X \text{ is bad or } \Y \text{ is bad.}\\
            n^j & \text{ otherwise, where } j \text{ is the maximum index for which } \X_{1 \to j} = \Y_{1 \to j}\,. 
        \end{cases}
    \]
\end{mydefinition}


\begin{example}
Suppose $L = 3$. Consider the graph $G = ([n], E)$ and consider the sequences of vertices $\X = (1,2,3,4)$, $\Y = (1,3,5,3)$, and $\Z = (1,2,5,4)$.
Since $\Y$ has repeated elements while $\X$ and $\Z$ do not, we see that $\X$ is good, $\Y$ is bad, and $\Z$ is good.
Then, for each $b \in \{0,1\}$, we have
\begin{itemize}
    \item $r(\cdot, g_{\Y, b}) = r(g_{\Y, b}, \cdot) = 0$ since $\Y$ is bad.
    \item $r(g_{\X,b}, g_{\Z,1-b}) = r(g_{\Z,1-b}, g_{\X,b}) = n^2$ since $\X$ is good, $\Z$ is good, $\X_{1 \to 2} = \Z_{1 \to 2}$, and $x_3 \neq z_3$.
\end{itemize}
\end{example}

The function $r$ will be used directly to invoke \cref{thm:our-variant}, but we also  define here some related helper functions to use $r$ in conjunction with certain indicator variables.

\begin{mydefinition}[The function $r_v$] \label{def:rv}
    For each $v \in [n]$, define $r_v : \mathcal{X} \times \mathcal{X} \to \mathbb{R}_{\geq 0}$ as follows: 
    \[
        r_v(F_1,F_2) = r(F_1,F_2) \cdot \mathbbm{1}_{\{F_1(v) \ne F_2(v)\}} \qquad \forall F_1, F_2 \in \mathcal{X} \,.
    \]
\end{mydefinition}
\begin{mydefinition}[The function $\widetilde{r}_v$] \label{def:rv_tilde}
    For each $v \in [n]$, define $\widetilde{r}_v : \mathcal{X} \times \mathcal{X} \to \mathbb{R}_{\geq 0}$ as follows:
    \[
        \widetilde{r}_v(g_{\X,b_1}, g_{\Y,b_2}) = r_v(g_{\X,b_1}, g_{\Y,b_2}) \cdot \mathbbm{1}_{\{\multiplicity(S_{\X},v) \le \multiplicity(S_{\Y},v)\}} \qquad \forall \X,\Y \in \{1\} \times [n]^L \; \; \forall b_1,b_2 \in \{0,1\} \,.
    \]
\end{mydefinition}

\begin{observation}
\label{obs:r_v_tilde_equal_b}
For each $\vec{x}, \vec{y} \in \{1\} \times [n]^L$ and $b \in \{0,1\}$, we have 
\begin{align}
\widetilde{r}_v(g_{\X, b}, g_{\Y, b})  = {r}_v(g_{\X, b}, g_{\Y, b})  = {r}(g_{\X, b}, g_{\Y, b})  = 0\,.
\end{align}
\end{observation}
\begin{proof}
By definition of $r$, we have $r(g_{\X, b_1}, g_{\Y, b_2}) = 0$ when $b_1 = b_2$. Then for all $b \in \{0,1\}$:
\begin{align}
    \widetilde{r}_v(g_{\X, b}, g_{\Y, b}) &= r_v(g_{\X, b}, g_{\Y, b}) \cdot \mathbbm{1}_{\{\multiplicity(S_{\X},v) \leq \multiplicity(S_{\Y},v)\}} \notag \\
    &= r(g_{\X, b}, g_{\Y, b}) \cdot \mathbbm{1}_{\{g_{\X,b}(v) = g_{\Y,b}(v)\}} \cdot \mathbbm{1}_{\{\multiplicity(S_{\X},v) \le \multiplicity(S_{\Y},v)\}} = 0 \,. 
\end{align}
\end{proof}

%
%
%

\subsection{Proof of the congestion lower bound}

In this section we include the proof of \cref{thm:low_congestion_implies_local_search_hard}. The proofs of lemmas used in the theorem are included afterwards, in Section~\ref{sec:helper_lemmas_congestion}.

\lowcongestionimplieslocalsearchhard*
\begin{proof}
Consider the following setting of parameters:
\begin{enumerate}[(a)]
    \item  $L = \lfloor \sqrt{n}\rfloor -1$.
    \item   Fix an all-pairs set of paths $\cP = \left\{P^{u,v}\right\}_{u,v \in [n]}$ for $G$, such that $\mathcal{P}$ has vertex congestion $g$.
    \item The finite set $A$ is the set of vertices $[n]$.
    \item The finite set $B$ is $\{-n^2 - n, \ldots, 0, \ldots, n\} \times \{-1,0,1\}$.
    \item The functions $f_{\vec{x}}$, $g_{\vec{x},b}$, and the set $\mathcal{X}$ given by \cref{def:mathcal_X}. \emph{(Recall  $g_{\X,b}(v) = (f_{\X},c)$  for all $v \in [n]$, where $c = -1$ 
 if $v \neq x_{L+1}$, and $c =b$ if $v = x_{L+1}$, i.e. $c =b$ if and only if $v$ is a local minimum of $f_{\vec{x}}$. Also,   
$\cX = \{ g_{\vec{x}, b} \mid \X \in \{1\} \times [n]^{L} \text{ and } b \in \{0,1\}\}\,.$)}
    \item Map $\cH: \cX \to \{0,1\}$ as  in \cref{def:map_H}. \emph{(Recall $\mathcal{H}(g_{\vec{x}, b}) = b$  for all $  \vec{x} \in \{1\} \times [n]^L$ and $b \in \{0,1\}$.)}
    \item The function $r$ as  in \cref{def:r}.
\end{enumerate}

By \cref{lem:proof_f_X_is_valid}, each function  $f_{\X}$ is valid for all $\vec{x} \in \{1\} \times [n]^L$, so \cref{lem:valid_implies_unique_local_minimum} implies that each function $f_{\vec{x}}$ has a unique local minimum (at $x_{L+1}$).
Therefore by \cref{lem:reduction_to_decision} invoked with $f=f_{\X}$ and $h_b = g_{\X,b}$, it suffices to show a lower bound for the corresponding decision problem: return the hidden bit $b \in \{0,1\}$ given oracle access to the function $g_{\X,b}$.

For each $\cZ \subseteq \cX$, let  
\begin{align}  \label{eq:remind_Z_def}
M(\cZ) = \sum_{F_1 \in \cZ} \sum_{F_2 \in \cX} r(F_1, F_2)\,.  
\end{align}
By \cref{lem:parameters_to_new_aaronson_are_ok}, there exists a subset $\cZ \subseteq \cX$ with
$q(\cZ) > 0$.
Thus the conditions required by  \cref{thm:our-variant} are met. 
By invoking \cref{thm:our-variant} with the parameters in (a-g), we get that the randomized query complexity of the decision problem, and thus also of local search on $G$, is 
\begin{align} 
\Omega\left(\min_{\cZ \subseteq \cX: q(\cZ) > 0}
\frac{M(\cZ)}{q(\cZ)}\right), \text{ where } q(\cZ) = \max_{v \in [n]}{\sum_{F_1 \in \cZ} \sum_{F_2 \in \cZ} r(F_1,F_2) \cdot \mathbbm{1}_{\{F_1(v) \neq F_2(v)\}}}\,. \notag 
\end{align} 

To get an explicit lower bound in terms of congestion, we will upper bound $q(\cZ)$ and lower bound $M(\cZ)$ for  subsets $\cZ \subseteq \cX$ with $q(\cZ) > 0$.

\medskip 

Fix an arbitrary subset $\cZ \subseteq \cX$ with $q(\cZ) > 0$.
Since $r(F_1,F_2) = 0$ {when $F_1$ or $F_2$ is bad}, it suffices to consider subsets $\cZ \subseteq \cX$ where each function $F \in \cZ$ is good. 


\paragraph{Upper bounding $q(\cZ)$.}

Let $v \in [n]$ be arbitrary. 

Fix an arbitrary function $F_1 \in \mathcal{Z}$. Since $F_1$ is good, there exist $\vec{x} \in \{1\} \times [n]^L$ and $b_1 \in \{0,1\}$ such that $F_1 = g_{\vec{x}, b_1}$ and $\vec{x}$ is good.
Since $\cZ \subseteq \cX$ and $\widetilde{r}_v \geq 0$, we have 
\begin{align} \label{eq:simple_inequality_tilde_r_v_F_1_F_2_sum_over_X}
\sum_{F_2 \in \mathcal{Z}} \widetilde{r}_v(F_1,F_2) 
& \leq \sum_{F_2 \in \mathcal{X}} \widetilde{r}_v(F_1,F_2) \,. 
\end{align}
Using the definition of $\mathcal{X} = \{ g_{\Y,b_2} \mid \Y \in \{1\} \times [n]^L, \; b_2 \in \{0,1\}\}$, the fact that $F_1 = g_{\vec{x},b_1}$, and partitioning the space of functions $F_2 \in \mathcal{X}$ by the length of the prefix that the staircase corresponding to  $F_2$  shares with the staircase corresponding to $F_1$, we can upper bound the right hand  side of \cref{eq:simple_inequality_tilde_r_v_F_1_F_2_sum_over_X}:
\begin{align} \label{eq:simple_ineq_r_tilde_v_F_1_F_2_in_X_part2}
\sum_{F_2 \in \mathcal{X}} \widetilde{r}_v(F_1,F_2) & = \sum_{\substack{\Y \in \{1\} \times [n]^L, \; b_2 \in \{0,1\}}} \widetilde{r}_v(g_{\X,b_1},g_{\Y,b_2}) \leq 
\sum_{j=1}^{L+1} \sum_{\substack{\Y \in \{1\} \times [n]^L, \; b_2 \in \{0,1\}\\ j = \max\{i \;:\; \X_{1 \to i} = \Y_{1 \to i}\}}} \widetilde{r}_v(g_{\X,b_1},g_{\Y,b_2}) \,. 
\end{align}

Combining \cref{eq:simple_inequality_tilde_r_v_F_1_F_2_sum_over_X} and \cref{eq:simple_ineq_r_tilde_v_F_1_F_2_in_X_part2}, we get

\begin{align} \label{eq:simple_ineq_r_tilde_v_F_1_F_2_in_X_part2_prime}
\sum_{F_2 \in \mathcal{Z}} \widetilde{r}_v(F_1,F_2) \leq \sum_{j=1}^{L+1} \sum_{\substack{\Y \in \{1\} \times [n]^L, \; b_2 \in \{0,1\}\\ j = \max\{i \;:\; \X_{1 \to i} = \Y_{1 \to i}\}}} \widetilde{r}_v(g_{\X,b_1},g_{\Y,b_2}) \,.
\end{align}

For each $j \in [L+1]$, let 
\[ 
T_j = \left\{\Y \in \{1\} \times [n]^L \mid \max\{i : \X_{1 \to i} = \Y_{1 \to i} \} = j \text{ and } v \in Tail(j,S_{\Y}) \right\} \,.
\]
Then for each $j \in [L]$, we can bound the part of the sum in  \cref{eq:simple_ineq_r_tilde_v_F_1_F_2_in_X_part2_prime} corresponding to index $j$ via the next chain of inequalities:
\begin{align} 
& \sum_{\substack{\Y \in \{1\} \times [n]^L, \; b_2 \in \{0,1\} :\\ j = \max\{i \;:\; \X_{1 \to i} = \Y_{1 \to i}\}}} \widetilde{r}_v(g_{\X,b_1},g_{\Y,b_2}) \leq \sum_{\substack{\Y \in \{1\} \times [n]^L :\\ j = \max\{i \;:\; \X_{1 \to i} = \Y_{1 \to i}\}  \\ v \in Tail(j, S_{\Y}) }} {r}(g_{\X,b_1},g_{\Y,1-b_1}) \explain{By \cref{lem:j_le_ell_implies_v_in_tail} }\\
& \qquad    \qquad \leq n^j \cdot |T_j| \explain{Since $r(g_{\X,b_1},g_{\Y,1-b_1}) \leq n^j$ when $j = \max\{i : \X_{1 \to i} = \Y_{1 \to i} \} $} \\
&  \qquad     \qquad \leq  n^j \cdot \left( \numPaths{v}(x_j) \cdot n^{L - j} + L \cdot g \cdot n^{L - j - 1} \right), \explain{By \cref{lem:upper_bound_set_v_in_Tail_Y}.}
\end{align}
where we recall that $\numPaths{v}(u)$ is the number of paths in $\mathcal{P}$ that start at vertex $u$ and contain $v$.

Using the identity $n^j \cdot \left( \numPaths{v}(x_j) \cdot n^{L - j} + L \cdot g \cdot n^{L - j - 1} \right)= n^{L} \cdot \left( \numPaths{v}(x_j) + \frac{L \cdot g}{n} \right)$, we obtain 
\begin{align} 
\label{eq:sum_for_one_gxb_decomposed_by_j_solved}
\sum_{\substack{\Y \in \{1\} \times [n]^L, \; b_2 \in \{0,1\} :\\ j = \max\{i \;:\; \X_{1 \to i} = \Y_{1 \to i}\}}} \widetilde{r}_v(g_{\X,b_1},g_{\Y,b_2})
 \leq n^{L} \cdot \left( \numPaths{v}(x_j) + \frac{L \cdot g}{n} \right) \,.
\end{align}

When $j = L+1$, since $\widetilde{r}_v(g_{\X,b_1}, g_{\Y,b_2}) > 0$ implies  $b_2 = 1 - b_1$ (see Observation \ref{obs:r_v_tilde_equal_b}), we have 
\begin{equation}
\label{eq:sum_for_one_gxb_decomposed_by_ell_plus_one_solved}
\sum_{\substack{\Y \in \{1\} \times [n]^L, \; b_2 \in \{0,1\} :\\ L+1 = \max\{i : \X_{1 \to i} = \Y_{1 \to i}\}}} \widetilde{r}_v(g_{\X,b_1},g_{\Y,b_2}) \leq  n^{L+1}  \,.
\end{equation}

Summing \cref{eq:sum_for_one_gxb_decomposed_by_j_solved} for all  $j \in [L]$ and adding it to  \cref{eq:sum_for_one_gxb_decomposed_by_ell_plus_one_solved} for $j=L+1$, we can now upper bound the right hand side of \cref{eq:simple_ineq_r_tilde_v_F_1_F_2_in_X_part2_prime} as follows: 
\begin{align}
\sum_{F_2 \in \mathcal{Z}} \widetilde{r}_v(F_1,F_2)  & \leq  \sum_{j=1}^{L+1} \sum_{\substack{\Y \in \{1\} \times [n]^L, b_2 \in \{0,1\}:\\ \max\{i : \X_{1 \to i} = \Y_{1 \to i}\} = j }} \widetilde{r}_v(g_{\X,b_1},g_{\Y,b_2}) \explain{By  \cref{eq:simple_ineq_r_tilde_v_F_1_F_2_in_X_part2_prime}} \\
& \leq n^L \cdot \sum_{j=1}^{L} \left( \numPaths{v}(x_j) + \frac{L \cdot g}{n} \right) + n^{L+1} \explain{By \cref{eq:sum_for_one_gxb_decomposed_by_j_solved} and \cref{eq:sum_for_one_gxb_decomposed_by_ell_plus_one_solved}} \\
&   \leq n^L \left( \sum_{u \in [n]} \numPaths{v}(u)\right) + n^{L} \left(\sum_{j=1}^{L} \frac{g L}{n} \right) + n^{L+1} 
\explain{Since $\vec{x}$ is good, i.e. $\vec{x}$ has no repeated vertices}\\
&    \leq n^L \cdot g + n^L \cdot \frac{g \cdot L^2}{n} + n^{L+1} 
\explain{Since $\sum_{u \in [n]} \numPaths{v}(u)$ equals the number of paths in $\cP$ that contain $v$, which is  the congestion at $v$.}\\\
&   \leq 3 \cdot g \cdot n^L\,.  \explain{Since $L \leq \sqrt{n} - 1$ and $g \geq n$}
\end{align}

Thus, for each good function $F_1 \in \mathcal{Z}$, we have  
\begin{align}  \label{eq:good_F_1_sum_over_all_F_2_in_Z}
\sum_{F_2 \in \mathcal{Z}} \widetilde{r}_v(F_1,F_2) \leq  3 \cdot g \cdot n^L\,.
\end{align}

Summing \cref{eq:good_F_1_sum_over_all_F_2_in_Z} over all $F_1 \in \mathcal{Z}$ (each of which is good, since $\mathcal{Z}$ was chosen to have good functions only), and invoking Lemma~\ref{lem:sum_rv_bounded_by_two_sum_rv_tilde} yields 
\begin{align}
\sum_{F_1,F_2 \in \mathcal{Z}}  r_v(F_1, F_2) & \leq 
 2 \cdot \sum_{F_1,F_2 \in \mathcal{Z}} 
 \widetilde{r}_v(F_1, F_2) \explain{By \cref{lem:sum_rv_bounded_by_two_sum_rv_tilde}}\\
& \leq  2 \cdot |\cZ| \cdot 3 \cdot g \cdot n^L \explain{By \cref{eq:good_F_1_sum_over_all_F_2_in_Z}}\\
& = |\cZ| \cdot 6g \cdot n^L \,. \label{eq:sum_F_1_F_2_in_cal_Z_r_v_almost_there_for_q_Z}
\end{align}
Since we had considered an arbitrary vertex $v \in [n]$, taking the maximum over all $v \in [n]$ in \cref{eq:sum_F_1_F_2_in_cal_Z_r_v_almost_there_for_q_Z} yields 
\begin{equation}
\label{eq:numerator_bound}
q(\cZ)
= \max_{v \in [n]} \sum_{F_1 \in \cZ} \sum_{F_2 \in \cZ} r_v(F_1,F_2)
\leq |\cZ| \cdot 6g \cdot n^L\,.
\end{equation}

    \paragraph{Lower bounding $M(\cZ)$.}  
    Since each function $F_1 \in \mathcal{Z}$ is good by choice of $\mathcal{Z}$, 
    \cref{lem:M_large} yields
    \begin{align}  \label{eq:simple_bound_implied_by_Lemma_31_for_F1_good}
    \sum_{F_2 \in \cX} r(F_1,F_2) \geq \frac{1}{2e} \cdot (L+1) \cdot n^{L+1} \qquad \forall F_1 \in \mathcal{Z} \,.
    \end{align}
    Using \cref{eq:simple_bound_implied_by_Lemma_31_for_F1_good} and recalling the definition of $M(\mathcal{Z})$ from \cref{eq:remind_Z_def}, we get 
    \begin{align} 
M(\mathcal{Z}) & = \sum_{F_1 \in \cZ} \sum_{F_2 \in \cX} r(F_1, F_2) \geq \frac{|\cZ| }{2e} \cdot (L+1) \cdot n^{L+1} \,. \label{eq:denominator_bound}
\end{align} 

    \paragraph{Combining the bounds.}
    Combining the bounds from \cref{eq:numerator_bound} and \cref{eq:denominator_bound}, we can now estimate the bound from  \cref{thm:our-variant}:
    \begin{align*}
    \min_{\substack{\cZ \subseteq \cX:\\ q(\cZ) > 0}}
    \frac{M(\cZ)}{q(\cZ)}
    &\geq \frac{ \frac{|\cZ|}{2e} \cdot (L+1)  n^{L+1}}{|\cZ| \cdot 6g \cdot n^L} \explain{By \cref{eq:numerator_bound} and \cref{eq:denominator_bound}}\\
    &\geq \frac{n^{1.5}}{24e \cdot g}  \,. \explain{Since $L+1 = \lfloor \sqrt{n} \rfloor \geq {\sqrt{n}}/{2}$}
    \end{align*}
    
    Therefore, the  randomized query complexity of local search is
    \[
    \Omega\left( \min_{\substack{\cZ \subseteq \cX:\\ q(\cZ) > 0}}
    \frac{M(\cZ)}{q(\cZ)} \right)
    \subseteq \Omega\left( \frac{n^{1.5}}{g} \right)\;.
    \]
    This completes the proof of the theorem.
\end{proof}

\subsection{Helper lemmas} \label{sec:helper_lemmas_congestion}

In this section we prove the helper lemmas that are used in the proof of \cref{thm:low_congestion_implies_local_search_hard}. All the lemmas assume the setup  of the parameters from \cref{thm:low_congestion_implies_local_search_hard}. 

\begin{mylemma} \label{lem:proof_f_X_is_valid}
For each  $\X \in \{1\} \times [n]^L$, the function $f_{\X}$  is valid for the staircase $S_{\X}$ induced by $\X$ and $\mathcal{P}$, where $\mathcal{P} = \{P^{u,v}\}_{u,v \in [n]}$ is the  all-pairs set of paths in $G$.
\end{mylemma}
\begin{proof}
Let $S_{\X} = (w_1, \ldots, w_s)$ be the vertices of the staircase $S_{\X}$ induced by $\X$ and $\mathcal{P}$.
We show that all the three conditions required by the definition of a valid function (\cref{def:valid_function}) hold.

To show the first condition of validity, consider two vertices $v_1,v_2 \in S_{\X}$.
Define 
\[
i_1 = \max\{k \in [L] \mid v_1 \in P^{x_k, x_{k+1}}\}\,.
\]
Define $i_2$ similarly for $v_2$.
Let $j_1$ and $j_2$ be the indices of $v_1$ and $v_2$ in $P^{x_{i_1}, x_{i_1+1}}$ and $P^{x_{i_2}, x_{i_2+1}}$ respectively. Note that $j_1,j_2 \in [n]$.
By definition of the function $f_{\vec{x}}$ (\cref{def:mathcal_X}), we have
\begin{align}
f_{\X}(v_1) = -n \cdot i_1 - j_1 \; \; \text{and} \;  \; 
f_{\X}(v_2) = -n \cdot i_2 - j_2\,. \label{eq:f_at_v1_v2}
\end{align}
Without loss of generality, the last time vertex $v_1$ appears on the path $S_{\X}$, starting from $w_1$ towards $w_s$, is earlier than the last time vertex $v_2$ appears, that is:
\[\max\{k \in [s] \mid v_1 = w_k\} < \max\{k \in [s] \mid v_2 = w_k\} \,.\]
Then $i_1 \le i_2$. We consider two cases:
\begin{itemize} 
\item  If $i_1 = i_2$, then
    \begin{align}
       f_{\X}(v_1) - f_{\X}(v_2) &= j_2 - j_1 \explain{By \cref{eq:f_at_v1_v2}}\\
            &> 0 \explain{Since by assumption $v_1$ last appears on $S_{\X}$ before $v_2$.}
    \end{align}
\item       If  $i_1 < i_2$, then
    \begin{align}
        f_{\X}(v_1) - f_{\X}(v_2) &\ge n + (j_2 - j_1) \explain{By \cref{eq:f_at_v1_v2}}\\
            & > 0 \explain{Since $j_1,j_2 \in [n]$.}
    \end{align}
\end{itemize}
Therefore the first condition of validity is satisfied.

Also by \cref{def:mathcal_X}, of the function $f_{\vec{x}}$, we have that:
\begin{itemize}
\item  $f_{\X}(v) = dist(1,v)$ for all $v \notin S_{\X}$, so the second condition of validity is satisfied.
\item $f_{\X}(v)  \le 0$ for all $v \in S_{\X}$, so the third condition of validity is satisfied.
\end{itemize}
Therefore $f_{\X}$ is valid for the staircase $S_{\X}$ induced by $\X$ and $\cP$.
\end{proof}

\begin{mylemma} \label{lem:parameters_to_new_aaronson_are_ok}
The next two properties hold:
\begin{itemize}
    \item Let $F_1, F_2 \in \mathcal{X}$. Then $r(F_1,F_2) = 0$ when $\mathcal{H}(F_1) = \mathcal{H}(F_2)$. 
    \item There exists a subset $\mathcal{Z} \subseteq \mathcal{X}$ such that 
    \[
    q(\mathcal{Z}) = \max_{v \in [n]} \sum_{F_1 \in \mathcal{Z}} \sum_{F_2 \in \mathcal{Z}} r(F_1,F_2) \cdot \mathbbm{1}_{\{F_1(v) \ne F_2(v)\}} > 0 \;.
    \]
\end{itemize}
\end{mylemma}
\begin{proof}

We first show that $r(F_1,F_2) = 0$ when $\mathcal{H}(F_1) = \mathcal{H}(F_2)$.
To see this, suppose $\mathcal{H}(F_1) = \mathcal{H}(F_2)$ for some functions $F_1,F_2 \in \mathcal{X}$. Then by definition of the set of functions $\mathcal{X}$, there exist  sequences of vertices $\X,\Y \in \{1\} \times [n]^{L}$ and bits $b_1, b_2 \in \{0,1\}$ such that $F_1 = g_{\X,b_1}$ and $F_2 = g_{\Y,b_2}$. By definition of $\mathcal{H}$, we have $\mathcal{H}(g_{\X,b_1}) = b_1$ and $\mathcal{H}(g_{\Y,b_2}) = b_2$. Since $\mathcal{H}(F_1) = \mathcal{H}(F_2)$, we have  $b_1 = b_2$. Then  $r(g_{\X,b_1}, g_{\Y,b_2}) = 0$ by definition of $r$, or equivalently, $r(F_1,F_2) = 0$.

Next we show  there is a subset $\mathcal{Z} \subseteq \mathcal{X}$ with
$q(\mathcal{Z}) > 0$.
To see this, consider two disjoint sets of vertices 
$U_1, U_2 \subset [n]$ such that $U_1 = \{u_{2}^1, \ldots, u_{L+1}^1\}$, $U_2= \{u_2^2, \ldots, u_{L+1}^2\}$,  each vertex $u_j^i$ appears exactly once in $U_i$, and $u_{j}^i \neq 1$ for all $i,j$. 
Such sets $U_1,U_2$ exist since 
$$|U_1| + |U_2| + |\{1\}| = 2 L + 1  = 2 (\lfloor  \sqrt{n} \rfloor - 1) + 1 \leq 2 \sqrt{n} \leq  n \; \; \; \text{for } n \geq 4\,.$$
Form the sequences of vertices $W^1 = (1, u_2^1, \ldots, u_{L+1}^1)$ and $W^2= (1, u_2^2, \ldots, u_{L+1}^2)$. 
Then both $W^1$ and $W^2$ are good. Consider now the functions $g_{W^1, 0}$ and $g_{W^2, 1}$. By definition of $r$, we have $r(g_{W^1, 0}, g_{W^2, 1}) = n$, since the maximum index $j$ for which $W^1_{1\to j} = W^2_{1 \to j} $ is $j=1$. Then 
\[ 
q(\{g_{W^1,0}, g_{W^2,1}\}) \ge r(g_{W^1, 0}, g_{W^2, 1}) \cdot \mathbbm{1}_{\{g_{W^1, 0}(u^1_{L+1}) \ne g_{W^2, 1}(u^1_{L+1})\}} = n > 0 \;.
\]
Thus there exists a subset $\mathcal{Z} \subseteq \mathcal{X}$ with
$q(\mathcal{Z}) > 0$
as required.
\end{proof}

\begin{mylemma} \label{lem:sum_rv_bounded_by_two_sum_rv_tilde}
For each  $v \in [n]$ and subset $\mathcal{Z} \subseteq \mathcal{X}$,
    we have
    \[
        \sum_{F_1, F_2 \in \mathcal{Z} } {r}_v(F_1,F_2) \le 2 \sum_{F_1,F_2 \in \mathcal{Z}} \widetilde{r}_v(F_1,F_2) \,. 
    \]
\end{mylemma}
\begin{proof}
By Definition~\ref{def:rv}, we have  
$r_v(F_1, F_2) = r(F_1, F_2) \cdot \mathbbm{1}_{\{F_1(v) \neq F_2(v)\}}$ for all $v \in [n]$ and $F_1, F_2 \in \mathcal{X}$. Then  $r_v$ is symmetric since both the function $r$ and the  indicator $\mathbbm{1}_{\{F_1(v) \neq F_2(v)\}}$ are symmetric.
Also  recalling that for a walk $Q$,  the number of times that a vertex $u$ appears in  $Q$ is denoted $\multiplicity(Q,u)$, 
we have:
    \begin{align} 
   &  \sum_{F_1, F_2 \in \mathcal{Z} } {r}_v(F_1,F_2)  = 
     \sum_{\substack{F_1,F_2 \in \mathcal{Z}: \\F_1 = g_{\X,b_1}; F_2 = g_{\Y,b_2} \\
    \multiplicity(S_{\X},v) \leq  
    \multiplicity(S_{\Y},v)}} {r}_v(F_1,F_2) 
   + 
    \sum_{\substack{F_1,F_2 \in \mathcal{Z}: \\F_1 = g_{\X,b_1}; F_2 = g_{\Y,b_2} \\
    \multiplicity(S_{\Y},v) < \multiplicity(S_{\X},v)  
    }} {r}_v(F_1,F_2)  \notag \\
    & \qquad  =
     \sum_{\substack{F_1,F_2 \in \mathcal{Z}: \\F_1 = g_{\X,b_1}; F_2 = g_{\Y,b_2} \\
    \multiplicity(S_{\X},v) \leq  
    \multiplicity(S_{\Y},v)}} {r}_v(F_1,F_2) 
    +   
     \left( \sum_{\substack{F_1,F_2 \in \mathcal{Z}: \\F_1 = g_{\X,b_1}; F_2 = g_{\Y,b_2} \\
    \multiplicity(S_{\Y},v) \leq \multiplicity(S_{\X},v) }} {r}_v(F_2,F_1)  -   \sum_{\substack{F_1,F_2 \in \mathcal{Z}: \\F_1 = g_{\X,b_1}; F_2 = g_{\Y,b_2} \\
    \multiplicity(S_{\X},v) =  
    \multiplicity(S_{\Y},v)}} {r}_v(F_2,F_1) \right) \notag \\
     & \qquad  = 2 \sum_{\substack{F_1,F_2 \in \mathcal{Z}: \\F_1 = g_{\X,b_1}; F_2 = g_{\Y,b_2} \\
    \multiplicity(S_{\X},v) \leq  
    \multiplicity(S_{\Y},v)}} {r}_v(F_1,F_2) 
    -
     \sum_{\substack{F_1,F_2 \in \mathcal{Z}: \\F_1 = g_{\X,b_1}; F_2 = g_{\Y,b_2} \\
    \multiplicity(S_{\Y},v) = \multiplicity(S_{\X},v) }} {r}_v(F_2,F_1)\,. \label{eq:r_v_sum_step2}  
    \end{align}

Recall  from  \cref{def:rv_tilde}, of the function $\widetilde{r}_v$, that 
\[
        \widetilde{r}_v(g_{\X,b_1}, g_{\Y,b_1}) = r_v(g_{\X,b_1}, g_{\Y,b_2}) \cdot \mathbbm{1}_{\{\multiplicity(S_{\X},v) \le \multiplicity(S_{\Y},v)\}} \;\;\; \forall \X,\Y \in \{1\} \times [n]^L , \forall b_1,b_2 \in \{0,1\}.
        \]
Then $\widetilde{r}_v(F_1,F_2) = 0$ when $\multiplicity(S_{\X},v) > \multiplicity(S_{\Y}, v)$, which substituted in \cref{eq:r_v_sum_step2} gives

\begin{align} 
\sum_{F_1, F_2 \in \mathcal{Z} } {r}_v(F_1,F_2)
    & = 2 \sum_{F_1,F_2 \in \mathcal{Z}} \widetilde{r}_v(F_1,F_2) -  \sum_{\substack{F_1,F_2 \in \mathcal{Z}: \\F_1 = g_{\X,b_1}; F_2 = g_{\Y,b_2}\\
    \multiplicity(S_{\Y},v) = \multiplicity(S_{\X},v) }} {r}_v(F_2,F_1)  \\
    & \leq 2 \sum_{F_1,F_2 \in \mathcal{Z}} \widetilde{r}_v(F_1,F_2) \,. \explain{Since $r_v(F_2,F_1) \geq 0$  $\forall F_1,F_2 \in \mathcal{X}, v \in [n]$}
    \end{align}
This completes the proof of the lemma.
\end{proof}

\begin{mylemma} \label{lem:rv_tilde_implies_several_things}
   Let  $\X,\Y \in \{1\} \times [n]^L$,  $b_1,b_2 \in \{0,1\}$,   $v \in [n]$. Let $j \in [L+1]$ be the maximum index for which $\X_{1 \to j} = \Y_{1 \to j}$.
    If $\; \widetilde{r}_v(g_{\X,b_1}, g_{\Y,b_2}) > 0$, then at least one of the next two properties holds:
    \begin{enumerate}[(i)]
        \item $v \in Tail(j, S_{\Y})$. 
        \item $\X = \Y$.
    \end{enumerate}
\end{mylemma}
\begin{proof}
We start with a few observations.

Recall from \cref{def:rv_tilde} that 
$\widetilde{r}_v(g_{\X,b_1}, g_{\Y,b_1}) = r_v(g_{\X,b_1}, g_{\Y,b_2}) \cdot \mathbbm{1}_{\{\multiplicity(S_{\X},v) \le \multiplicity(S_{\Y},v)\}}$.
By the lemma statement, we have $\widetilde{r}_v(g_{\X,b_1}, g_{\Y,b_2}) > 0$, and so both of the next inequalities hold:
    \begin{align} 
    & r_v(g_{\X,b_1}, g_{\Y,b_2}) > 0  \label{eq:r_v_strictly_positive_2}\\
    & \multiplicity(S_{\X},v) \leq \multiplicity(S_{\Y},v)\,. \label{eq:m_SX_v_les_m_SY_v}
    \end{align} 
    By definition of $r_v$, we have $ r_v(g_{\X,b_1}, g_{\Y,b_2}) = r(g_{\X,b_1}, g_{\Y,b_2}) \cdot \mathbbm{1}_{\{g_{\X,b_1}(v) \neq g_{\Y,b_2}(v)\}} \,.$ 
Then \cref{eq:r_v_strictly_positive_2} implies    
\begin{align} 
& g_{\X,b_1}(v) \neq g_{\Y,b_2}(v) \,.  \label{eq:v_not_equal_gXb1_gXb2}
\end{align}

To prove that $v \in Tail(j, S_{\Y})$ or $\X = \Y$ we consider two cases: 

    \paragraph{Case 1: $v \in Tail(j,S_{\X})$.}
 
Let us decompose the staircase $S_{\X}$ into the initial segment $S_{\X_{1 \to j}}$ and the remainder $Tail(j,S_{\X})$. Similarly, we decompose  the staircase $S_{\Y}$
 into initial segment  $S_{\Y_{1 \to j}}$ and the remainder $Tail(j,S_{\Y})$. We get:
    \begin{align}
         & \multiplicity(S_{\X},v) \le  \multiplicity(S_{\Y},v) \explain{By \cref{eq:m_SX_v_les_m_SY_v}.}  \\
        \iff & \multiplicity(S_{\X_{1 \to j}},v) + \multiplicity(Tail(j,S_{\X}),v) \le \multiplicity(S_{\Y_{1 \to j}},v) + \multiplicity(Tail(j,S_{\Y}),v) \notag \\ 
         \iff & \multiplicity(Tail(j,S_{\X}),v) \le \multiplicity(Tail(j,S_{\Y}),v) \explain{Since $\X_{1 \to j} = \Y_{1 \to j}$.}  \\
    \end{align}
    Since $v \in Tail(j,S_{\X})$, we have $\multiplicity(Tail(j,S_{\X}),v) \geq 1$, and so  
    \begin{align*}
        1 \leq  \multiplicity(Tail(j,S_{\Y}),v) \,.
    \end{align*}
    Thus $v \in Tail(j,S_{\Y})$, so property $(i)$ from the lemma statement holds. This completes Case 1.
    
    \paragraph{Case 2: $v \notin Tail(j,S_{\X})$.}
If $\X=\Y$, then property $(ii)$ from the lemma statement  holds. 

Now suppose  $\X \neq \Y$. Then $Tail(j, S_{\Y}) \neq \emptyset$.
We claim  $v \in S_{\X} \cup S_{\Y}$. 
Suppose towards a contradiction that $v \not \in S_{\X} \cup  S_{\Y}$.

Then for each  $b \in \{0,1\}$,  $u \in [n]$, and sequence $\Z = (1, z_2, z_3 \ldots, z_{L+1}) \in \{1\} \times [n]^{L}$, we have by \cref{eq:def:g_X_b} (which defines the function $g_{\vec{z}, b}$) that 
\begin{align} 
g_{\vec{z},b}(u) = 
\begin{cases}
\bigl(f_{\vec{z}}(u), b\bigr) & \text{if } u = z_{L+1}\, \\
 \bigl(f_{\vec{z}}(u), -1\bigr) & \text{otherwise}\,.
\end{cases} \notag 
\end{align}  

Since $v \not\in S_{\X}$, we have $v \neq \X_{L+1}$, and so  $g_{\X,b_1}(v) =(f_{\X}(v), -1)$.
 Moreover, since  $x_1 = y_1 $ and $v \not \in S_{\X} \cup S_{\Y}$, we have $f_{\X}(v) = dist(v,x_1) = dist(v,y_1) = f_\Y(v)\,.$ Combining these observations yields  
$g_{\X,b_1}(v) = (f_{\X}(v), -1)  
= (f_\Y(v),-1)  
 = g_{\Y,b_2}(v),$
which contradicts \cref{eq:v_not_equal_gXb1_gXb2}.
Thus the assumption must have been false and $v \in S_{\X} \cup S_{\Y}$. 

To summarize, we have  $\X_{1\to j} = \Y_{1 \to j}$, $\X \neq \Y$, $v \in S_{\vec{x}} \cup S_{\vec{y}}$, and $v \not \in Tail(j,S_{\X})$. Suppose towards a contradiction that $v \not \in Tail(j,S_{\Y})$. Then 
\begin{align}
& g_{\X,b_1}(v)  =(f_{\X}(v), -1) \explain{Since $v \neq x_{L+1}$, as $v \not \in Tail(j, S_{\X})$.} \\
& = (f_\Y(v), -1) \explain{Since $v \in S_{\X} \cup S_{\Y}$ and $\X_{1 \to j} = \Y_{1 \to j}$ and ($v \not \in Tail(j,S_{\X})$, $v \not \in Tail(j,S_{\Y})$).} \\
& = g_{\Y,b_2}(v) \explain{Since $v \neq y_{L+1}$, as $v \not \in Tail(j, S_{\Y})$.},
 \end{align}
 which contradicts \cref{eq:v_not_equal_gXb1_gXb2}.

Thus the assumption must have been false and $v  \in Tail(j,S_{\Y})$, so property $(i)$ from the lemma statement holds. 

We conclude that at least one of properties $(i)$ and $(ii)$ holds. This completes Case 2, as well as the proof of the lemma.
\end{proof}

\begin{mylemma} \label{lem:j_le_ell_implies_v_in_tail}
For each $\X \in \{1\} \times [n]^L$, $b_1 \in \{0,1\}$, $v \in [n]$, and $j \in [L]$, we have 
\[
\sum_{\substack{{\Y} \in \{1\} \times [n]^L, \; b_2 \in \{0,1\} :\\ \max\{i : \X_{1 \to i} = \Y_{1 \to i}\} = j }} \widetilde{r}_v(g_{\X,b_1},g_{\Y,b_2})  \le
\sum_{\substack{{\Y} \in \{1\} \times [n]^L :\\ \max\{i : \X_{1 \to i} = \Y_{1 \to i}\} = j \\ v \in Tail(j, S_{\Y}) }} {r}(g_{\X,b_1},g_{\Y,1-b_1})\,.
\]
\end{mylemma}
\begin{proof}
Let $\X \in \{1\} \times [n]^L$, $b_1 \in \{0,1\}$, $v \in [n]$, and $j \in [L]$.

Recall the function $r$ from \cref{def:r}, the function $r_v$ from \cref{def:rv}, and the function $\widetilde{r}_v$ from \cref{def:rv_tilde}. In particular,  we have 
%
\begin{description}
\item[$\bullet$] $r_v(F_1,F_2) = r(F_1,F_2) \cdot \mathbbm{1}_{\{F_1(v) \ne F_2(v)\}}$ for all $F_1, F_2 \in \mathcal{X}  $.
\item[$\bullet$] $\widetilde{r}_v(g_{\X,b_1}, g_{\Y,b_2}) = r_v(g_{\X,b_1}, g_{\Y,b_2}) \cdot \mathbbm{1}_{\{\multiplicity(S_{\X},v) \le \multiplicity(S_{\Y},v)\}}$  for each $ \X,\Y \in \{1\} \times [n]^L$ \\and $b_1,b_2 \in \{0,1\} $.
\end{description}

Then we get the next chain of identities:
\begin{small}
\begin{align} 
        & \sum_{\substack{{\Y} \in \{1\} \times [n]^L, \; b_2 \in \{0,1\} :\\ \max\{i : \X_{1 \to i} = \Y_{1 \to i}\} = j }} \widetilde{r}_v(g_{\X,b_1},g_{\Y,b_2}) 
          =  \sum_{\substack{{\Y} \in  \{1\} \times [n]^L, \; b_2 \in \{0,1\} :\\ \max\{i : \X_{1 \to i} = \Y_{1 \to i}\} = j \\ \widetilde{r}_v(g_{\X,b_1},g_{\Y,b_2}) > 0} } \widetilde{r}_v(g_{\X,b_1},g_{\Y,b_2})  \explain{Since $\widetilde{r}_v$ is non-negative.} \\
         & \qquad \qquad = \sum_{\substack{{\Y} \in  \{1\} \times [n]^L :\\ \max\{i : \X_{1 \to i} = \Y_{1 \to i}\} = j \\ \widetilde{r}_v(g_{\X,b_1},g_{\Y,1-b_1}) > 0} } \widetilde{r}_v(g_{\X,b_1},g_{\Y,1-b_1})  \explain{Since $\widetilde{r}_v(g_{\X,b_1}, g_{\Y,b_1}) = 0$.} \\
         &  \qquad \qquad = \sum_{\substack{{\Y} \in  \{1\} \times [n]^L :\\ \max\{i : \X_{1 \to i} = \Y_{1 \to i}\} = j \\ \multiplicity(S_{\X}, v) \leq \multiplicity(S_{\Y}, v)  \\\widetilde{r}_v(g_{\X,b_1},g_{\Y,1-b_1})  > 0 }}{r}_v(g_{\X,b_1},g_{\Y,1-b_1}) \explain{By definition of $\widetilde{r}_v$.} \\
         & \qquad \qquad  = \sum_{\substack{{\Y} \in  \{1\} \times [n]^L :\\ \max\{i : \X_{1 \to i} = \Y_{1 \to i}\} = j \\ \multiplicity(S_{\X}, v) \leq \multiplicity(S_{\Y}, v)  \\\widetilde{r}_v(g_{\X,b_1},g_{\Y,1-b_1})  > 0 }} {r}(g_{\X,b_1},g_{\Y,1-b_1}) \cdot \mathbbm{1}_{\{g_{\X,b_1}(v) \neq g_{\Y,1-b_1}(v) \}} \explain{By definition of $r_v$.} \\
          & \qquad \qquad  = \sum_{\substack{{\Y} \in  \{1\} \times [n]^L :\\ \max\{i : \X_{1 \to i} = \Y_{1 \to i}\} = j \\ \multiplicity(S_{\X}, v) \leq \multiplicity(S_{\Y}, v) \\g_{\X,b_1}(v) \neq g_{\Y,1-b_1}(v) \\ \widetilde{r}_v(g_{\X,b_1},g_{\Y,1-b_1})  > 0 }} {r}(g_{\X,b_1},g_{\Y,1-b_1}) \,.
          \label{eq:sum_for_one_gxb_decomposed_by_j_partial} 
\end{align}
\end{small}

Consider  an arbitrary $\Y \in  \{1\} \times [n]^L$ meeting the properties from the last sum of \cref{eq:sum_for_one_gxb_decomposed_by_j_partial}: 
\begin{itemize}
\item $\max\{i : \X_{1 \to i} = \Y_{1 \to i}\} = j$
\item $ \multiplicity(S_{\X}, v) \leq \multiplicity(S_{\Y}, v) $
\item $g_{\X,b_1}(v) \neq g_{\Y,1-b_1}(v)$
\item $ \widetilde{r}_v(g_{\X,b_1},g_{\Y,1-b_1})  > 0$
\end{itemize}
By \cref{lem:rv_tilde_implies_several_things}, 
the inequality $\widetilde{r}_v(g_{\X,b_1},g_{\Y,1-b_1})>0$ implies that 
 at least one of  $v \in Tail(j, S_{\Y})$ or $\X=\Y$  holds. However, $\X$ cannot be equal to such $\Y$ as $j < L+1$ is the maximum index for which $\X_{1 \to j} = \Y_{1 \to j}$. Thus $v \in Tail(j,S_{\Y})$. We can continue to bound the sum from \cref{eq:sum_for_one_gxb_decomposed_by_j_partial} as follows:
\begin{align} 
 \sum_{\substack{{\Y} \in  \{1\} \times [n]^L :\\ \max\{i : \X_{1 \to i} = \Y_{1 \to i}\} = j \\ \multiplicity(S_{\X}, v) \leq \multiplicity(S_{\Y}, v) \\g_{\X,b_1}(v) \neq g_{\Y,1-b_1}(v)  \\ \widetilde{r}_v(g_{\X,b_1},g_{\Y,1-b_1})  > 0 }} {r}(g_{\X,b_1},g_{\Y,1-b_1})  \leq \sum_{\substack{{\Y} \in  \{1\} \times [n]^L :\\ \max\{i : \X_{1 \to i} = \Y_{1 \to i}\} = j \\ v \in Tail(j, S_{\Y}) }} {r}(g_{\X,b_1},g_{\Y,1-b_1})\,.  \label{eq:sum_for_one_gxb_decomposed_by_j_partial_step2} 
 \end{align}
 Combining \cref{eq:sum_for_one_gxb_decomposed_by_j_partial} and \cref{eq:sum_for_one_gxb_decomposed_by_j_partial_step2}, we get the inequality required by the lemma statement:
 \begin{align} 
 \sum_{\substack{{\Y} \in \{1\} \times [n]^L, \; b_2 \in \{0,1\} :\\ \max\{i : \X_{1 \to i} = \Y_{1 \to i}\} = j }} \widetilde{r}_v(g_{\X,b_1},g_{\Y,b_2}) \leq   \sum_{\substack{{\Y} \in  \{1\} \times [n]^L :\\ \max\{i : \X_{1 \to i} = \Y_{1 \to i}\} = j \\ v \in Tail(j, S_{\Y}) }} {r}(g_{\X,b_1},g_{\Y,1-b_1})\,.  \notag 
 \end{align} 
\end{proof}

\begin{mylemma}\label{lem:upper_bound_set_v_in_Tail_Y}
Let $\X \in \{1\} \times [n]^L$, $v \in [n]$, and $j \in [L]$.
Then
\begin{small}
\begin{align} 
\Bigl| \Bigl\{ \Y \in \{1\} \times [n]^L \mid \max\{i : \X_{1 \to i} = \Y_{1 \to i} \} = j \text{ and } v \in Tail(j,S_{\Y}) \Bigr\}  \Bigr|    \leq  \numPaths{v}(x_j) \cdot n^{L - j} + L \cdot g \cdot n^{L - j - 1}\,, \notag 
\end{align} 
\end{small}
recalling that $\numPaths{v}(u)$ is the number of paths in the set $\mathcal{P}$ that start at vertex $u$ and contain $v$.
\end{mylemma}
\begin{proof}
    There are two ways $v$ could be in $Tail(j, S_{\Y})$: either $v \in P^{y_j, y_{j+1}}$ or $v \in P^{y_i, y_{i+1}}$ for some $j < i \leq L$. 
    We will now upper bound the number of $\Y$ for each of these possibilities separately.
    \begin{enumerate}[(i)] 
\item  We count the number of sequences of vertices $\Y = (y_1, \ldots, y_{L+1})$ with 
\begin{align}
    \X_{1 \to j} = \Y_{1 \to j} \text{ and } v \in P^{y_j, y_{j+1}}\,.  \label{eq:Y_prop_related_to_X}
\end{align} 
There are $\numPaths{v}(x_j)$ choices of vertices for $y_{j+1}$ and  $n^{L - j}$ choices for sequences $(y_{j+2}, \ldots, y_{L+1})$. Thus there are $\numPaths{v}(x_j) \cdot n^{L - j}$ choices of $\Y$ for which  \cref{eq:Y_prop_related_to_X} holds.
\item  For each $i$ with $ j < i \leq L$, we count the number of sequences of vertices $\Y = (y_1, \ldots, y_{L+1})$ with 
\begin{align}
    \X_{1 \to j} = \Y_{1 \to j} \text{ and } v \in P^{y_i, y_{i+1}}\,.  \label{eq:Y_prop_related_to_X_i_j}
\end{align} 
There are $n^{L - j -1}$  tuples of the form $(y_{j+1}, \ldots, y_{i-1}, y_{i+2}, \ldots, y_{L+1})$, as these vertices can be chosen arbitrarily. Since there are at most $g$ paths in $\cP$ that contain $v$, the number of choices for the pair $(y_{i}, y_{i+1})$ is at most $g$ as well. Thus there are $g \cdot n^{L - j -1}$ choices of $\Y$ for which \cref{eq:Y_prop_related_to_X_i_j} holds.

There are at most $L-j \leq L$ possible locations for $i$ when $j < i \leq L$.
So, there are at most $L \cdot g \cdot n^{L-j-1}$ choices for $\Y$ such that \cref{eq:Y_prop_related_to_X_i_j} holds with some index $ i \in \{j+1, \ldots, L\}$.

\end{enumerate}
\medskip 

Combining the analysis from (i)-(ii), we get that:
\begin{small}
    \[
    \Bigl| \Bigl\{ \Y \in \{1\} \times [n]^L \mid \max\{i : \X_{1 \to i} = \Y_{1 \to i} \} = j \text{ and } v \in Tail(j,S_{\Y}) \Bigr\}  \Bigr|    \leq  \numPaths{v}(x_j) \cdot n^{L - j} + L \cdot g \cdot n^{L - j - 1} \,.
    \]
\end{small}
This completes the proof of the lemma.
\end{proof}

\begin{mylemma} \label{lem:count_Y_denominator}
Let $j \in [L]$ and $\X \in \{1\} \times [n]^L$ be an arbitrary good sequence of vertices. 
   Then 
    \[
     \Bigl| \Y \in \{1\} \times [n]^L \mid \Y \; \text{is good and } j = \max\{i:\X_{1 \to i} = \Y_{1 \to i}\} \Bigr|  =  (n-j-1)\prod_{i=j+2}^{L+1} (n-i+1) \,.
    \]
\end{mylemma}
\begin{proof}
Let $\Y = (y_1, \ldots, y_{L+1}) \in [n]^{L+1}$ be such that $x_1 = y_1 = 1$,  $j = \max\{i:\X_{1 \to i} = \Y_{1 \to i}\}$, and $\vec{y}$ is good.
Recall that for $\Y$ to be good, each element in $(y_1, \ldots, y_{L+1})$ has to be distinct.

    We will count by choosing each vertex of $\Y$ in order from $y_1$ to $y_{L+1}$.
    There is only one choice for each of $y_1, \ldots, y_j$ since we need $\X_{1 \to j} = \Y_{1 \to j}$.
    
    Consider the number of possible vertices for $y_{j+1}$:
    \begin{description}
        \item[$\bullet$] $j$ vertices already appeared in $\{y_1, \ldots, y_j\}$ so they cannot be reused, otherwise $\Y$ becomes bad.
        \item[$\bullet$] $j = \max\{i:\X_{1 \to i} = \Y_{1 \to i}\}$ means $x_{j+1} \neq y_{j+1}$, so $x_{j+1}$ cannot be reused either.
    \end{description}
    Thus, there are $n-j-1$ choices of vertices for $y_{j+1}$.
    
    For all $i > j+1$, there are $n-i+1$ choices of vertices for $y_i$ since $i-1$ options have already been used.
    Thus the final count is
$(n-j-1)\prod_{i=j+2}^{L+1} (n-i+1)$, as required.
\end{proof}

\begin{restatable}{mylemma}{Mlarge}
\label{lem:M_large}
If $F_1 \in \mathcal{X}$ is good, then 
$ \; 
\sum_{F_2 \in \cX} r(F_1,F_2) \geq \frac{1}{2e} \cdot (L+1) \cdot n^{L+1}\,.
$
\end{restatable}
\begin{proof}
Since $F_1$ is good, there exists a good sequence  
 $\X \in \{1\} \times [n]^L$ and a bit  $b_1 \in \{0,1\}$ such that $F_1 = g_{\vec{x}, b_1}$. Then we have the following chain of identities:
\begin{align}
\sum_{F_2 \in \cX} r(F_1, F_2) & =
\sum_{F_2 \in \cX} r(g_{\X,b_1}, F_2) \notag \\
&= \sum_{\substack{\Y \in \{1\} \times [n]^L,\; b_2 \in \{0,1\}}} r(g_{\X,b_1}, g_{\Y,b_2}) \explain{By definition of the set $\mathcal{X}$}\\
&= \sum_{\substack{\Y \in \{1\} \times [n]^L}} r(g_{\X,b_1}, g_{\Y,1-b_1}) \explain{Since $r(g_{\X,b_1}, g_{\Y,b_1}) = 0$}\\
&= \sum_{j=1}^{L+1} \sum_{\substack{\Y \in \{1\} \times [n]^L\\ j = \max\{i : \X_{1 \to i} = \Y_{1 \to i}\}}} r(g_{\X,b_1}, g_{\Y,1-b_1}) \,.
 \label{eq:denominator_M_to_sum_nj_0} 
 \end{align}
Since $r(g_{\X,b_1}, g_{\Y,1-b_1}) = 0$ if $\vec{y}$ is bad, we have 
 \begin{align}
\sum_{j=1}^{L+1} \sum_{\substack{\Y \in \{1\} \times [n]^L:\\ j = \max\{i : \X_{1 \to i} = \Y_{1 \to i}\}}} r(g_{\X,b_1}, g_{\Y,1-b_1}) &= \sum_{j=1}^{L+1} \sum_{\substack{\Y \in \{1\} \times [n]^L:\\ j = \max\{i : \X_{1 \to i} = \Y_{1 \to i}\}\\ \Y \text{ is good } }} r(g_{\X,b_1}, g_{\Y,1-b_1})  \,. \label{eq:denominator_M_to_sum_nj_1}
\end{align}
Combining \cref{eq:denominator_M_to_sum_nj_0}  and \cref{eq:denominator_M_to_sum_nj_1}, the last sum in  \cref{eq:denominator_M_to_sum_nj_1} can be written as:
\begin{align}
%
\sum_{F_2 \in \cX} r(g_{\X,b_1}, F_2)= \sum_{j=1}^{L+1} \sum_{\substack{\Y \in \{1\} \times [n]^L:\\ j = \max\{i : \X_{1 \to i} = \Y_{1 \to i}\}\\ \Y \text{ is good } }} r(g_{\X,b_1}, g_{\Y,1-b_1})\,.  \label{eq:denominator_M_to_sum_nj_2}
\end{align}

Substituting the definition of $r$ in \cref{eq:denominator_M_to_sum_nj_2}, we get 
\begin{equation}
\label{eq:denominator_M_to_sum_nj}
\sum_{F_2 \in \cX} r(g_{\X,b_1}, F_2)
= \sum_{j=1}^{L+1} \sum_{\substack{\Y \in \{1\} \times [n]^L:\\ j = \max\{i : \X_{1 \to i} = \Y_{1 \to i}\}\\ \Y \text{ is good } }} n^j\,.
\end{equation}

From here, we only need to count the number of such $\Y$.
There is exactly one good sequence $\Y \in \{1\} \times [n]^L$ such that $\X_{1 \to L+1} = \Y_{1 \to L+1}$, namely $\Y = \X$.
Therefore,
\begin{align*}
\sum_{F_2 \in \cX} r(g_{\X,b_1}, F_2)
= &\; \sum_{j=1}^{L+1} \sum_{\substack{\Y \in \{1\} \times [n]^L:\\ \max\{i : \X_{1 \to i} = \Y_{1 \to i}\} = j\\ \Y \text{ is good } }} n^j \explain{By \cref{eq:denominator_M_to_sum_nj}}\\
= &\; n^{L+1} + \sum_{j=1}^{L} \sum_{\substack{\Y \in \{1\} \times [n]^L:\\ \max\{i : \X_{1 \to i} = \Y_{1 \to i}\} = j\\ \Y \text{ is good } }} n^j \explain{ 
{Can pull out $n^{L+1}$ since there is only one possibility}
{of } $\X_{1 \to L+1} = \Y_{1 \to L+1}$ { when } $\X = \Y$}\\
= &\; n^{L+1} + \sum_{j=1}^{L} n^j \left[ (n-j-1) \cdot \prod_{i=j+2}^{L+1} (n-i+1) \right] \explain{By \cref{lem:count_Y_denominator}}\\
\geq &\; n^{L+1} + \sum_{j=1}^{L} (n-L-1)^{L+1} \explain{Since $j \in [L]$} \\
\geq &\; (L+1) \cdot n^{L+1} \cdot \left( 1-\frac{L+1}{n} \right)^{L+1}\\
\geq &\; (L+1) \cdot n^{L+1} \cdot \left( 1-\frac{1}{L+1} \right)^{L+1} \explain{Since $L+1 \leq \sqrt{n}$}\\
\geq &\; \frac{1}{2e} \cdot (L+1) \cdot n^{L+1} \,. \explain{Since $n \geq 4$, so $L+1 \geq 2$}
\end{align*}
Since $F_1 = g_{\vec{x},b_1}$, this is the required inequality, which completes the proof.
\end{proof}

\subsection{Corollaries for expanders}

Given the lower bound based on congestion, we can now state an implication for expanders. For this we rely on the next corollary of a  result from \cite{BFU99}, which shows that  $d$-regular expanders have systems of paths with low vertex congestion when the degree  $d$ is constant.

\begin{mylemma} \label{cor:bfu99}
Let $G = ([n],E)$ be a $d$-regular $\beta$-expander, where $d$ and $\beta$ are constant.
Consider the collection of all $ n^2$ ordered pairs of vertices $\{(a_1, b_1), \ldots, (a_{n^2}, b_{n^2})\}$, including each vertex with itself.
Then there exists a set $\cP$ of $n^2$ paths $\{P_1, \ldots, P_{n^2}\}$, such that each path $P_i$ connects $a_i$ to $b_i$ and the congestion on each vertex of $G$ is  $ \cO \left( {n \ln n}\right)$.
\end{mylemma}
\begin{proof}
We invoke \cref{thm:bfu99} with parameters $K = n^2$, $\alpha(n) = n \ln{n}$, and $s = 2n$ to get a set of $K$ paths $\cP = \{P_1,\ldots,P_K\}$.
Since $\alpha(n) \ge 1/2$, we get that the edge congestion $g$ of $\cP$ is at most:
\begin{align*}
    g \in \cO(s + \alpha + \ln \ln n) = \cO(n \ln n)\,.
\end{align*}
To convert from edge congestion to vertex congestion, consider the vertex $v$ with the highest vertex congestion with respect to $\cP$.
The vertex congestion at $v$ is no more than the sum of the edge congestions on each of the edges incident to $v$.
Because $G$ is $d$-regular, this means the vertex congestion at $v$ is at most $d\cdot g$.
Since $d$ is constant, the vertex congestion is $\cO(n \ln n)$.
\end{proof}

Using this corollary, we get the following result for local search on expanders.

\corollaryExpandersLowerBound*
\begin{proof}
By \cref{cor:bfu99}, the graph $G$ has an all-pairs set of paths $\cP$ with vertex congestion $g \in \cO \left( {n \ln n}\right)$.
By \cref{thm:low_congestion_implies_local_search_hard}, the randomized query complexity of local search on $G$ is $\Omega(\sqrt{n}/\ln n)$.
\end{proof}

By alternatively using a prior result from \cite{leighton1999multicommodity}, we get a result in terms of the expansion and maximum degree for any graph.

\begin{mylemma}
\label{cor:chuzhoy2016routing}
Let $G = ([n],E)$ be a $\beta$-expander with maximum vertex degree $\Delta$.
Consider the collection of all $ n^2$ ordered pairs of vertices $\{(a_1, b_1), \ldots, (a_{n^2}, b_{n^2})\}$, including each vertex with itself.
Then there exists a set $\cP$ of $n^2$ paths $\{P_1, \ldots, P_{n^2}\}$, such that each path $P_i$ connects $a_i$ to $b_i$ and the congestion on each vertex of $G$ is $O(n  \ln^2 (n) \cdot \frac{\Delta}{\beta})$.
\end{mylemma}
\begin{proof}
Decompose the clique on $n$ vertices into $n$ partial matchings $M_1, M_2, \ldots, M_n$.
Invoke \cref{thm:chuzhoy2016routing} on each partial matching $M_i$ to get a set $\cP_{u,v}$ of $\lceil \ln n \rceil$ paths from $u$ to $v$ for every pair of vertices $\{u,v\} \in M_i$.
Let $P_i$ be an arbitrary path from $\cP_{a_i, b_i}$ for every $i \in [n^2]$.

By \cref{thm:chuzhoy2016routing}, the edge congestion for each partial matching is at most $O(\ln^2 (n) / \beta)$.
Therefore the edge congestion of the union of the results of invoking \cref{thm:chuzhoy2016routing} is at most $O(n \cdot \ln^2 (n) / \beta)$.
Since $\cP$ contains only one path from each $\cP_{u,v}$, it also has edge congestion at most $O(n \cdot \ln^2 (n) / \beta)$.

To convert from edge congestion to vertex congestion, consider the vertex $v$ with the highest vertex congestion with respect to $\cP$.
The vertex congestion at $v$ is no more than the sum of the edge congestions on each of the edges incident to $v$.
Because the maximum degree is $\Delta$, this means the vertex congestion at $v$ is at most $O(n  \ln^2 (n) \cdot \frac{\Delta}{\beta})$, so 
 the vertex congestion of $O(n  \ln^2 (n) \cdot \frac{\Delta}{\beta})$.
\end{proof}

Using this congestion result, we get a corollary  for expansion and maximum degree.

\corollaryViaExpansionLowerBound*
\begin{proof}
    By \cref{cor:chuzhoy2016routing},  $G$ has an all-pairs set of paths $\cP$ with vertex congestion $g \in \cO \left( n \cdot \ln^2 (n) \cdot \Delta / \beta \right)$.
    By \cref{thm:low_congestion_implies_local_search_hard}, the randomized query complexity of local search on $G$ is $\Omega\left(\frac{\beta \sqrt{n}}{\Delta \log^2{n}}\right)$.
\end{proof}

\subsection{Corollary for Cayley graphs}

We also get a corollary for undirected Cayley graphs thanks to a construction by \cite{dinh2010quantum}, which has vertex congestion of at most $(d+1)\cdot n$. This improves the result of \cite{dinh2010quantum} for randomized algorithms by a $\ln n$ factor.

\begin{mylemma} \label{lem:cayley_low_congestion}
    Let $G = (V,E)$ be an undirected Cayley graph with $n$ vertices and diameter $d$.
    Then there exists an all-pairs set of paths $\mathcal{P} = \left\{P^{u,v}\right\}_{u,v \in V}$, such that each path $P^{u,v}$ connects $u$ to $v$ and the congestion on each vertex of $G$ is at most $(d+1)\cdot n$.
\end{mylemma}
\begin{proof}
    Let $1 \in V$ be the group identity of $G$.
    For each $v \in V$, fix $P^{1,v}$ to be an arbitrary shortest path from $1$ to $v$.
    Then for each pair $u,v \in V$ with $u \ne 1$, let $P^{u,v} = u \cdot P^{1,w}$, where $w = u^{-1} \cdot v$.
    By construction, $P^{u,v}$ starts at $u \cdot 1 = u$ and ends at $u \cdot u^{-1} \cdot v = v$.

    For each $w \in V$, let $\cP_w = \{P^{u,v} : u^{-1} \cdot v = w\}$.
    For all $w,x \in V$, exactly $|P^{1,w}|$ paths in $\cP_w$ contain $x$: one has $x$ in the first position, one has $x$ in the second position, etc.
    Therefore $\cP_w$ has the same vertex congestion at every vertex.
    Then since $\cP$ is the disjoint union $\bigcup_{w \in V} \cP_w$, we get that $\cP$ has the same vertex congestion at every vertex.
    
    Every path in $\cP$ is a shortest path, and so has length at most $d+1$ vertices.
    Therefore in total there are $(d+1) \cdot n^2$ vertices in $\cP$, so the vertex congestion of $\cP$ is $(d+1) \cdot n$ since every vertex has the same congestion.
\end{proof}

Using this lemma we get the following result for local search on undirected Cayley graphs.

\corollaryCayleyLowerBound*
\begin{proof}
    By \cref{lem:cayley_low_congestion}, the graph $G$ has an all-pairs set of paths $\cP$ that has vertex congestion $g \le (d+1) \cdot n$.
     By \cref{thm:low_congestion_implies_local_search_hard}, the randomized query complexity of local search on $G$ is $\Omega\left(\frac{\sqrt{n}}{d}\right)$.
\end{proof}

\subsection{Corollary for the hypercube}

Finally, we get a corollary for the $\{0,1\}^n$ hypercube.

\corollaryHypercubeLowerBound*
\begin{proof}
We first quantify the vertex congestion $g$ of the Boolean hypercube, and then invoke \cref{thm:low_congestion_implies_local_search_hard} to obtain a lower bound for local search.

        The number of vertices in the $\{0,1\}^n$ hypercube is $N = 2^n$.
    The vertices of the graph can be viewed as bit strings of length $n$, with bit strings of Hamming distance $1$ connected by an edge.
    Fix an order on the bits. Then for every pair of vertices $u,v$, the path $P^{u,v}$ is obtained by iterating over the bits, toggling each bit that differs between $u$ and $v$.

    By symmetry of the construction, the system of paths $\cP$ has equal congestion at every vertex.
    Furthermore, the average length of a path in $\cP$ is $1 + n/2$.
    Thus the congestion of $\cP$ is $N \cdot (1+ n/2)$.
    Since every path is a shortest path and the congestion is evenly distributed, this is optimal and the congestion of the graph is $g = N \cdot (1+ n/2)$.

    By \cref{thm:low_congestion_implies_local_search_hard},  the randomized query complexity of local search on the   hypercube is 
        $\Omega\left(\frac{2^{n/2}}{n}\right)$.
\end{proof}

\newpage 

\section{Lower bound for local search via separation number}
\label{sec:appendix-separation-number}

In this section we include the proofs needed to show the lower bound of $\Omega\left((s/\Delta)^{1/4}\right)$.

\medskip 

We start with  the basic definitions. Then we state and prove the lower bound. Afterwards, we show the helper lemmas used in the proof of the theorem.



\subsection{Basic definitions for separation}
\label{sec:appendix-separation-number-setup}

\paragraph{Notation.} 
Recall we have a graph $G = ([n],E)$ with separation number $s$ and maximum degree $\Delta$.
This means that every subset $H \subseteq [n]$ can be split in two parts, $S$ and $H \setminus S$, each of size at least $|H|/4$, such that at most $s$ vertices in $H \setminus S$ are adjacent to $S$.

Let $c \in \mathbb{N}$ be a parameter that we set later.

\medskip 
Given a sequence of $k$ indices $\X = (x_1, \ldots, x_k)$, we write $\X_{1 \to j}=(x_1,\ldots,x_j)$ to refer to a prefix of the sequence, for an index $j \in [k]$.

\smallskip  

Given a walk $Q = (v_1, \ldots, v_k)$ in $G$, let $Q_i$ refer to the $i$-th vertex in the walk  (i.e. $Q_i = v_i$). For each vertex $u \in [n]$, let $\multiplicity(Q,u)$ be the number of times that vertex $u$ appears in $Q$. 

Next we introduce
the notion of inter-cluster paths  from  \cite{santha2004quantum}.

\begin{mydefinition}[Path Arrangement and Inter/Intra-Cluster Paths]
\label{def:path-arrangement-parameter}
A \emph{path arrangement} with parameter $m$ for graph $G = ([n], E)$ is a set of connected, disjoint subsets $N_1, \ldots, N_m \subseteq V$ with the following property:
for all $i,j \in [m]$, there exist $m$ paths $P_1(i,j), \ldots, P_m(i,j)$ such that
\begin{itemize}
    \item For each $k \in [m]$, the first vertex of $P_{k}(i,j)$ is in $N_i$, the last vertex of $P_{k}(i,j)$ is in $N_j$, and every other vertex of $P_{k}(i,j)$ is outside $N_i$ and $N_j$.
    \item For every pair $i,j \in [m]$, vertices in $V \setminus (N_i \cup N_j)$ are visited at most once collectively by the paths $P_1(i,j), \ldots, P_m(i,j)$.
\end{itemize}
\end{mydefinition}


\medskip 

\begin{observation} 
\cref{def:path-arrangement-parameter} differs slightly from the original construction of \cite{santha2004quantum} in that \cref{def:path-arrangement-parameter} uses inter-cluster paths of the form $P_k(i,i)$; i.e. from a cluster to itself.
However, $P_k(i,i)$ may always be chosen as a degenerate single-vertex path in $N_i$, so this deviation does not affect the value of $m$ for any graph.
\end{observation}

\newpage 

\begin{example}
Let $H$ be the  grid of \cref{fig:grid-cluster-diagram}. There exists a path arrangement for $H$ with parameter $\sqrt{n}$ since one can use the columns as the $N_i$ and the shortest path in each row as the inter-cluster paths.

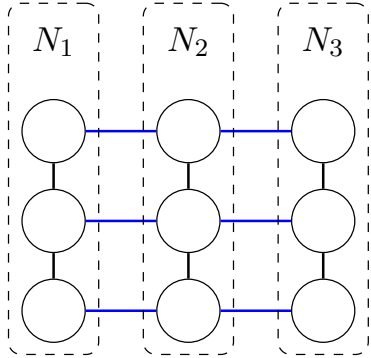
\begin{figure}[h!]
\centering
\resizebox{0.3\linewidth}{!}{%
\begin{tikzpicture}
%
%
\node[draw, circle, minimum size=20pt, inner sep=2pt] at (-1.5,1) (n11) {};
\node[draw, circle, minimum size=20pt, inner sep=2pt] at (-1.5,0) (n12) {};
\node[draw, circle, minimum size=20pt, inner sep=2pt] at (-1.5,-1) (n13) {};
\node[] at (-1.5, 2) (n1-label) {$N_1$};

%
%
\node[draw, circle, minimum size=20pt, inner sep=2pt] at (0,1) (n21) {};
\node[draw, circle, minimum size=20pt, inner sep=2pt] at (0,0) (n22) {};
\node[draw, circle, minimum size=20pt, inner sep=2pt] at (0,-1) (n23) {};
\node[] at (0, 2) (n2-label) {$N_2$};

%
%
\node[draw, circle, minimum size=20pt, inner sep=2pt] at (1.5,1) (n31) {};
\node[draw, circle, minimum size=20pt, inner sep=2pt] at (1.5,0) (n32) {};
\node[draw, circle, minimum size=20pt, inner sep=2pt] at (1.5,-1) (n33) {};
\node[] at (1.5, 2) (n3-label) {$N_3$};

\draw[thick] (n11) -- (n12);
\draw[thick] (n12) -- (n13);
\draw[thick] (n21) -- (n22);
\draw[thick] (n22) -- (n23);
\draw[thick] (n31) -- (n32);
\draw[thick] (n32) -- (n33);
\draw[thick, blue] (n11) -- (n21);
\draw[thick, blue] (n21) -- (n31);
\draw[thick, blue] (n12) -- (n22);
\draw[thick, blue] (n22) -- (n32);
\draw[thick, blue] (n13) -- (n23);
\draw[thick, blue] (n23) -- (n33);

\node[draw, fit=(n11)(n12)(n13)(n1-label), rounded corners, dashed] {};
\node[draw, fit=(n21)(n22)(n23)(n2-label), rounded corners, dashed] {};
\node[draw, fit=(n31)(n32)(n33)(n3-label), rounded corners, dashed] {};
\end{tikzpicture}
}
\caption{
Let $H$ be  the $\sqrt{n} \times \sqrt{n}$ graph with clusters $N_1, \ldots, N_{\sqrt{n}}$, where $n = 9$.
The black edges represent intra-cluster paths while the blue edges represent inter-cluster paths.
Then, there exists a path arrangement for this graph with parameter $\sqrt{n}=3$.
}
\label{fig:grid-cluster-diagram}
\end{figure}
\end{example}

Next we present without proof a lemma from \cite{santha2004quantum}, which relates the path arrangement parameter to the separation number and maximum degree of the graph.

\begin{mylemma}[\cite{santha2004quantum}, Theorem 6]
\label{lem:m_lb}
If an undirected graph $G$ has maximum degree $\Delta$ and 
separation number $s$, then there exist a path arrangement on $G$ with parameter at least $\max\{\lfloor \sqrt{s/2\Delta} \rfloor, 1\}$.
\end{mylemma}

\paragraph{Path arrangement number $m$ of the graph $G$.} Let $m$ be the maximum number such that there exists a path arrangement on our graph $G$ with parameter $m$.
Let $N_1, \ldots, N_m$ be the corresponding clusters. {An example of a set of inter-cluster paths between $N_i$ and $N_j$ from \cite{santha2004quantum} is shown in \cref{fig:inter-cluster-paths}.

\begin{figure}[h!]
\centering
\includegraphics[scale=1.2]{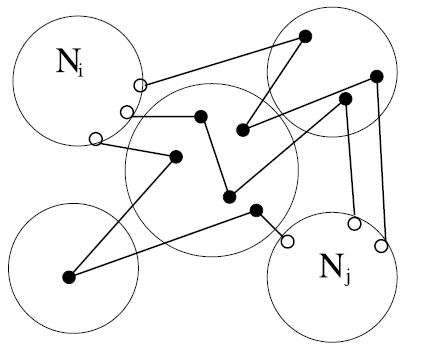}
\caption{A set of inter-cluster paths between $N_i$ and $N_j$ from \cite{santha2004quantum}.}
\label{fig:inter-cluster-paths}
\end{figure}





For all $i \in [m]$ and $u,v \in N_i$, let $E_i(u,v)$ be an arbitrary shortest path from $u$ to $v$ within $N_i$.

Define $first(P_i(j,k))$ and $last(P_i(j,k))$ as the first and last vertices in $P_i(j,k)$, respectively.
We have $first(P_i(j,k)) \in N_j$ and $last(P_i(j,k)) \in N_k$,  by definition of $P_i(j,k)$.
\medskip 

We will map a sequence of indices to walks starting from a fixed starting vertex $v_{start} \in N_1$.


\begin{mydefinition}[Staircase; separation version]
For $\X = (x_1 = 1, x_2, \ldots, x_{2c+1}) \in \{1\} \times [m]^{2c}$, and $i \in [2c]$, we define
\begin{equation}
    S_{\X,i} = \begin{cases}
    E_1(v_{start}, first(P_{x_2}(1,x_3)) & \text{if}\; i=1\\
    P_{x_i}(x_{i-1}, x_{i+1}) & \text{if}\; i \; \text{is even}\\
    E_{x_i}(last(Q_{x_{i-1}}(x_{i-2},x_{i})), first(P_{x_{i+1}}(x_{i},x_{i+2}))) & \text{if $i > 1$ is odd}
    \end{cases}
\end{equation}
Then for 
all sequences of indices $\X = (1, x_2, \ldots, x_{2k+1}) \in \{1\} \times [m]^{2k}$, the full walk $S_{\X}$ induced by $\X$ is the concatenation $S_{\X} = S_{\X,1} \circ S_{\X,2} \circ \ldots \circ S_{\X,2k}$.

\end{mydefinition}

That is, we use the inter-cluster paths dictated by the even indices of $\X$ to travel between clusters dictated by the odd indices of $\X$, stitching the inter-cluster paths together using shortest paths within the clusters.
So $S_\X$ starts at $v_{start} \in N_1$, travels via $S_{\X,1}$ and $S_{\X,2}$ to a vertex in $N_{x_3}$, and so on.
Observe $S_{\X}$ is a well-defined walk since $S_{\X,i}$ connects the last vertex of $S_{\X,i-1}$ and first vertex of $S_{\X,i+1}$ for any odd $i > 1$.
An illustrated example is given next.

\begin{example}
Consider the  graph in Figure \ref{fig:separation-staircase-example}.
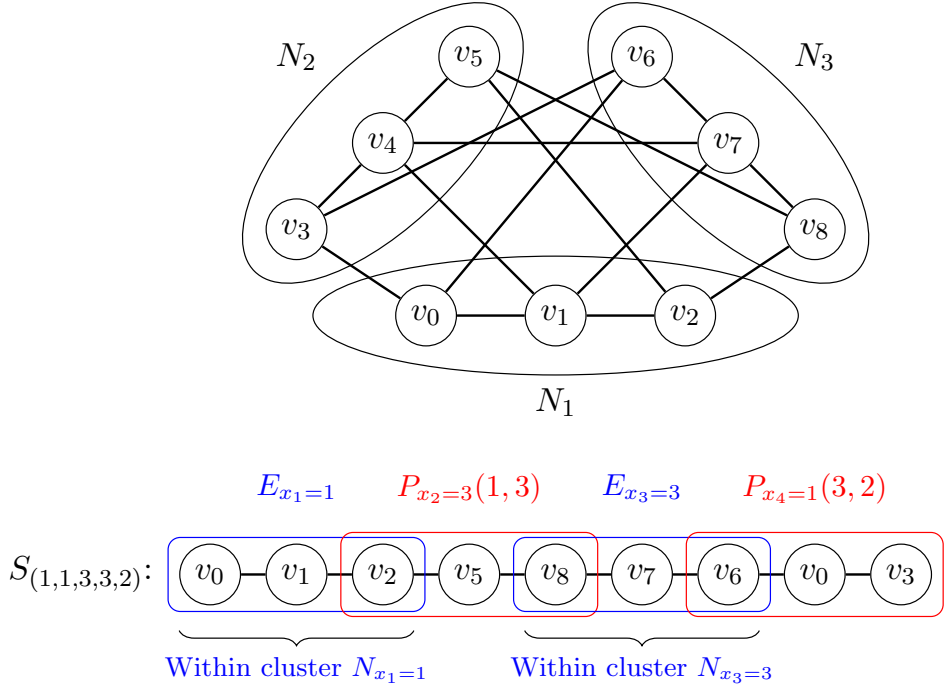
\begin{figure}[h!]
\centering
\resizebox{0.9\linewidth}{!}{%
\begin{tikzpicture}
%
%
\node[draw, circle, minimum size=20pt, inner sep=2pt] at (-0.5,0) (v0) {$v_0$};
\node[draw, circle, minimum size=20pt, inner sep=2pt] at (1,0) (v1) {$v_1$};
\node[draw, circle, minimum size=20pt, inner sep=2pt] at (2.5,0) (v2) {$v_2$};

%
%
\node[draw, circle, minimum size=20pt, inner sep=2pt] at (-2,1) (v3) {$v_3$};
\node[draw, circle, minimum size=20pt, inner sep=2pt] at (-1,2) (v4) {$v_4$};
\node[draw, circle, minimum size=20pt, inner sep=2pt] at (0,3) (v5) {$v_5$};

%
%
\node[draw, circle, minimum size=20pt, inner sep=2pt] at (2,3) (v6) {$v_6$};
\node[draw, circle, minimum size=20pt, inner sep=2pt] at (3,2) (v7) {$v_7$};
\node[draw, circle, minimum size=20pt, inner sep=2pt] at (4,1) (v8) {$v_8$};

%
%
\node[draw, ellipse, fit=(v0)(v1)(v2)] (fitn1) {};
\node[draw, ellipse, fit=(v3)(v4)(v5), rotate=45, inner ysep=-20pt] (fitn2) {};
\node[draw, ellipse, fit=(v6)(v7)(v8), rotate=135, inner ysep=-20pt] (fitn3) {};
\node[] at ($(v1) + (0,-1)$) {$N_1$};
\node[] at ($(v4) + (-1,1)$) {$N_2$};
\node[] at ($(v7) + (1,1)$) {$N_3$};

%
%
\draw[thick] (v0) -- (v3) -- (v6) -- (v0);
\draw[thick] (v1) -- (v4) -- (v7) -- (v1);
\draw[thick] (v2) -- (v5) -- (v8) -- (v2);
\draw[thick] (v0) -- (v1) -- (v2);
\draw[thick] (v3) -- (v4) -- (v5);
\draw[thick] (v6) -- (v7) -- (v8);

%
%
\node[minimum size=20pt, inner sep=2pt] at (-4.5,-3) (S) {$S_{(1,1,3,3,2)}$:};
\node[minimum size=20pt, inner sep=2pt] at (6.5,-3) (S) {\hphantom{$S_{(1,1,3,3,2)}$:}}; 

\node[draw, circle, minimum size=20pt, inner sep=2pt] at (-3,-3) (S1) {$v_0$};
\node[draw, circle, minimum size=20pt, inner sep=2pt] at (-2,-3) (S2) {$v_1$};
\node[draw, circle, minimum size=20pt, inner sep=2pt] at (-1,-3) (S3) {$v_2$};
\node[draw, circle, minimum size=20pt, inner sep=2pt] at (0,-3) (S4) {$v_5$};
\node[draw, circle, minimum size=20pt, inner sep=2pt] at (1,-3) (S5) {$v_8$}; 
\node[draw, circle, minimum size=20pt, inner sep=2pt] at (2,-3) (S6) {$v_7$};
\node[draw, circle, minimum size=20pt, inner sep=2pt] at (3,-3) (S7) {$v_6$};
\node[draw, circle, minimum size=20pt, inner sep=2pt] at (4,-3) (S8) {$v_0$};
\node[draw, circle, minimum size=20pt, inner sep=2pt] at (5,-3) (S9) {$v_3$};
\draw[thick] (S1) -- (S2) -- (S3) -- (S4) -- (S5) -- (S6) -- (S7) -- (S8) -- (S9);
\node[draw, fit=(S1)(S2)(S3), rounded corners, blue, inner ysep=2pt] (q1) {};
\node[draw, fit=(S3)(S4)(S5), rounded corners, red] (e2) {};
\node[draw, fit=(S5)(S6)(S7), rounded corners, blue, inner ysep=2pt] (q3) {};
\node[draw, fit=(S7)(S8)(S9), rounded corners, red] (e4) {};
\node[blue] at ($(q1) + (0,1)$) {\small $E_{x_1=1}$};
\node[red] at ($(e2) + (0,1)$) {\small $P_{x_2=3}(1,3)$};
\node[blue] at ($(q3) + (0,1)$) {\small $E_{x_3=3}$};
\node[red] at ($(e4) + (0,1)$) {\small $P_{x_4=1}(3,2)$};
\draw [decorate, decoration={mirror, brace, amplitude=5pt}] ($(S1.west) + (0,-0.65)$) -- ($(S3.east) + (0,-0.65)$) node[blue, midway, below, yshift=-5pt] {\footnotesize Within cluster $N_{x_1=1}$};
\draw [decorate, decoration={mirror, brace, amplitude=5pt}] ($(S5.west) + (0,-0.65)$) -- ($(S7.east) + (0,-0.65)$) node[blue, midway, below, yshift=-5pt] {\footnotesize Within cluster $N_{x_3=3}$};
\end{tikzpicture}
}
\caption{
Graph on nine nodes, with connected cluster partitions $N_1 = \{v_0, v_1, v_2\}$, \\ $N_2 = \{v_3, v_4, v_5\}$, and $N_3 = \{v_6, v_7, v_8\}$.
}
\label{fig:separation-staircase-example}
\end{figure}
Suppose $m = 3$. For each $i \in \{1,2,3\}$, let   the $i^{th}$ path for each pair of clusters consist of  vertices $\{v_{i-1}$, $v_{i-1+m}, v_{i-1+2m}\}$.

For example, $P_{2}(2,3) = (v_4, v_1, v_7)$ is the second path from $N_2$ to $N_3$ and $P_{2}(1,3) = (v_1, v_4, v_7)$ is the second path from $N_1$ to $N_3$.
\smallskip  

Fix $v_{start} = v_0$ and let $\X = (x_1, x_2, x_3, x_4, x_5) = (1,1,3,3,2)$.
We have  $P_3(1,3) = (v_2, v_5, v_8)$ and $P_1(3,2) = (v_6, v_0, v_3)$.
Thus  
\begin{align} 
S_{\X} & = S_{\X,1} \circ S_{\X,2} \circ S_{\X,3} \circ S_{\X,4} \notag \\
& = E_{x_1=1}( v_0, v_2) \circ P_{x_2=3}(1,3) \circ E_{x_3=3}(v_8, v_6) \circ P_{x_4=1}(3,2) \notag \\
& = (v_0, v_1, v_2) \circ (v_2, v_5, v_8) \circ (v_8, v_7, v_6) \circ (v_6, v_0, v_3) \notag  \\ 
& = (v_0, v_1, v_2, v_5, v_8, v_7, v_6, v_0, v_3)\,. \notag 
\end{align}
\end{example}

\begin{mydefinition}[Tail of a staircase; separation version]
\label{def:tail_separation}
Let $S_{\vec{x}} = S_{\vec{x},1} \circ \ldots \circ S_{\vec{x}, 2k}$ be a staircase induced by some sequence $\X  = (1, x_2, \ldots, x_{2k+1}) \in \{1\} \times [m]^{2k}$. 
For each odd $j \in [2k]$, let $T = S_{\vec{x},j} \circ \ldots \circ S_{\vec{x}, 2k}$.
Then $Tail(j, S_{\X})$ is obtained from $T$ by removing the first occurrence of the first vertex in $T$ (and only the first occurrence).
Let $Tail(2k+1, S_{\X})$ be the empty sequence.

Observe $Tail(j, S_{\X})$ is only defined for odd $j$ since staircase $S_{\X_{1 \to j}}$ is only defined for odd $j$.
\end{mydefinition}


Next, we define the set of functions $\cX$ that will be used when invoking \cref{thm:our-variant}. 

\begin{mydefinition}[The functions $f_{\vec{x}}$ and $g_{\vec{x},b}$; the  set $\mathcal{X}$; separation version]
\label{def:mathcal_X_separation}
Given an input graph $G$, let $m$ be its path arrangement parameter and $N_1, \ldots, N_m$ be the clusters with respect to \cref{def:path-arrangement-parameter}.
For each sequence of indices $\X \in \{1\} \times [m]^{2c}$, define $f_{\X} : [n] \to  \{-n^2, -n^2+1, \ldots, n\}$ such that for each $v \in [n]$:
\begin{itemize}
    \item If $v \notin S_{\X}$, then $f_{\X}(v) = dist(v,v_{start})$, where $S_{\X}$ is the staircase induced by $\X$ and the cluster construction. 
    \item If $v \in S_{\X}$, then $f_{\X}(v) = -i$, where $i$ is the maximum index such that $v$ is the $i$-th vertex in $S_{\X}$.
\end{itemize}
Also, for each $\X \in \{1\} \times [m]^{2c}$ and $b \in \{0,1\}$, let $g_{\X,b} : [n] \to \{-n^2, \ldots, 0,\ldots,n\} \times \{-1,0,1\}$ be such that, for all $v \in [n]$:
\[
g_{\X,b}(v) =
\begin{cases}
    (f_{\X}(v), b) & \text{if $v$ is the last vertex in $S_{\X}$}\\
    (f_{\X}(v), -1) & \text{if $v$ is not the last vertex in $S_{\X}$}\,.
\end{cases}
\]
Let $\cX = \left\{g_{\X,b} \mid \X \in \{1\} \times [m]^{2c} \text{ and } b \in \{0,1\}\right\}$.
\end{mydefinition}

\begin{mydefinition}[The map $\mathcal{H}$; separation version] \label{def:map_H_separation}
Given an input graph $G$, let $m$ be its path arrangement parameter and $N_1, \ldots, N_m$ be the clusters with respect to \cref{def:path-arrangement-parameter}.
Let $\mathcal{X}$ be the set of functions $g_{\vec{x},b}$ from \cref{def:mathcal_X_separation}.
Define $\mathcal{H} : \mathcal{X} \to \{0,1\}$ as 
\[
\cH(g_{\vec{x},b}) = b \qquad  \forall \vec{x} \in \{1\} \times [m]^{2c} \; \text{and} \; b \in \{0,1\} \,. 
\]
\end{mydefinition}

\begin{mydefinition}[Good/bad sequences of indices; Good/bad functions; separation version]
\label{def:good_separation}
A sequence of $k$ indices $\X = (x_1, \ldots, x_k)$ is \emph{good} if $x_i \neq x_j$ for all $i,j$ with $1 \leq i < j \leq k$; otherwise, $\X$ is bad.
For each $b \in \{0,1\}$, a function $F = g_{\X,b} \in \cX$ is good if $\X$ is good, and bad otherwise.
\end{mydefinition}


\begin{mydefinition}[The function $r$; separation version]
\label{def:r_separation}
Let $r : \cX \times \cX \to \mathbb{R}_{\geq 0}$ be defined by, for all $X,Y \in \{1\} \times [m]^{2c}$ and $b_1,b_2 \in \{0,1\}$, letting
\[
r(g_{\X,b_1}, g_{\Y,b_2}) =
\begin{cases}
    0 & \text{if at least one of the following holds: $b_1 = b_2$, $\X$ is bad, $\Y$ is bad}\\
    m^j & \text{otherwise, where $j$ is the maximum odd index for which $\X_{1 \to j} = \Y_{1 \to j}$}
\end{cases}
\]
\end{mydefinition}

Note the specification that $j$ be odd in the definition of $r$. Intuitively, this is because it takes two indices to specify an additional portion of $S_{\X}$: one for the destination cluster and one for the inter-cluster path by which to reach it.


The function $r$ will be used directly to invoke \cref{thm:our-variant}, but we also  define here some related helper functions to use $r$ in conjunction with certain indicator variables.

\begin{mydefinition}[The function $r_v$; separation version] \label{def:rv_sep}
    For each $v \in [n]$, define $r_v : \mathcal{X} \times \mathcal{X} \to \mathbb{R}_{\geq 0}$ as follows: 
    \[
        r_v(F_1,F_2) = r(F_1,F_2) \cdot \mathbbm{1}_{\{F_1(v) \ne F_2(v)\}} \qquad \forall F_1, F_2 \in \mathcal{X} \,.
    \]
\end{mydefinition}
\begin{mydefinition}[The function $\widetilde{r}_v$; separation version] \label{def:rv_tilde_sep}
    For each $v \in [n]$, define $\widetilde{r}_v : \mathcal{X} \times \mathcal{X} \to \mathbb{R}_{\geq 0}$ as follows:
    \[
        \widetilde{r}_v(g_{\X,b_1}, g_{\Y,b_2}) = r_v(g_{\X,b_1}, g_{\Y,b_2}) \cdot \mathbbm{1}_{\{\multiplicity(S_{\X},v) \le \multiplicity(S_{\Y},v)\}} \qquad \forall \X,\Y \in \{1\} \times [n]^L \; \; \forall b_1,b_2 \in \{0,1\} \,.
    \]
\end{mydefinition}



\subsection{Proof of the separation number lower bound}

Next we give the proof of \cref{thm:separation-number-result}.
The proofs of lemmas used in the theorem are included afterwards in \cref{sec:helper_lemmas_separation}.

\theoremseparationnumberlowerbound*

\begin{proof}
Because we are proving an asymptotic bound in $s/\Delta$, we may assume $s/\Delta \ge 162$.
Applying \cref{lem:m_lb} with the assumption that $s/\Delta \geq 162$ gives $m \geq 9$, so $c \ge 1$.
Additionally, we know $\Delta \ge 1$ and $s \le n$ by definitions of $\Delta$ and $s$.
Therefore $n \ge s/\Delta$, so we may also assume $n \ge 9$.

Consider the following setting of parameters:
\begin{enumerate}[(a)]
    \item Let $m$ be its path arrangement parameter and $N_1, \ldots, N_m$ be the clusters with respect to \cref{def:path-arrangement-parameter}.
    \item Let $c = \lfloor \sqrt{m}/2 -1/2 \rfloor$ and each staircase has $2c$ quasisegments.
    \item The finite set $A$ is the set of vertices $[n]$.
    \item The finite set $B$ is $\{-n^2,\ldots,n\} \times \{-1,0,1\}$.
    \item The set of functions $f_{\X}$, $g_{\X,b}$ and the set $\cX$ as defined in \cref{def:mathcal_X_separation}. Recall $g_{\X,b} = (f_{\X},c)$ for all $v \in [n]$, where $c = -1$ if $v$ is not the last vertex of the induced staircase $S_{\X}$, and $c = b$ if $v$ is (i.e.\ $c = b$ if and only if $v$ is a local minimum of $f_{\X}$). Also recall $\cX = \left\{g_{\X,b} \mid \X \in \{1\} \times [m]^{2c} \text{ and } b \in \{0,1\}\right\}$.
    \item Map $\cH: \cX \to \{0,1\}$ as in \cref{def:map_H_separation}. Recall $\cH(g_{\X,b}) = b$ for all $\X \in \{1\} \times [m]^{2c}$ and $b \in \{0,1\}$.
    \item The function $r$ as defined in \cref{def:r_separation}.
\end{enumerate}

By \cref{lem:f_X_is_valid_separation}, each function  $f_{\X}$ is valid for all $\X \in \{1\} \times [m]^{2c}$, so \cref{lem:valid_implies_unique_local_minimum} implies that each function $f_{\X}$ has a unique local minimum (at the last vertex of $S_{\X}$).
Therefore by \cref{lem:reduction_to_decision} invoked with $f=f_{\X}$ and $h_b = g_{\X,b}$, it suffices to show a lower bound for the corresponding decision problem: return the hidden bit $b \in \{0,1\}$ given oracle access to the function $g_{\X,b}$.

For each $\cZ \subseteq \cX$, let  
\begin{align}  \label{eq:remind_Z_def_sep}
M(\cZ) = \sum_{F_1 \in \cZ} \sum_{F_2 \in \cX} r(F_1, F_2)\,.  
\end{align}
Since by assumption $n \ge 9$, we may invoke \cref{lem:parameters_to_variant_are_ok_separation} to get a subset $\cZ \subseteq \cX$ with $q(\cZ) > 0$.
Thus the conditions required by  \cref{thm:our-variant} are met. 
By invoking \cref{thm:our-variant} with the parameters in (a-g), we get that the randomized query complexity of the decision problem, and thus also of local search on $G$, is 
\begin{align} 
\Omega\left(\min_{\cZ \subseteq \cX: q(\cZ) > 0}
\frac{M(\cZ)}{q(\cZ)}\right), \text{ where } q(\cZ) = \max_{v \in [n]}{\sum_{F_1 \in \cZ} \sum_{F_2 \in \cZ} r(F_1,F_2) \cdot \mathbbm{1}_{\{F_1(v) \neq F_2(v)\}}}\,. \notag 
\end{align} 

To get an explicit lower bound in terms of congestion, we will upper bound $q(\cZ)$ and lower bound $M(\cZ)$ for  subsets $\cZ \subseteq \cX$ with $q(\cZ) > 0$.

\medskip 

Fix an arbitrary subset $\cZ \subseteq \cX$ with $q(\cZ) > 0$.
Since $r(F_1,F_2) = 0$ {when $F_1$ or $F_2$ is bad}, it suffices to consider subsets $\cZ \subseteq \cX$ where each function $F \in \cZ$ is good. 




\paragraph{Upper bounding $q(\cZ)$.}

Let $v \in [n]$ be arbitrary. 

Fix an arbitrary good function $F_1 = g_{\vec{x}, b_1} \in \mathcal{Z}$ for some good $\X \in \{1\} \times [m]^{2c}$ and $b_1 \in \{0,1\}$.

Since $\cZ \subseteq \cX$ and $\widetilde{r}_v \geq 0$, we have 
\begin{align}
\label{eq:simple_inequality_tilde_r_v_F_1_F_2_sum_over_X_sep}
\sum_{F_2 \in \mathcal{Z}} \widetilde{r}_v(F_1,F_2) 
& \leq \sum_{F_2 \in \mathcal{X}} \widetilde{r}_v(F_1,F_2) \,. 
\end{align}
Using the definition of $\mathcal{X} = \{ g_{\Y,b_2} \mid \Y \in \{1\} \times [m]^{2c}, \; b_2 \in \{0,1\}\}$, the fact that $F_1 = g_{\vec{x},b_1}$, and partitioning the space of functions $F_2 \in \mathcal{X}$ by the length of the prefix that the staircase corresponding to  $F_2$  shares with the staircase corresponding to $F_1$, we can upper bound the right hand of \cref{eq:simple_inequality_tilde_r_v_F_1_F_2_sum_over_X_sep}:
\begin{align} \label{eq:simple_ineq_r_tilde_v_F_1_F_2_in_X_part2_sep}
\sum_{F_2 \in \mathcal{X}} \widetilde{r}_v(F_1,F_2) & = \sum_{\substack{\Y \in \{1\} \times [m]^{2c}, \; b_2 \in \{0,1\}}} \widetilde{r}_v(g_{\X,b_1},g_{\Y,b_2}) \notag \\
& \leq \sum_{j=1}^{2c+1} \sum_{\substack{\Y \in \{1\} \times [m]^{2c}, \; b_2 \in \{0,1\}\\ j = \max\{i \;:\; \text{$i$ is odd}, \X_{1 \to i} = \Y_{1 \to i}\}}} \widetilde{r}_v(g_{\X,b_1},g_{\Y,b_2}) \,. 
\end{align}

Combining \cref{eq:simple_inequality_tilde_r_v_F_1_F_2_sum_over_X_sep} and \cref{eq:simple_ineq_r_tilde_v_F_1_F_2_in_X_part2_sep}, we get

\begin{align} \label{eq:simple_ineq_r_tilde_v_F_1_F_2_in_X_part2_prime_sep}
\sum_{F_2 \in \mathcal{Z}} \widetilde{r}_v(F_1,F_2)
\leq \sum_{j=1}^{2c+1} \sum_{\substack{\Y \in \{1\} \times [m]^{2c}, \; b_2 \in \{0,1\}\\ j = \max\{i \;:\; \text{$i$ is odd}, \X_{1 \to i} = \Y_{1 \to i}\}}} \widetilde{r}_v(g_{\X,b_1},g_{\Y,b_2}) \,.
\end{align}

For each $j \in [2c+1]$, let 
\[ 
T_j = \left|\left\{\Y \in \{1\} \times [m]^{2c} \mid \max\{i : \text{$i$ is odd}, \X_{1 \to i} = \Y_{1 \to i} \} = j \text{ and } v \in Tail(j,S_{\Y}) \right\} \right|\,.
\]
Then for each odd $j \in [2c+1]$ such that $v \not\in N_{x_j}$, we can bound the part of the sum in  \cref{eq:simple_ineq_r_tilde_v_F_1_F_2_in_X_part2_prime_sep} corresponding to index $j$ via the next chain of inequalities:
\begin{align} 
&\; \sum_{\substack{\Y \in \{1\} \times [m]^{2c}, \; b_2 \in \{0,1\} :\\ j = \max\{i \;:\; \text{$i$ is odd}, \X_{1 \to i} = \Y_{1 \to i}\}}} \widetilde{r}_v(g_{\X,b_1},g_{\Y,b_2}) \notag  \\
\leq &\; \sum_{\substack{\Y \in \{1\} \times [m]^{2c} :\\ j = \max\{i \;:\; \X_{1 \to i} = \Y_{1 \to i}\}  \\ v \in Tail(j, S_{\Y}) }} {r}(g_{\X,b_1},g_{\Y,1-b_1}) \explain{By \cref{lem:j_le_ell_implies_v_in_tail_separation} }\\
\leq &\; m^j \cdot |T_j| \explain{Since $r(g_{\X,b_1},g_{\Y,1-b_1}) \leq m^j$ when $j = \max\{i : \text{$i$ is odd}, \X_{1 \to i} = \Y_{1 \to i} \} $} \\
\leq &\; m^j \cdot \left( (2c+1) \cdot m^{2c-j} \right) \explain{By \cref{lem:upper_bound_set_v_in_Tail_Y_sep}}\\
= &\; (2c+1) \cdot m^{2c} \label{eq:sum_for_one_gxb_decomposed_by_j_solved_sep}
\end{align}

Meanwhile, since $\X$ is good, there is at most one odd $j$ such that $v \in N_{x_j}$.
For that index $j$, since $\widetilde{r}_v(g_{\X,b_1}, g_{\Y,b_2}) > 0$ implies  $b_2 = 1 - b_1$ (see observation \ref{obs:r_v_tilde_equal_b}), we have 
\begin{align}
\label{eq:sum_for_one_gxb_decomposed_by_ell_plus_one_solved_sep}
\sum_{\substack{\Y \in \{1\} \times [m]^{2c}, \; b_2 \in \{0,1\} :\\ j = \max\{i \;:\; \text{$i$ is odd}, \X_{1 \to i} = \Y_{1 \to i}\}}} \widetilde{r}_v(g_{\X,b_1},g_{\Y,b_2})
& \leq m^j \cdot | \{\Y \in \{1\} \times [m]^{2c} \mid \X_{1 \to j} = \Y_{1 \to j} \}| \notag \\
& = m^j \cdot m^{2c+1-j} = m^{2c+1}\,.
\end{align}

Summing \cref{eq:sum_for_one_gxb_decomposed_by_j_solved_sep} to  \cref{eq:sum_for_one_gxb_decomposed_by_ell_plus_one_solved_sep}, we can now upper bound the right hand side of \cref{eq:simple_ineq_r_tilde_v_F_1_F_2_in_X_part2_prime_sep} as follows: 
\begin{align}
\sum_{F_2 \in \mathcal{Z}} \widetilde{r}_v(F_1,F_2) &  \leq  \sum_{j=1}^{2c+1} \sum_{\substack{\Y \in \{1\} \times [m]^{2c}, \; b_2 \in \{0,1\}\\ j = \max\{i \;:\; \text{$i$ is odd}, \X_{1 \to i} = \Y_{1 \to i}\}}} \widetilde{r}_v(g_{\X,b_1},g_{\Y,b_2}) \explain{By  \cref{eq:simple_ineq_r_tilde_v_F_1_F_2_in_X_part2_prime_sep}} \\
& \leq (c+1) \cdot (2c+1) \cdot m^{2c} + m^{2c+1} \explain{By \cref{eq:sum_for_one_gxb_decomposed_by_j_solved_sep} and \cref{eq:sum_for_one_gxb_decomposed_by_ell_plus_one_solved_sep}} \\
&   \leq 3 \cdot m^{2c+1}\,.  \explain{Since $c \leq \sqrt{m}-1$}
\end{align}

Thus, for each good function $F_1 \in \mathcal{Z}$, we have  
\begin{align}  \label{eq:good_F_1_sum_over_all_F_2_in_Z_sep}
\sum_{F_2 \in \mathcal{Z}} \widetilde{r}_v(F_1,F_2) \leq  3 \cdot m^{2c+1}\,.
\end{align}

Summing \cref{eq:good_F_1_sum_over_all_F_2_in_Z_sep} over all $F_1 \in \mathcal{Z}$ (each of which is good, since $\mathcal{Z}$ was chosen to have good functions only), and invoking \cref{lem:sum_rv_bounded_by_two_sum_rv_tilde} yields 
\begin{align}
\sum_{F_1,F_2 \in \mathcal{Z}}  r_v(F_1, F_2) & \leq 
 2 \cdot \sum_{F_1,F_2 \in \mathcal{Z}} 
 \widetilde{r}_v(F_1, F_2) \explain{By \cref{lem:sum_rv_bounded_by_two_sum_rv_tilde}}\\
\leq &\; 2 \cdot 3 \cdot m^{2c+1} \explain{By \cref{eq:good_F_1_sum_over_all_F_2_in_Z_sep}}\\
= & \;|\cZ| \cdot 6 m^{2c+1} \,. \label{eq:sum_F_1_F_2_in_cal_Z_r_v_almost_there_for_q_Z_sep}
\end{align}
Since we had considered an arbitrary vertex $v \in [n]$, taking the maximum over all $v \in [n]$ in \cref{eq:sum_F_1_F_2_in_cal_Z_r_v_almost_there_for_q_Z_sep} yields 
\begin{equation}
\label{eq:numerator_bound_sep}
q(\cZ)
= \max_{v \in [n]} \sum_{F_1 \in \cZ} \sum_{F_2 \in \cZ} r_v(F_1,F_2)
\leq |\cZ| \cdot 6 m^{2c+1}\,.
\end{equation}

\paragraph{Lower bounding $M(\cZ)$.}  
    Each function $F_1 \in \mathcal{Z}$ is good by choice of $\mathcal{Z}$.
    Additionally, since $c \ge 1$, \cref{lem:M_large_separation} yields
    \begin{align}  \label{eq:simple_bound_implied_by_Lemma_31_for_F1_good_sep}
    \sum_{F_2 \in \cX} r(F_1,F_2) \geq \frac{1}{2e} \cdot (c+1) \cdot m^{2c+1} \qquad \forall F_1 \in \mathcal{Z} \,.
    \end{align}
    Using \cref{eq:simple_bound_implied_by_Lemma_31_for_F1_good_sep} and recalling the definition of $M(\mathcal{Z})$ from \cref{eq:remind_Z_def_sep}, we get 
    \begin{align} 
M(\mathcal{Z}) & = \sum_{F_1 \in \cZ} \sum_{F_2 \in \cX} r(F_1, F_2) \notag  \\
& \geq \frac{|\cZ| }{2e} \cdot (c+1) \cdot m^{2c+1} \,. \label{eq:denominator_bound_sep}
\end{align} 

    \paragraph{Combining the bounds.}
    Combining the bounds from \cref{eq:numerator_bound_sep} and \cref{eq:denominator_bound_sep}, we can now estimate the bound from \cref{thm:our-variant}:
    \begin{align*}
    \min_{\substack{\cZ \subseteq \cX:\\ q(\cZ) > 0}}
    \frac{M(\cZ)}{q(\cZ)}
    &\geq \frac{ \frac{|\cZ|}{2e} \cdot (c+1)  m^{2c+1}}{|\cZ| \cdot 6m^{2c+1}} \explain{By \cref{eq:numerator_bound_sep} and \cref{eq:denominator_bound_sep}}\\
    &\geq \frac{\sqrt{m}}{24e} \explain{Since $c+1 \geq \sqrt{m}/2$}\\
    &\geq \frac{1}{96} \cdot \left( \frac{s}{\Delta} \right)^{1/4} \explain{By \cref{lem:m_lb}}
    \end{align*}
    
    Therefore, the randomized query complexity of local search is
    \[
    \Omega\left( \min_{\substack{\cZ \subseteq \cX:\\ q(\cZ) > 0}}
    \frac{M(\cZ)}{q(\cZ)} \right)
    \subseteq \Omega\left( \left( \frac{s}{\Delta} \right)^{1/4} \right)\;.
    \]
    This completes the proof of the theorem.
\end{proof}

\subsection{Helper lemmas}
\label{sec:helper_lemmas_separation}

Here we prove the helper lemmas that are used in the proof of \cref{thm:separation-number-result}. All the lemmas assume the setup of the parameters from \cref{thm:separation-number-result}. 


\begin{restatable}{mylemma}{fXisvalidseparation}
\label{lem:f_X_is_valid_separation}
For each $\X \in \{1\} \times [m]^{2c}$, the function $f_{\X}$ is valid for the staircase $S_{\X}$ induced by $\X$ and the cluster paths.
\end{restatable}
\begin{proof}
Let $S_{\X} = (w_1, \ldots, w_s)$ be the vertices of the staircase $S_{\X}$ induced by $\X$ and $\mathcal{P}$.
We show that all the three conditions required by the definition of a valid function (\cref{def:valid_function}) hold.

To show the first condition of validity, consider two vertices $v_1,v_2 \in S_{\X}$.
Let $i_1$  be the maximum index such that $v_1$ is the $i_1$-th vertex in $S_{\X}$.
Let $i_2$ be defined similarly for $v_2$.
By \cref{def:mathcal_X}, we have $f_{\X}(v_1) = -i_1 < 0$ and $f_{\X}(v_1) = -i_2 < 0$.
Furthermore, if $i_1 < i_2$, then $f_{\X}(v_1) > f_{\X}(v_2)$.
Therefore the first condition of validity is satisfied.

Also by \cref{def:mathcal_X}, of the function $f_{\vec{x}}$, we have that:
\begin{itemize}
\item By definition, $S_{\X}$ starts at $v_{start} \in N_1$
\item $f_{\X}(v) = dist(v, v_{start}) > 0$ for all $v \notin S_{\X}$, so the second condition of validity is satisfied.
\item $f_{\X}(v)  \le 0$ for all $v \in S_{\X}$, so the third condition of validity is satisfied.
\end{itemize}
Therefore $f_{\X}$ is valid for the staircase $S_{\X}$ induced by $\X$ and $\cP$.
\end{proof}

\begin{restatable}{mylemma}{parameterstovariantareokseparation}
\label{lem:parameters_to_variant_are_ok_separation}
If $n \ge 9$, then the next two properties hold:
\begin{itemize}
    \item Let $F_1, F_2 \in \mathcal{X}$. Then $r(F_1,F_2) = 0$ when $\mathcal{H}(F_1) = \mathcal{H}(F_2)$. 
    \item There exists a subset $\mathcal{Z} \subseteq \mathcal{X}$ such that
    \[
    q(\mathcal{Z}) = \max_{v \in [n]} \sum_{F_1 \in \mathcal{Z}} \sum_{F_2 \in \mathcal{Z}} r(F_1,F_2) \cdot \mathbbm{1}_{\{F_1(v) \ne F_2(v)\}} > 0 \;.
    \]
\end{itemize}
\end{restatable}
\begin{proof}
We first show that $r(F_1,F_2) = 0$ when $\mathcal{H}(F_1) = \mathcal{H}(F_2)$.
To see this, suppose $\mathcal{H}(F_1) = \mathcal{H}(F_2)$ for some functions $F_1,F_2 \in \mathcal{X}$. Then by definition of the set of functions $\mathcal{X}$, there exist  sequences of vertices $\X,\Y \in \{1\} \times [m]^{2c}$ and bits $b_1, b_2 \in \{0,1\}$ such that $F_1 = g_{\X,b_1}$ and $F_2 = g_{\Y,b_2}$. By definition of $\mathcal{H}$, we have $\mathcal{H}(g_{\X,b_1}) = b_1$ and $\mathcal{H}(g_{\Y,b_2}) = b_2$. Since $\mathcal{H}(F_1) = \mathcal{H}(F_2)$, we have  $b_1 = b_2$. Then  $r(g_{\X,b_1}, g_{\Y,b_2}) = 0$ by definition of $r$, or equivalently, $r(F_1,F_2) = 0$.

Next we show  there is a subset $\mathcal{Z} \subseteq \mathcal{X}$ with $q(\mathcal{Z}) > 0$.
To see this, consider two disjoint sets of vertices 
$U_1, U_2 \subset [m]$ such that $U_1 = \{u_{2}^1, \ldots, u_{2c+1}^1\}$, $U_2= \{u_2^2, \ldots, u_{2c+1}^2\}$,  each vertex $u_j^i$ appears exactly once in $U_i$, and $u_{j}^i \neq 1$ for all $i,j$. 
Such sets $U_1,U_2$ exist since $m \leq n$ and
$$|U_1| + |U_2| + |\{1\}|= 4c + 1
= 4 (\lfloor \sqrt{m} \rfloor - 1) + 1
\leq 4 \sqrt{n} - 3
\leq n
\qquad \text{for $n \geq 9$.}$$
Form the sequences of vertices $\X = (1, u_2^1, \ldots, u_{2c+1}^1)$ and $\Y= (1, u_2^2, \ldots, u_{2c+1}^2)$. 
Then both $\X$ and $\Y$ are good.
Consider now the functions $g_{\X, 0}$ and $g_{\Y, 1}$. By definition of $r$, we have $r(g_{\X, 0}, g_{\Y, 1}) = n$, since the maximum index $j$ for which $\X_{1\to j} = \Y_{1 \to j} $ is $j=1$.
Let $v$ be the last vertex of $S_\X$.
Then 
\[ 
q(\{g_{\X,0}, g_{\Y,1}\}) \ge r(g_{\X, 0}, g_{\Y, 1}) \cdot \mathbbm{1}_{\{g_{\X, 0}(u_{2c+1}^1) \ne g_{\Y, 1}(u_{2c+1}^1)\}} = n > 0 \;.
\]
Thus there exists a subset $\mathcal{Z} \subseteq \mathcal{X}$ with $q(\mathcal{Z}) > 0$
as required.
\end{proof}

\begin{restatable}{mylemma}{rvtildeimpliesseveralthingsseparation}
\label{lem:rv_tilde_implies_several_things_separation}
   Let  $\X,\Y \in \{1\} \times [m]^{2c}$,  $b_1,b_2 \in \{0,1\}$,   $v \in [n]$. Let $j \in [2c+1]$ be the maximum odd index for which $\X_{1 \to j} = \Y_{1 \to j}$.
    Then if $\; \widetilde{r}_v(g_{\X,b_1}, g_{\Y,b_2}) > 0$, then at least one of the next two properties holds:
    \begin{enumerate}[(i)]
        \item $v \in Tail(j, S_{\Y})$. 
        \item $\X = \Y$.
    \end{enumerate}
\end{restatable}
\begin{proof}
We start with a few observations.

Recall from \cref{def:rv_tilde_sep} that 
$\widetilde{r}_v(g_{\X,b_1}, g_{\Y,b_1}) = r_v(g_{\X,b_1}, g_{\Y,b_2}) \cdot \mathbbm{1}_{\{\multiplicity(S_{\X},v) \le \multiplicity(S_{\Y},v)\}}$. By the lemma statement, we have $\widetilde{r}_v(g_{\X,b_1}, g_{\Y,b_2}) > 0$, and so the next two inequalities hold:
    \begin{align} 
    & r_v(g_{\X,b_1}, g_{\Y,b_2}) > 0 \label{eq:r_v_strictly_positive_2_sep}\\
    & \multiplicity(S_{\X},v) \leq \multiplicity(S_{\Y},v)\,. \label{eq:m_SX_v_les_m_SY_v_sep}
    \end{align} 
    Also recall by definition of $r_v$ that $ r_v(g_{\X,b_1}, g_{\Y,b_2}) = r(g_{\X,b_1}, g_{\Y,b_2}) \cdot \mathbbm{1}_{\{g_{\X,b_1}(v) \neq g_{\Y,b_2}(v)\}} \,.$ 
Then \cref{eq:r_v_strictly_positive_2_sep} implies 
 that   
\begin{align} 
& g_{\X,b_1}(v) \neq g_{\Y,b_2}(v) \,.  \label{eq:v_not_equal_gXb1_gXb2_separation}
\end{align}

To prove that $v \in Tail(j, S_{\Y})$ or $\X = \Y$ we consider two cases: 

    \paragraph{Case 1: $v \in Tail(j,S_{\X})$.}
    
    Decomposing the staircase $S_{\X}$ into the initial segment $S_{\X_{1 \to j}}$ and the remainder $Tail(j,S_{\X})$, and similarly the staircase $S_{\Y}$
 into initial segment  $S_{\Y_{1 \to j}}$ and the remainder $Tail(j,S_{\Y})$, we get:
    \begin{align}
         & \multiplicity(S_{\X},v) \le  \multiplicity(S_{\Y},v) \explain{By \cref{eq:m_SX_v_les_m_SY_v_sep}}  \\
        \iff & \multiplicity(S_{\X_{1 \to j}},v) + \multiplicity(Tail(j,S_{\X}),v) \le \multiplicity(S_{\Y_{1 \to j}},v) + \multiplicity(Tail(j,S_{\Y}),v) \notag \\ 
         \iff & \multiplicity(Tail(j,S_{\X}),v) \le \multiplicity(Tail(j,S_{\Y}),v) \explain{Since $\X_{1 \to j} = \Y_{1 \to j}$.}
    \end{align}
    But since $v \in Tail(j,S_{\X})$, we have $\multiplicity(Tail(j,S_{\X}),v) \geq 1$, so
    \begin{align*}
          1 \leq  \multiplicity(Tail(j,S_{\Y}),v)
    \end{align*}
    Thus $v \in Tail(j,S_{\Y})$, so property $(i)$ from the lemma statement holds. This completes Case 1.
    
    \paragraph{Case 2: $v \notin Tail(j,S_{\X})$.}
If $\X=\Y$, then property $(ii)$ from the lemma statement  holds. 

Now suppose  $\X \neq \Y$. Then $Tail(j, S_{\Y}) \neq \emptyset$.
We claim  $v \in S_{\X} \cup S_{\Y}$. 

Suppose towards a contradiction that $v \not \in S_{\X} \cup  S_{\Y}$. For each  $b \in \{0,1\}$,  $u \in [n]$, and sequence $\Z = (1, z_2, z_3 \ldots, z_{2c+1}) \in \{1\} \times [m]^{2c}$, we have by \cref{eq:def:g_X_b} (which defines the function $g_{\vec{z}, b}$) that 
\begin{align} 
g_{\Z,b}(u) =
\begin{cases}
(f_{\Z}(u), b) & \text{$u$ is the last vertex of $S_{\Z}$}\\
(f_{\Z}(u), -1) & \text{otherwise}
\end{cases}
\end{align}  

Since $v \not \in Tail(j,S_{\X})$, $v$ is not the last vertex of $S_{\X}$, and so  $g_{\X,b_1}(v) =(f_{\X}(v), -1)$.
 Moreover, since  $x_1 = y_1 = 1$ and $v \not \in S_{\X} \cup S_{\Y}$, we have $f_{\X}(v) = dist(v,x_1) = dist(v,y_1) = f_\Y(v)\,.$ Combining these observations yields  
$g_{\X,b_1}(v) = (f_{\X}(v), -1)  
= (f_\Y(v),-1)  
 = g_{\Y,b_2}(v),$
which contradicts \cref{eq:v_not_equal_gXb1_gXb2_separation}.
Thus the assumption must have been false and $v \in S_{\X} \cup S_{\Y}$.

To summarize, we have  $\X_{1\to j} = \Y_{1 \to j}$, $\X \neq \Y$, $v \in S_{\vec{x}} \cup S_{\vec{y}}$, and $v \not \in Tail(j,S_{\X})$. Suppose towards a contradiction that $v \not \in Tail(j,S_{\Y})$. Then 
\begin{align}
g_{\X,b_1}(v) & =(f_{\X}(v), -1) \explain{Since $v \neq x_{L+1}$, as $v \not \in Tail(j, S_{\X})$.} \\
& = (f_\Y(v), -1) \explain{Since $v \not \in Tail(j,S_{\X})$ and $v \not \in Tail(j,S_{\Y})$.} \\
& = g_{\Y,b_2}(v) \explain{Since $v \neq y_{L+1}$, as $v \not \in Tail(j, S_{\Y})$.},
 \end{align}
 which contradicts \cref{eq:v_not_equal_gXb1_gXb2_separation}.
Thus the assumption must have been false and $v  \in Tail(j,S_{\Y})$, so property $(i)$ from the lemma statement holds. 

We conclude that at least one of properties $(i)$ and $(ii)$ holds. This completes Case 2, as well as the proof of the lemma.
\end{proof}

\begin{restatable}{mylemma}{jleellimpliesvintailseparation}
\label{lem:j_le_ell_implies_v_in_tail_separation}
For each $\X \in \{1\} \times [m]^{2c}$, $b_1 \in \{0,1\}$, odd $j \in [2c+1]$, and $v \in [n]$ such that $v \notin N_{x_j}$, we have 
\[
\sum_{\substack{{\Y} \in \{1\} \times [m]^{2c}, \; b_2 \in \{0,1\} :\\ \max\{i : \text{$i$ is odd}, \X_{1 \to i} = \Y_{1 \to i}\} = j }} \widetilde{r}_v(g_{\X,b_1},g_{\Y,b_2})  \le
\sum_{\substack{{\Y} \in \{1\} \times [m]^{2c} :\\ \max\{i : \text{$i$ is odd}, \X_{1 \to i} = \Y_{1 \to i}\} = j \\ v \in Tail(j, S_{\Y}) }} {r}(g_{\X,b_1},g_{\Y,1-b_1})\,.
\]
\end{restatable}
\begin{proof}
Let $\X \in \{1\} \times [m]^{2c}$, $b_1 \in \{0,1\}$, odd $j \in [2c+1]$, and $v \in [n]$ such that $v \notin N_{x_j}$.
Using the definitions of $r$, $r_v$, and $\widetilde{r}_v$, we have the following chain of identities:
\begin{small}
\begin{align} 
        & \sum_{\substack{{\Y} \in \{1\} \times [m]^{2c}, \; b_2 \in \{0,1\} :\\ \max\{i : \text{$i$ is odd},\X_{1 \to i} = \Y_{1 \to i}\} = j }} \widetilde{r}_v(g_{\X,b_1},g_{\Y,b_2}) 
          =  \sum_{\substack{{\Y} \in  \{1\} \times [m]^{2c}, \; b_2 \in \{0,1\} :\\ \max\{i : \text{$i$ is odd},\X_{1 \to i} = \Y_{1 \to i}\} = j \\ \widetilde{r}_v(g_{\X,b_1},g_{\Y,b_2}) > 0} } \widetilde{r}_v(g_{\X,b_1},g_{\Y,b_2})  \explain{Since $\widetilde{r}_v$ is non-negative.} \\
         & \qquad \qquad \qquad = \sum_{\substack{{\Y} \in  \{1\} \times [m]^{2c} :\\ \max\{i : \text{$i$ is odd},\X_{1 \to i} = \Y_{1 \to i}\} = j \\ \widetilde{r}_v(g_{\X,b_1},g_{\Y,1-b_1}) > 0} } \widetilde{r}_v(g_{\X,b_1},g_{\Y,1-b_1})  \explain{Since $\widetilde{r}_v(g_{\X,b_1}, g_{\Y,b_1}) = 0$.} \\
         & \qquad \qquad \qquad = \sum_{\substack{{\Y} \in  \{1\} \times [m]^{2c} :\\ \max\{i : \text{$i$ is odd},\X_{1 \to i} = \Y_{1 \to i}\} = j \\ \multiplicity(S_{\X}, v) \leq \multiplicity(S_{\Y}, v)  \\\widetilde{r}_v(g_{\X,b_1},g_{\Y,1-b_1})  > 0 }}{r}_v(g_{\X,b_1},g_{\Y,1-b_1}) \explain{By definition of $\widetilde{r}_v$.} \\
         & \qquad \qquad \qquad = \sum_{\substack{{\Y} \in  \{1\} \times [m]^{2c} :\\ \max\{i : \text{$i$ is odd},\X_{1 \to i} = \Y_{1 \to i}\} = j \\ \multiplicity(S_{\X}, v) \leq \multiplicity(S_{\Y}, v)  \\\widetilde{r}_v(g_{\X,b_1},g_{\Y,1-b_1})  > 0 }} {r}(g_{\X,b_1},g_{\Y,1-b_1}) \cdot \mathbbm{1}_{\{g_{\X,b_1}(v) \neq g_{\Y,1-b_1}(v) \}} \explain{By definition of $r_v$.} \\
          & \qquad \qquad \qquad = \sum_{\substack{{\Y} \in  \{1\} \times [m]^{2c} :\\ \max\{i : \text{$i$ is odd},\X_{1 \to i} = \Y_{1 \to i}\} = j \\ \multiplicity(S_{\X}, v) \leq \multiplicity(S_{\Y}, v) \\g_{\X,b_1}(v) \neq g_{\Y,1-b_1}(v) \\ \widetilde{r}_v(g_{\X,b_1},g_{\Y,1-b_1})  > 0 }} {r}(g_{\X,b_1},g_{\Y,1-b_1}) 
          \label{eq:sum_for_one_gxb_decomposed_by_j_partial_separation}
\end{align}
\end{small}

Consider an arbitrary $\Y \in  \{1\} \times [m]^{2c}$ meeting the properties from the last sum of \cref{eq:sum_for_one_gxb_decomposed_by_j_partial_separation}: 
\begin{itemize}
\item $\max\{i : \text{$i$ is odd},\X_{1 \to i} = \Y_{1 \to i}\} = j$
\item $ \multiplicity(S_{\X}, v) \leq \multiplicity(S_{\Y}, v) $
\item $g_{\X,b_1}(v) \neq g_{\Y,1-b_1}(v)$
\item $ \widetilde{r}_v(g_{\X,b_1},g_{\Y,1-b_1})  > 0$
\end{itemize}
By \cref{lem:rv_tilde_implies_several_things_separation}, 
the inequality $\widetilde{r}_v(g_{\X,b_1},g_{\Y,1-b_1})>0$ implies that 
 at least one of  $v \in Tail(j, S_{\Y})$ or $\X=\Y$  holds.

 However, $\X$ cannot be equal to such $\Y$. To see this, suppose for sake of contradiction that $\X = \Y$. Then since $g_{\X,b_1}(v) \neq g_{\Y,1-b_1}(v)$ we would have that $v$ must be the last vertex of $S_\Y$. Therefore $v \in N_{y_{2c+1}}$.
 Also, we would have $j = 2c+1$ since $\X = \Y$.
 But then $v \in N_j$, which is a contradiction with the assumption $v \notin N_j$.
 Therefore $\X \ne \Y$, and so $v \in Tail(j,S_{\Y})$.
 
 We can continue to bound the sum from \cref{eq:sum_for_one_gxb_decomposed_by_j_partial_separation} as follows:
\begin{align} 
 \sum_{\substack{{\Y} \in  \{1\} \times [m]^{2c} :\\ \max\{i : \text{$i$ is odd},\X_{1 \to i} = \Y_{1 \to i}\} = j \\ \multiplicity(S_{\X}, v) \leq \multiplicity(S_{\Y}, v) \\g_{\X,b_1}(v) \neq g_{\Y,1-b_1}(v)  \\ \widetilde{r}_v(g_{\X,b_1},g_{\Y,1-b_1})  > 0 }} {r}(g_{\X,b_1},g_{\Y,1-b_1})  \leq \sum_{\substack{{\Y} \in  \{1\} \times [m]^{2c} :\\ \max\{i : \text{$i$ is odd},\X_{1 \to i} = \Y_{1 \to i}\} = j \\ v \in Tail(j, S_{\Y}) }} {r}(g_{\X,b_1},g_{\Y,1-b_1})\,.  \label{eq:sum_for_one_gxb_decomposed_by_j_partial_step2_separation} 
 \end{align}
 Combining \cref{eq:sum_for_one_gxb_decomposed_by_j_partial_separation} and \cref{eq:sum_for_one_gxb_decomposed_by_j_partial_step2_separation}, we get the inequality required by the lemma statement:
 \begin{align} 
 \sum_{\substack{{\Y} \in \{1\} \times [m]^{2c}, \; b_2 \in \{0,1\} :\\ \max\{i : \text{$i$ is odd},\X_{1 \to i} = \Y_{1 \to i}\} = j }} \widetilde{r}_v(g_{\X,b_1},g_{\Y,b_2}) \leq   \sum_{\substack{{\Y} \in  \{1\} \times [m]^{2c} :\\ \max\{i : \text{$i$ is odd},\X_{1 \to i} = \Y_{1 \to i}\} = j \\ v \in Tail(j, S_{\Y}) }} {r}(g_{\X,b_1},g_{\Y,1-b_1})\,.  \notag 
 \end{align} 
\end{proof}

\begin{restatable}{mylemma}{upperboundsetvinTailYseparation}
\label{lem:upper_bound_set_v_in_Tail_Y_sep}
Let $\X \in \{1\} \times [m]^{2c}$, odd $j \in [2c+1]$, and $v \in [n]$ such that $v \notin N_{x_i}$.
Then
\begin{small}
\begin{align} 
\Bigl| \Bigl\{ \Y \in \{1\} \times [m]^{2c} \mid \max\{i : \text{$i$ is odd},\X_{1 \to i} = \Y_{1 \to i} \} = j \text{ and } v \in Tail(j,S_{\Y}) \Bigr\}  \Bigr|    \leq  (2c+1) \cdot m^{2c-j}  \,. \notag 
\end{align} 
\end{small}
\end{restatable}
\begin{proof}
We will now do case analysis on different conditions on $i$ and $j$.
    \begin{enumerate}[(i)] 
    \item For odd $i \in [2c]$ with $i \geq j$, consider the number of $\Y \in \{1\} \times [m]^{2c}$ such that
\begin{equation}
\label{eq:Y_prop_related_to_X_1}
\X_{1 \to j} = \Y_{1 \to j}
\quad \text{and} \quad
v \in S_{\Y,i} \;.
\end{equation}
Since $i$ is odd, $S_{\Y,i}$ is a path within a cluster, and thus can only contain $v$ if $N_{y_i}$ is the one cluster containing $v$.
If $i = j$, then we know by assumption that $N_{y_i}$ does \emph{not} contain $v$.
Otherwise, there is only one choice for $y_i$ while there are $m$ choices for each of the rest of $y_{j+1}, y_{j+2}, \ldots, y_{2c+1}$.
Thus, there are at most $m^{2c-j}$ choices of $\Y$ for which \cref{eq:Y_prop_related_to_X_1} holds.

    \item For even $i \in [2c]$ with $i \geq j$, consider the number of $\Y \in \{1\} \times [m]^{2c}$ such that
\begin{equation}
\label{eq:Y_prop_related_to_X_2}
\X_{1 \to j} = \Y_{1 \to j},
v \in S_{\Y,i}
\text{, and}
\;v \notin S_{\Y,k}\;
\text{for all odd $k \geq j$ and $v$ is not the last vertex of $S_{\Y}$} \;.  
\end{equation}
Since $i$ is even, $S_{\Y,i}$ is an inter-cluster path.
Recall the $m$ inter-cluster paths between any two clusters are disjoint except for their start and endpoints.
Meanwhile $v$ cannot be their start or endpoints since otherwise $v \in S_{\Y,k}$ for an odd $k \geq j$ or $v$ would be the last vertex of $S_{\Y}$.
Therefore, $S_{\Y,i}$ can only contain $v$ for at most one value of $y_i$ while there are $m$ choices for each of the rest of $y_{j+1}, y_{j+2}, \ldots, y_{2c+1}$.
Thus, there are at most $m^{2c-j}$ choices of $\Y$ for which \cref{eq:Y_prop_related_to_X_2} holds.

    \item Finally, consider the number of $\Y \in \{1\} \times [m]^{2c}$ such that
\begin{equation}
\label{eq:Y_prop_related_to_X_3}
\X_{1 \to j} = \Y_{1 \to j}
\quad \text{and} \quad
\text{$v$ is the last vertex of $S_{\Y}$} \;.  
\end{equation}
This cannot occur unless $v \in N_{y_{2c+1}}$.
Therefore, there are at most $m^{2c-j}$ such $\Y$.

For $v$ to be in $Tail(j,S_{\Y})$, we must have $\Y$ satisfying one of \cref{eq:Y_prop_related_to_X_1}, \cref{eq:Y_prop_related_to_X_2}, or \cref{eq:Y_prop_related_to_X_3} for some value of $i$.
Summing over choices of $i$ gives us that
\begin{align}
& \Bigl| \left\{ \Y \in \{1\} \times [m]^{2c} \mid j = \max\{i \;:\; \text{$i$ is odd},\; \X_{1 \to i} = \Y_{1 \to i}\} \text{ and } v \in Tail(j, S_{\Y}) \right\} \Bigr| \notag \\
& \qquad \leq (2c+1) \cdot m^{2c-j} \,.  \notag  
\end{align}
\end{enumerate}

This completes the proof of the lemma.
\end{proof}

\begin{restatable}{mylemma}{countYdenominatorseparation}
\label{lem:count_Y_denominator_separation}
Let $j \in [2c]$ and $\X \in \{1\} \times [m]^{2c}$ be an arbitrary good sequence of indices. 
   Then 
    \[
     \Bigl| \Y \in \{1\} \times [m]^{2c} \mid \Y \; \text{is good and } j = \max\{i:\text{$i$ is odd},\X_{1 \to i} = \Y_{1 \to i}\} \Bigr|  =  (m-j-1) \cdot \prod_{i=j+2}^{2c+1} (m-i+1) \,.
    \]
\end{restatable}
\begin{proof}
Let $\Y = (y_1, \ldots, y_{2c+1}) \in [m]^{2c+1}$ be such that $x_1 = y_1 = 1$,  $j = \max\{i:\text{$i$ is odd},\X_{1 \to i} = \Y_{1 \to i}\}$, and $\vec{y}$ is good..
Recall that for $\Y$ to be good, each index in $(y_1, \ldots, y_{2c+1})$ has to be distinct.

    We will count by choosing each index of $\Y$ in order from $y_1$ to $y_{2c+1}$.
    There is only one choice for each of $y_1, \ldots, y_j$ since we need $\X_{1 \to j} = \Y_{1 \to j}$.
    
    Consider the number of possible indices for $y_{j+1}$:
    \begin{description}
        \item[$\bullet$] $j$ indices already appeared in $\{y_1, \ldots, y_j\}$ so they cannot be reused, otherwise $\Y$ becomes bad.
        \item[$\bullet$] $j = \max\{i:\text{$i$ is odd},\X_{1 \to i} = \Y_{1 \to i}\}$ means $x_{j+1} \neq y_{j+1}$, so $x_{j+1}$ cannot be used either.
    \end{description}
    Thus, there are $m-j-1$ choices of indices for $y_{j+1}$.
    
    For all $i > j+1$, there are $m-i+1$ choices of indices for $y_i$ since $i-1$ options have already been used.
    Thus the final count is
$(m-j-1) \cdot \prod_{i=j+2}^{2c+1} (m-i+1)$, as required.
\end{proof}

\begin{restatable}{mylemma}{Mlargeseparation}
\label{lem:M_large_separation}
If $F_1 \in \mathcal{X}$ is good and $c \ge 1$, then 
$ \; 
\sum_{F_2 \in \cX} r(F_1,F_2) \geq \frac{1}{2e} \cdot (c+1) \cdot m^{2c+1}\,.
$
\end{restatable}
\begin{proof}
Since $F_1 \in \cX$, there exists a sequence  
 $\X \in \{1\} \times [m]^{2c}$ and a bit  $b_1 \in \{0,1\}$ such that $F_1 = g_{\vec{x}, b_1}$. Then we have the following chain of identities:
\begin{align}
\sum_{F_2 \in \cX} r(F_1, F_2) & =
\sum_{F_2 \in \cX} r(g_{\X,b_1}, F_2) \notag \\
&= \sum_{\substack{\Y \in \{1\} \times [m]^{2c},\; b_2 \in \{0,1\}}} r(g_{\X,b_1}, g_{\Y,b_2}) \explain{By definition of the set $\mathcal{X}$}\\
&= \sum_{\substack{\Y \in \{1\} \times [m]^{2c}}} r(g_{\X,b_1}, g_{\Y,1-b_1}) \explain{Since $r(g_{\X,b_1}, g_{\Y,b_1}) = 0$}\\
&= \sum_{j=1}^{2c+1} \sum_{\substack{\Y \in \{1\} \times [m]^{2c}\\ j = \max\{i \;:\; \text{$i$ is odd}, \X_{1 \to i} = \Y_{1 \to i}\}}} r(g_{\X,b_1}, g_{\Y,1-b_1}) \,.
 \label{eq:denominator_M_to_sum_nj_0_sep} 
 \end{align}
Since $r(g_{\X,b_1}, g_{\Y,1-b_1}) = 0$ if $\vec{y}$ is bad, we have 
 \begin{align}
\sum_{j=1}^{2c+1} \sum_{\substack{\Y \in \{1\} \times [m]^{2c}:\\ j = \max\{i \;:\; \text{$i$ is odd}, \X_{1 \to i} = \Y_{1 \to i}\}}} r(g_{\X,b_1}, g_{\Y,1-b_1}) &= \sum_{j=1}^{2c+1} \sum_{\substack{\Y \in \{1\} \times [m]^{2c}:\\ j = \max\{i \;:\; \text{$i$ is odd}, \X_{1 \to i} = \Y_{1 \to i}\}\\ \Y \text{ is good } }} r(g_{\X,b_1}, g_{\Y,1-b_1})  \,. \label{eq:denominator_M_to_sum_nj_1_sep}
\end{align}
Combining \cref{eq:denominator_M_to_sum_nj_0_sep}  and \cref{eq:denominator_M_to_sum_nj_1_sep}, the last sum in  \cref{eq:denominator_M_to_sum_nj_1_sep} can be written as:
\begin{align}
%
\sum_{F_2 \in \cX} r(g_{\X,b_1}, F_2)= \sum_{j=1}^{2c+1} \sum_{\substack{\Y \in \{1\} \times [m]^{2c}:\\ j = \max\{i \;:\; \text{$i$ is odd}, \X_{1 \to i} = \Y_{1 \to i}\}\\ \Y \text{ is good } }} r(g_{\X,b_1}, g_{\Y,1-b_1})\,.  \label{eq:denominator_M_to_sum_nj_2_sep}
\end{align}

Substituting the definition of $r$ in \cref{eq:denominator_M_to_sum_nj_2_sep}, we get 
\begin{equation}
\label{eq:denominator_M_to_sum_nj_sep}
\sum_{F_2 \in \cX} r(g_{\X,b_1}, F_2)
= \sum_{j=1}^{2c+1} \sum_{\substack{\Y \in \{1\} \times [m]^{2c}:\\ j = \max\{i \;:\; \text{$i$ is odd}, \X_{1 \to i} = \Y_{1 \to i}\}\\ \Y \text{ is good } }} m^j\,.
\end{equation}

From here, we only need to count the number of such $\Y$.
There is exactly one good sequence $\Y \in \{1\} \times [m]^{2c}$ such that $\X_{1 \to 2c+1} = \Y_{1 \to 2c+1}$, namely $\Y = \X$.
Therefore,
\begin{align*}
\sum_{F_2 \in \cX} r(g_{\X,b_1}, F_2)
= &\; \sum_{j=1}^{2c+1} \sum_{\substack{\Y \in \{1\} \times [m]^{2c}:\\ \max\{i \;:\; \text{$i$ is odd}, \X_{1 \to i} = \Y_{1 \to i}\} = j\\ \Y \text{ is good } }} m^j \explain{By \cref{eq:denominator_M_to_sum_nj_sep}}\\
= &\; m^{2c+1} + \sum_{j=1}^{2c} \sum_{\substack{\Y \in \{1\} \times [m]^{2c}:\\ \max\{i \;:\; \text{$i$ is odd}, \X_{1 \to i} = \Y_{1 \to i}\} = j\\ \Y \text{ is good } }} m^j\\
= &\; m^{2c+1} + \sum_{j=1}^{2c} m^j \cdot \mathbbm{1}_{\{\text{$j$ is odd}\}} \cdot (m-j-1) \cdot \prod_{i=j+2}^{2c+1} (m-i+1) \explain{By \cref{lem:count_Y_denominator_separation}}\\
\geq &\; (c+1) \cdot m^{2c+1} \cdot \left( 1-\frac{2c+1}{m} \right)^{2c+1}\\
\geq &\; (c+1) \cdot m^{2c+1} \cdot \left( 1-\frac{1}{2c+1} \right)^{2c+1} \explain{Since $2c+1 \leq \sqrt{m}$.}\\
\geq &\; \frac{1}{2e} \cdot (c+1) \cdot m^{2c+1} \explain{Since $2c+1 \geq 3$, which is more than sufficient. }
\end{align*}
Since $F_1 = g_{\vec{x},b_1}$, this is the required inequality, which completes the proof.
\end{proof}

\end{document}